\documentclass{article}

\usepackage{amsmath,amssymb,amsthm}
\usepackage{times}
\usepackage[scaled=0.9]{newtxtt} 
\usepackage{color}
\usepackage{xspace}
\usepackage[only bigsqcap]{stmaryrd}
\usepackage{pifont}
\usepackage{url}
\usepackage{hyperref}
\usepackage{booktabs}
\usepackage[defblank]{paralist}
\usepackage{listings}
\usepackage{booktabs}
\usepackage[table, dvipsnames]{xcolor}

\usepackage{tikz}

\tikzset{nd/.style={rectangle,draw,fill=gray!20,rounded corners=1mm,inner ysep=1.5pt,inner xsep=4pt}}

\lstdefinelanguage{SQLT}{
  language     = SQL,
  morekeywords = {IS},
}

\newcommand{\triple}[3]{#1 \ \ #2 \ \ #3}

\newtheorem{theorem}{Theorem}
\newtheorem{proposition}[theorem]{Proposition}
\newtheorem{corollary}[theorem]{Corollary}

\theoremstyle{definition}
\newtheorem{example}[theorem]{Example}

\newcommand{\owlql}{OWL\,2\,QL\xspace}
\newcommand{\myIRI}{\ensuremath{\textsf{I}}}
\newcommand{\myBNK}{\ensuremath{\textsf{B}}}
\newcommand{\myLIT}{\ensuremath{\textsf{L}}}
\newcommand{\myVAR}{\ensuremath{\textsf{V}}}
\newcommand{\sDOM}{\textsf{C}}

\newcommand{\join}{\textsc{Join}}
\newcommand{\project}{\textsc{Proj}}
\newcommand{\leftjoin}{\textsc{Opt}}
\newcommand{\minus}{\textsc{Minus}}
\newcommand{\union}{\textsc{Union}}
\newcommand{\filter}{\textsc{Filter}}
\newcommand{\bind}{\textsc{Bind}}

\newcommand{\distinct}{\textsc{Dist}}

\newcommand{\var}{\textit{var}}
\newcommand{\dom}{\textit{dom}}
\newcommand{\tuple}{t} 
\newcommand{\sANS}[2]{\llbracket #1\rrbracket_{#2}}

\newcommand{\Null}{\textit{null}\xspace}
\newcommand{\isNull}{\textit{isNull}}

\newcommand{\smerge}{\oplus}

\newcommand{\coalesce}{\textit{coalesce}}
\newcommand{\compatible}{\textit{comp}}
\newcommand{\nullify}{\textit{nullify}}

\newcommand{\tra}{\boldsymbol{\tau}}
\newcommand{\extV}[1][V]{\textit{ext}_{#1}}

\def\ojoin{\setbox0=\hbox{$\Join$}%
  \rule[0ex]{.25em}{.4pt}\llap{\rule[0.8\ht0]{.25em}{.4pt}}}

\newcommand{\LJoin}{\mathbin{\ojoin\mkern-7.6mu\Join}}

\newcommand{\urione}{\textsc{iri}_1}
\newcommand{\uritwo}{\textsc{iri}_2}

\newcommand{\notnull}[1]{\neg\isNull(#1)}
\newcommand{\proj}[1]{\pi_{#1}}

\begin{document}

\title{Efficient Handling of SPARQL OPTIONAL for OBDA (Extended Version)\thanks{This is an extended version of the ISWC 2018 paper.}}

\author{G. Xiao$^1$ \and R. Kontchakov$^2$ \and
  B. Cogrel$^1$ \and D. Calvanese$^1$ \and E.
  Botoeva$^1$}
  
\date{\normalsize$^1$KRDB Research Centre, Free University of Bozen-Bolzano, Italy\\
  \textit{lastname}\url{@inf.unibz.it}\\
  $^2$Dept.\  of Computer Science and Inf. Sys., Birkbeck, University of London, UK\\
  \url{roman@dcs.bbk.ac.uk}
}

\maketitle

\begin{abstract}
  \texttt{OPTIONAL} is a key feature in SPARQL for dealing with
  missing information. While this operator is used extensively,
  it is also known for its complexity, which can make 
  efficient evaluation of queries with \texttt{OPTIONAL} challenging.
  We tackle this problem
  in the Ontology-Based Data Access (OBDA) setting, where the
  data is stored in a SQL relational database and exposed as a
  virtual RDF graph by means of an R2RML mapping.  We start with a
  succinct translation of a SPARQL fragment into SQL. It fully respects bag
  semantics and three-valued logic and
  relies on the extensive use of the \texttt{LEFT} \texttt{JOIN}
  operator and \texttt{COALESCE} function.  We then propose
  optimisation techniques for reducing the size and improving the
  structure of generated SQL queries.  Our optimisations capture
  interactions between \texttt{JOIN}, \texttt{LEFT} \texttt{JOIN},
  \texttt{COALESCE} and integrity constraints such as
  attribute nullability, uniqueness and foreign key
  constraints.  Finally, we empirically verify effectiveness of our
  techniques on the BSBM OBDA benchmark.
\end{abstract}


\section{Introduction}

Ontology-Based Data Access (OBDA) aims at easing the access to database content
by bridging the semantic gap between \emph{information needs} (what users
want to know) and their formulation as executable queries (typically in SQL).
This approach hides the complexity of the database structure from users by
providing them with a high-level representation of the data as an RDF graph.
The RDF graph can be regarded as a view over the database defined by a
DB-to-RDF mapping (e.g., following the R2RML specification) and enriched by
means of an ontology~\cite{CCKK*17}.  Users can then formulate their
information needs directly as high-level SPARQL queries over the RDF graph.
We focus on the standard OBDA setting, where the RDF graph is not
materialised (and is called a \emph{virtual RDF graph}), and the database is
relational and supports SQL~\cite{PLCD*08}.

To answer a SPARQL query, an OBDA system reformulates it into a SQL query,
to be evaluated by the DBMS.  In theory, such a SQL query can be obtained by
\begin{inparaenum}[\itshape (1)]
\item translating the SPARQL query into a relational algebra expression
  over the ternary relation \textit{triple} of the RDF graph,
  and then
\item replacing the occurrences of \textit{triple} by the matching definitions
  in the mapping; the latter step is called \emph{unfolding}.
\end{inparaenum}
We note that, in general, step~(1) also includes rewriting 
the user query with respect to the given (\owlql) ontology~\cite{CDLLR07,KRRXZ14}; we, however,
assume that the query is already rewritten and, for efficiency reasons, 
the mapping is saturated; for details,
see~\cite{KRRXZ14,SeAM14}.

SPARQL joins are naturally translated into (\texttt{INNER}) \texttt{JOIN}s in
SQL~\cite{Cyga05}.  However, in contrast to expert-written SQL queries,
there typically is a \emph{high
 margin for optimisation} in naively translated and unfolded queries.  Indeed, since SPARQL, unlike SQL, is
based on a single ternary relation,
queries usually contain many more joins than SQL queries for the same information need;
this suggests that many of
the \texttt{JOIN}s in unfolded queries are redundant and could be eliminated.
In fact, the semantic query optimisation techniques such as self-join
elimination~\cite{ChGM90} can reduce the number of \texttt{INNER
 JOIN}s~\cite{RoKZ13,PrCS14}.

We are interested in SPARQL queries containing the \texttt{OPTIONAL} operator
introduced to deal with \emph{missing} information, thus serving a similar
purpose~\cite{Cyga05} to the \texttt{LEFT} (\texttt{OUTER}) \texttt{JOIN} operator in relational
databases.
The graph pattern $P_1 \mathrel{\texttt{OPTIONAL}} P_2$ returns answers to
$P_1$ extended (if possible) by answers to $P_2$; when an answer to $P_1$ has
no match in $P_2$ (due to incompatible variable assignments), the variables
that occur only in $P_2$ remain \emph{unbound} (\texttt{LEFT} \texttt{JOIN}
extends a tuple without a match with \texttt{NULL}s).
The focus of this work is the efficient handling of queries with
\texttt{OPTIONAL} in the OBDA setting.  This problem is important in practice
because
\begin{inparaenum}[\itshape (a)]
\item \texttt{OPTIONAL} is very frequent in real SPARQL
  queries~\cite{PiVa11,AFMD11};
\item it is a source of computational complexity: query
  evaluation is \textsc{PSpace}-hard for the fragment with
  \texttt{OPTIONAL} alone~\cite{ScML10} (in contrast, e.g., to basic graph
  patterns with filters and projection, which are \textsc{NP}-complete);
\item unlike expert-written SQL queries, the SQL translations of SPARQL queries
  (e.g.,~\cite{ChLF09}) tend to have more \texttt{LEFT}  \texttt{JOIN}s with more
  complex structure, which DBMSs may fail to optimise well.
\end{inparaenum}
We now illustrate the difference in the structure with  an example.

\begin{example}
  \label{ex:email-pref-simple}
  Let \texttt{people} be a database relation composed of a primary key
  attribute \texttt{id}, a non-nullable attribute \texttt{fullName} and two
  nullable attributes, \texttt{workEmail} and \texttt{homeEmail}:\\[6pt]
  \centerline{\small\renewcommand{\arraystretch}{0.9}\renewcommand{\tabcolsep}{12pt}
    \begin{tabular}{crrr}
      \toprule
      \normalsize\underline{\texttt{id}}\hfill\mbox{} & \normalsize\texttt{fullName}\hfill\mbox{}
      & \normalsize\texttt{workEmail}\hfill\mbox{} 
      & \normalsize\texttt{homeEmail}\hfill\mbox{} \\
      \midrule
      1 & Peter Smith & peter@company.com & peter@perso.org \\
      2 & John Lang & {\small\texttt{NULL}} & joe@perso.org \\
      3 & Susan Mayer & susan@company.com &  \texttt{NULL} \\
      \bottomrule
    \end{tabular}}\\[6pt]
   Consider an information need to retrieve the names of people and their e-mail
   addresses if they are available, with the \emph{preference} given to work
   over personal e-mails. In standard SQL, the IT expert can express such
   a preference by means of the \texttt{COALESCE} function: e.g.,
   \texttt{COALESCE($v_1, v_2$)} returns $v_1$ if it is not \texttt{NULL} and
   $v_2$ otherwise.  The following SQL query retrieves the required names and
   e-mail addresses:
\begin{lstlisting}[language=SQLT]
SELECT fullName, COALESCE(workEmail, homeEmail) FROM people.
\end{lstlisting}
   The same information need could naturally be expressed in SPARQL:
\begin{lstlisting}[language=SPARQL]
SELECT ?n ?e { ?p :name ?n  OPTIONAL { ?p :workEmail ?e }
                            OPTIONAL { ?p :personalEmail ?e } }.
\end{lstlisting}
Intuitively, for each person \texttt{?p}, after evaluating the first
\texttt{OPTIONAL} operator, variable~\texttt{?e} is bound to the work e-mail if
possible, and left unbound otherwise.  In the former case, the second
\texttt{OPTIONAL} cannot extend the solution mapping further because all its
variables are already bound; in the latter case, the second \texttt{OPTIONAL}
tries to bind a personal e-mail to~\texttt{?e}.  See~\cite{Cyga05} for a discussion on a
similar query, which is weakly
well-designed~\cite{DBLP:conf/icdt/KaminskiK16}.

One can see that the two queries are in fact equivalent: the SQL query gives 
the same answers on
the \texttt{people} relation  as the SPARQL query on the RDF
graph that encodes the relation by using \texttt{id} to generate IRIs and
populating data properties \texttt{:name}, \texttt{:workEmail} and
\texttt{:personalEmail} by the non-\texttt{NULL} values of the respective
attributes.

However, the unfolding of the translation of the  SPARQL query  above
would produce two \texttt{LEFT} \texttt{OUTER} \texttt{JOIN}s, even with known simplifications 
(see, e.g., $Q_2$ in~\cite{ChLF09}):
\begin{lstlisting}[language=SQLT]
SELECT v3.fullName AS n, COALESCE(v3.workEmail,v4.homeEmail) AS e
FROM (SELECT v1.fullName, v1.id, v2.workEmail FROM people v1
LEFT JOIN people v2 ON v1.id=v2.id AND v2.workEmail IS NOT NULL) v3
LEFT JOIN people v4 ON v3.id=v4.id AND v4.homeEmail IS NOT NULL
          AND (v3.workEmail=v4.homeEmail OR v3.workEmail IS NULL),
\end{lstlisting}
which is unnecessarily complex (compared to the expert-written SQL query above). 
Observe that the last bracket is an example of a  \emph{compatibility filter} encoding compatibility of SPARQL solution mappings in SQL: it contains disjunction and \texttt{IS NULL}.
\qed
\end{example}

Example~\ref{ex:email-pref-simple} shows that SQL translations with
\texttt{LEFT} \texttt{JOIN}s can be simplified drastically.  In fact, the
problem of optimising \texttt{LEFT} \texttt{JOIN}s has been
investigated both in relational databases~\cite{GaRo97,RaoPZ04} and RDF
triplestores~\cite{ChLF09,atre15}. In the
database setting, \emph{reordering} of \texttt{OUTER} \texttt{JOIN}s has
been studied extensively because it is essential for efficient
query plans, but also challenging  as these operators are neither commutative
nor associative (unlike \texttt{INNER} \texttt{JOIN}s).
To perform a reordering, query planners typically rely on simple joining
conditions, in particular, on conditions that reject
\texttt{NULL}s and do not use \texttt{COALESCE}~\cite{GaRo97}.
However, the SPARQL-to-SQL translation produces
precisely the opposite of what database query planners expect:
\texttt{LEFT} \texttt{JOIN}s with complex compatibility filters.
On the other hand, Chebotko~\emph{et~al.}~\cite{ChLF09}
proposed some simplifications when
an RDBMS stores the
\textit{triple} relation and acts as an RDF triplestore.
Although these simplifications are
undoubtedly useful in the OBDA setting, the presence of mappings
brings additional challenges and, more importantly, significant opportunities.

\begin{example}
  \label{ex:opt-work-email}
  Consider Example~\ref{ex:email-pref-simple} again and suppose we now want to retrieve
  people's names, and when available also their work e-mail addresses.  We can
  naturally represent this information need in SPARQL:
\begin{lstlisting}[language=SPARQL]
SELECT ?n ?e { ?p :name ?n OPTIONAL { ?p :workEmail ?e } }.
\end{lstlisting}
  We can also express it very simply in SQL:
\begin{lstlisting}[language=SQLT]
SELECT fullName, workEmail FROM people.
\end{lstlisting}
  Instead, the straightforward translation and unfolding of the SPARQL query produces
\begin{lstlisting}[language=SQLT]
SELECT v1.fullName AS n, v2.workEmail AS e
FROM people v1 LEFT JOIN people v2 ON v1.id=v2.id AND
                                         v2.workEmail IS NOT NULL.
\end{lstlisting}
R2RML mappings filter out \texttt{NULL} values from the
database because \texttt{NULL}s cannot appear in RDF triples.
Hence, the join condition in the unfolded query contains an
\texttt{IS}~\texttt{NOT}~\texttt{NULL} for the
\texttt{workEmail} attribute of \texttt{v2}.  On the other hand, the
\texttt{LEFT} \texttt{JOIN} of the query assigns a \texttt{NULL} value to
\texttt{workEmail}  if no tuple from \texttt{v2} satisfies the join condition
for a given tuple from \texttt{v1}.
We call an assignment of \texttt{NULL} values by a \texttt{LEFT} \texttt{JOIN} the
\emph{padding effect}.
A closer inspection of the query reveals, however, that the padding effect 
only applies when
\texttt{workEmail} in \texttt{v2} is \texttt{NULL}.  Thus,
the role of the \texttt{LEFT} \texttt{JOIN} in this query boils down
to re-introducing \texttt{NULL}s  eliminated by the mapping.
In fact, this situation is quite typical in OBDA but does not concern RDF
triplestores, which do not store \texttt{NULL}s, or classical data integration
systems, which can expose \texttt{NULL}s through their mappings.
\qed
\end{example}

In this paper we address these issues, and our contribution is summarised as follows.
\begin{asparaenum}[\itshape 1.]
\item In Sec.~\ref{sec:translation}, we provide a succinct translation of a fragment of SPARQL~1.1 with
  \texttt{OPTIONAL} and \texttt{MINUS} into relational algebra that relies
  on the use of \texttt{LEFT} \texttt{JOIN} and \texttt{COALESCE}. Even though the
  ideas can be traced back to Cyganiak~\cite{Cyga05} and Chebotko \emph{et al.}~\cite{ChLF09} for the earlier
  SPARQL~1.0, 
  our translation fully respects \emph{bag semantics} and the \emph{three-valued
   logic} of SPARQL~1.1 and SQL~\cite{GL17} (and is formally proven correct).
\item We develop optimisation techniques for SQL queries with complex \texttt{LEFT} \texttt{JOIN}s resulting from 
the translation and unfolding: 
Compatibility Filter Reduction (CFR, Sec.~\ref{sec:comp-filter-red}), which generalises~\cite{ChLF09},
\texttt{LEFT} \texttt{JOIN} Naturalisation (LJN, Sec.~\ref{sec:lj:nat})   to avoid  padding, Natural \texttt{LEFT} \texttt{JOIN} Reduction (NLJR, Sec.~\ref{sec:lj:red}),
\texttt{JOIN} Transfer (JT, Sec.~\ref{sec:join-transfer}) and \texttt{LEFT} \texttt{JOIN} Decomposition (LJD, Sec.~\ref{sec:lj:simpl}) complementing~\cite{GaRo97}. By CFR and LJN, compatibility filters and \texttt{COALESCE} are eliminated for well-designed SPARQL (Sec.~\ref{sec:wd:sparql}).
\item We carried out an evaluation of our optimisation
  techniques over the well-known OBDA benchmark BSBM~\cite{BiSc09}, 
  where
   \texttt{OPTIONAL}s, \texttt{LEFT} \texttt{JOIN}s and \texttt{NULL}s are ubiquitous.
  Our experiments (Sec.~\ref{sec:experiments}) show that the techniques of Sec.~\ref{sec:optimization} lead to a significant
  improvement in performance of the SQL translations, even for commercial DBMSs.
\end{asparaenum}


\section{Preliminaries}

We first formally define the syntax and semantics of the SPARQL fragment
we deal with and then present the relational algebra operators used for the translation from SPARQL. 

RDF provides a basic data model. Its vocabulary contains three pairwise
disjoint and countably infinite sets of symbols: IRIs $\myIRI$,
blank nodes $\myBNK$  and RDF literals $\myLIT$.  \emph{RDF
 terms} are elements of $\sDOM = \myIRI \cup \myBNK \cup \myLIT$, \emph{RDF triples} are elements of
$\sDOM \times \myIRI \times \sDOM$,  and an \emph{RDF graph}
is a finite set of RDF triples.

\subsection{SPARQL}
\label{sec:sparql}

SPARQL adds a countably infinite set $\myVAR$ of \emph{variables}, disjoint
from $\sDOM$.
A \emph{triple pattern} is an element of
$(\sDOM\cup\myVAR)\times (\myIRI\cup\myVAR) \times
(\sDOM\cup\myVAR)$. A \emph{basic graph pattern} (\emph{BGP}) is a
finite set of triple patterns. We consider
\emph{graph patterns}, $P$, defined by the grammar\footnote{A slight extension of the grammar and the full translation are given in Appendix~\ref{app:proof}.}
\begin{multline*}
  P \ ::= \
    B \ \mid \  \filter(P,F)  \ \mid \ 
    \union(P_1,P_2) \ \mid \join(P_1,P_2) \ \mid\\
    \leftjoin(P_1,P_2,F) \ \mid \ \minus(P_1, P_2) \ \mid \  \project(P,
    L),
\end{multline*}  
where $B$ is a BGP, $L\subseteq\myVAR$ and $F$, called a \emph{filter}, is a formula constructed
using logical connectives $\land$ and~$\neg$ from atoms of the form
$\textit{bound}(v)$, \mbox{$(v=c)$}, $(v=v')$, for $v,v' \in \myVAR$ and
$c \in \sDOM$.  The set of variables in
$P$ is denoted by $\var(P)$.  

Variables in graph patterns are assigned values by \emph{solution mappings},
which are \emph{partial} functions $s \colon \myVAR \to \sDOM$ with (possibly empty)
domain $\dom(s)$.  The \emph{truth-value $F^s \in \{\top,\bot,\varepsilon\}$ of
 a filter $F$ under a solution mapping $s$} is defined inductively:
\begin{compactitem}
\item $(\textit{bound}(v))^s$ is $\top$ if $v \in \dom(s)$, and $\bot$
  otherwise;
\item $(v = c)^s = \varepsilon$ (`error') if $v\notin\dom(s)$; otherwise, $(v = c)^s$ is
  the classical truth-value of the predicate $s(v) = c$; similarly,
  $(v = v')^s = \varepsilon$ if $\{v,v'\}\not\subseteq\dom(s)$; otherwise,
  $(v = v')^s$ is the classical truth-value of the predicate $s(v) = s(v')$;
\item $(\lnot F)^s=\begin{cases}
    \bot, & \text{if } F^s = \top,\\[-2pt]
    \top, & \text{if } F^s = \bot,\\[-2pt]
    \varepsilon, & \text{if } F^s = \varepsilon,
  \end{cases}$
  \quad and \quad
  $(F_1 \land F_2)^s = \begin{cases}%
    \bot, & \text{if } F_1^s =\bot \text{ or } F_2^s =\bot,\\[-2pt]
    \top, & \text{if  } F_1^s =F_2^s =\top,\\[-2pt]
    \varepsilon, &  \text{otherwise.}
  \end{cases}$
\end{compactitem}

We adopt bag semantics for SPARQL: the answer to a graph pattern over an RDF
graph is a multiset (or bag) of solution mappings. Formally, a \emph{bag of solution
 mappings} is a (total) function $\Omega$ from the set of all solution
mappings to non-negative integers~$\mathbb{N}$: $\Omega(s)$ is called the
\emph{multiplicity} of $s$ (we often use $s\in\Omega$ as a shortcut for
$\Omega(s) > 0$). Following the grammar of graph patterns, we define respective
operations on solution mapping bags. Solution mappings $s_1$ and~$s_2$ are
called \emph{compatible}, written $s_1\sim s_2$, if $s_1(v) = s_2(v)$, for each
$v \in \dom(s_1) \cap \dom(s_2)$, in which case $s_1\smerge s_2$ denotes a
solution mapping with domain $\dom(s_1) \cup \dom(s_2)$ and such that
$s_1 \smerge s_2\colon v \mapsto s_1(v)$, for $v \in \dom(s_1)$, and
$s_1 \smerge s_2\colon v \mapsto s_2(v)$, for $v \in \dom(s_2)$. We also denote
by $s|_L$ the restriction of $s$ on~$L\subseteq \myVAR$. Then the SPARQL
operations are defined as follows:
\begin{compactitem} 
\item $\filter(\Omega,F) = \Omega'$, where $\Omega'(s) = \Omega(s)$ if $s\in\Omega$ and $F^s = \top$, and $0$ otherwise;
\item $\union(\Omega_1, \Omega_2) = \Omega$, where  $\Omega(s) = \Omega_1(s) + \Omega_2(s)$;
\item $\join(\Omega_1,\Omega_2) = \Omega$, where
$\Omega(s) = \hspace*{-0.5em}\sum\limits_{\begin{subarray}{c}s_1\in \Omega_1, s_2 \in \Omega_2 \text{ with}\\s_1\sim s_2 \text{ and } s_1\smerge s_2 = s\end{subarray}} \hspace*{-1.5em}\Omega_1(s_1) \times \Omega_2(s_2)$;
\item \mbox{$\leftjoin(\Omega_1, \Omega_2, F) = \union(\filter(\join(\Omega_1, \Omega_2), F), \Omega)$, where 
$\Omega(s) = \Omega_1(s)$} if $F^{s\smerge s_2} \ne \top$, for all $s_2\in \Omega_2$ compatible with $s$, and $0$ otherwise;
\item $\minus(\Omega_1,\Omega_2) = \Omega$, where $\Omega(s) = \Omega_1(s)$ 
if $\dom(s) \cap \dom(s_2) = \emptyset$, for all solution mappings~$s_2\in \Omega_2$ compatible with $s$, and $0$ otherwise;
\item $\project(\Omega, L) = \Omega'$, where $\Omega'(s') = \sum\limits_{s \in \Omega \text{ with }s|_L = s'} \hspace*{-1em}\Omega(s)$.
\end{compactitem}
Given an RDF graph $G$ and a graph pattern $P$, the \emph{answer $\sANS{P}{G}$ to $P$ over $G$} is a bag of solution mappings defined by induction using the operations above and starting from  basic graph patterns: $\sANS{B}{G}(s) = 1$ if  $\dom(s) = \var(B)$ and $G$ contains the triple $s(B)$ 
obtained by replacing each variable $v$ in $B$ by $s(v)$, and $0$ otherwise ($\sANS{B}{G}$ is a set).


\subsection{Relational Algebra (RA)}

We recap the three-valued and bag semantics of relational
algebra~\cite{GL17} and fix the notation. Denote by $\Delta$ the
underlying domain, which contains a distinguished element $\Null$.
Let $U$ be a finite (possibly empty) set of \emph{attributes}. A \emph{tuple
  over} $U$ is a (total) map~\mbox{$\tuple \colon U \to \Delta$};  there is a unique tuple over $\emptyset$.
A \emph{relation~$R$ over $U$} is a \emph{bag} of tuples over~$U$, that is, a
function from all tuples over $U$ to~$\mathbb{N}$. For relations $R_1$ and $R_2$ over~$U$, we 
write~$R_1 \subseteq R_2$ ($R_1 \equiv R_2$) if $R_1(\tuple) \leq R_2(\tuple)$ ($R_1(\tuple) = R_2(\tuple)$, resp.), for all~$\tuple$.

A \emph{term} $v$ over $U$ is 
an attribute $u\in U$, 
a constant $c\in\Delta$ or
an expression $\textit{if}(F,v,v')$, for terms $v$ and $v'$ over $U$ and
  a filter $F$ over $U$.
A \emph{filter $F$ over $U$} is a formula constructed from atoms $\isNull(V)$
and $(v=v')$, for a set $V$ of terms and terms $v, v'$ over~$U$, using
connectives $\land$ and $\neg$. 
Given a tuple $\tuple$ over $U$, it is extended to terms as follows: 
\begin{equation*}
\tuple(c) = c, \text{ for constants } c\in\Delta,\qquad \text{ and } \qquad 
\tuple(\textit{if}(F, v,v')) = \begin{cases}\tuple(v), &\text{if } F^{\tuple} = \top,\\
\tuple(v'), & \text{otherwise},\end{cases}
\end{equation*}
where the \emph{truth-value $F^{\tuple} \in \{\top, \bot, \varepsilon\}$ of $F$ on $\tuple$} is defined
inductively ($\varepsilon$ is \emph{unknown}):
\begin{compactitem}
\item $(\isNull(V))^{\tuple}$ is $\top$ if $\tuple(v)$ is $\Null$, for all $v\in V$, and $\bot$ otherwise; 
\item  $(v = v')^{\tuple} =\varepsilon$ if  $\tuple(v)$ or $\tuple(v')$ is $\Null$, and the truth-value of 
$\tuple(v) = \tuple(v')$ otherwise;
\item and the standard clauses for $\neg$ and $\land$ in the three-valued logic (see~Sec.~\ref{sec:sparql}).
\end{compactitem}
We use standard abbreviations $\coalesce(v,v')$ for $\textit{if}(\neg\isNull(v), v, v')$ and $F_1 \lor F_2$ for~$\neg(\neg F_1\land \neg F_2)$. 
Unlike Chebotko~\emph{et al.}~\cite{ChLF09}, we treat $\textit{if}$ as primitive, even though 
the renaming operation with an $\textit{if}$ could be defined via standard  operations of RA. 

For filters in positive contexts, we define a weaker equivalence: 
filters $F_1$ and $F_2$ over $U$ are \emph{p-equivalent}, written 
$F_1 \!\equiv^{\scriptscriptstyle+}\! F_2$, in case $F_1^t = \top$ iff $F_2^t = \top$, for all  $t$ over~$U$.

We use standard relational algebra operations: union $\cup$, difference
$\setminus$, projection $\pi$, selection $\sigma$, renaming $\rho$, extension $\nu$, natural
(inner) join $\Join$ and duplicate elimination~$\delta$.
We say that tuples $\tuple_1$ over $U_1$ and $\tuple_2$ over $U_2$ are
\emph{compatible}\footnote{Note that, unlike in SPARQL, if $u$ is $\Null$ in
  either of the tuples, then they are incompatible.} if
$\tuple_1(u) = \tuple_2(u) \ne \Null$, for all $u \in U_1 \cap U_2$, in which
case $\tuple_1 \smerge \tuple_2$ denotes a tuple over $U_1 \cup U_2$ such that
$\tuple_1 \smerge \tuple_2\colon u\mapsto \tuple_1(u)$, for $u \in U_1$, and
$\tuple_1 \smerge \tuple_2\colon u \mapsto \tuple_2(u)$, for $u \in U_2$. For
a tuple $\tuple_1$ over $U_1$ and $U\subseteq U_1$, we denote by $\tuple_1|_U$
the restriction of $\tuple_1$ to $U$.
Let $R_i$ be relations over $U_i$, for $i=1,2$. The semantics of the above
operations is as follows:
\begin{compactitem}
\item If $U_1 = U_2$, then $R_1 \cup R_2$ and $R_1 \setminus R_2$ are
  relations over~$U_1$ satisfying 
  $(R_1\cup R_2)(\tuple) = R_1(\tuple) + R_2(\tuple)$ and
  $(R_1\setminus R_2)(\tuple)\!=\!R_1(\tuple)$ if $\tuple\notin R_2$ and $0$
  otherwise;
\item If $U \subseteq U_1$, then $\pi_{U}R_1$ is a relation over $U$ with
  $\pi_U R_1(\tuple) = \hspace*{-1em}\sum\limits_{\tuple_1\in R_1 \text{ with }
    \tuple_1|_U = \tuple} \hspace*{-1.5em}R_1(\tuple_1)$;
\item If $F$ is a filter over $U_1$, then $\sigma_F R_1$ is a relation over
  $U_1$ such that $\sigma_F R_1(\tuple)$ is $R_1(\tuple)$ if $\tuple\in R_1$
  and $F^{\tuple} = \top$, and $0$ otherwise;
\item $R_1 \Join R_2$  is a relation $R$ over $U_1 \cup U_2$ such that $R(\tuple) = \hspace*{-2em}\sum\limits_{\begin{subarray}{c}\tuple_1\in R_1 \text{ and } \tuple_2\in R_2\\\text{are compatible and } \tuple_1\smerge \tuple_2 = t\end{subarray}} \hspace*{-3em}R_1(\tuple_1) \times R_2(\tuple_2)$;
\item If $v$ is
  a term over $U_1$ and $u\notin U_1$ an attribute, then
  the \emph{extension} $\nu_{u \mapsto v}R_1$ is a relation $R$
  over $U_1\cup \{u\}$ with
  $R(\tuple \oplus \{ u \mapsto \tuple(v) \}) =
  R_1(\tuple)$, for all~$\tuple$. The \emph{extended
    projection} $\pi_{\{u_1/v_1,\dots, u_k/v_k\}}$ is a shortcut for
  $\pi_{\{u_1,\dots,u_k\}} \nu_{u_1\mapsto v_1}\cdots \nu_{u_k\mapsto v_k}$.
\item If $v\in U_1$ and $u\notin U_1$ are
  distinct attributes, then the \emph{renaming}
  $\rho_{u/v}R_1$ is a relation over
  $U_1\setminus \{v\} \cup \{u\}$ whose tuples $t$
  are obtained by replacing $v$ in the domain of~$t$ by~$u$. For terms $v_1,\dots,v_k$ over $U_1$, attributes $u_1,\dots,u_k$ (not necessarily distinct from~$U_1$) and $V\subseteq U_1$, let $u_1',\dots,u_k'$ be fresh attributes and abbreviate the sequence $\rho_{u_1/u'_1}\cdots \rho_{u_k/u'_k} \pi_{U_1 \cup \{u_1',\dots, u_k'\}\setminus V} \nu_{u_1' \mapsto v_1} \cdots \nu_{u_k' \mapsto v_k}$ by $\rho^V_{\{u_1/v_1,\dots, u_k/v_k\}}$. 

\item $\delta R_1$ is a relation over $U_1$ with $\delta R_1(\tuple) = \min(R_1(\tuple), 1)$.
\end{compactitem}
To bridge the gap between partial functions (solution mappings) of SPARQL and total functions (tuples) of RA, we use a \emph{padding} operation:
$\mu_{\{u_1, \dots, u_k\}} R_1$ denotes\linebreak $\nu_{u_1\mapsto \Null} \cdots \nu_{u_k\mapsto \Null} R_1$, for $u_1, \dots, u_k\notin U_1$.
Finally, we define the outer union, the (inner) join and left (outer) join operations by taking
\begin{multline*}
R_1 \uplus R_2 \ \  = \ \ \mu_{U_2\setminus U_1} R_1 \ \cup \ \mu_{U_1\setminus U_2} R_2,\qquad\quad
R_1 \Join_{F} R_2 \ \  = \ \  \sigma_F (R_1 \Join R_2),\\[-2pt]
R_1 \LJoin_{F} R_2 \ \  = \ \  (R_1 \Join_F R_2)  \ \uplus \  
 (R_1 \setminus \pi_{U_1}(R_1\Join_F R_2));
\end{multline*}
note that $\Join_{F}$ and  $\LJoin_{F}$ are \emph{natural joins}: they are over $F$ as well as shared attributes.

An \emph{RA query} $Q$ is an expression constructed from relation symbols, each with a fixed set of attributes, and filters using the RA operations (and complying with all restrictions). 
A \emph{data instance} $D$ gives  a relation over its set of attributes, for any relation symbol. The \emph{answer  to~$Q$ over $D$} is a relation $\|Q\|_D$ defined inductively in the obvious way starting from the base case of relation symbols:  $\|Q\|_D$ is the relation given by $D$.


\section{Succinct Translation of SPARQL to SQL}
\label{sec:translation}

We first provide a translation of SPARQL graph patterns to RA queries that improves the worst-case exponential translation of \cite{KRRXZ14}  in handling \join, \leftjoin{} and \minus: it relies on the \coalesce{} function (see also \cite{ChLF09,DBLP:conf/kesw/ChaloupkaN16}) and
produces linear-size RA queries.

For any graph pattern $P$, the RA query $\tra(P)$ returns the same answers as $P$ when solution mappings are represented as relational tuples. For a set $V$ of variables and  solution mapping $s$ with $\dom(s)\subseteq V$, let $\extV(s)$ be the tuple over $V$ obtained from $s$ by padding it with $\Null$s:  formally, $\extV(s)  =  s\smerge \{ v\mapsto \Null \mid v\in V \setminus\dom(s)\}$.
The \emph{relational answer $\| P\|_G$ to $P$ over an RDF graph $G$} is a bag $\Omega$ of tuples over~$\var(P)$ such that $\Omega(\extV[\var(P)](s)) = \sANS{P}{G}(s)$, for all solution mappings $s$.
Conversely, to evaluate $\tra(P)$, we view an RDF graph $G$ as a data instance $\textit{triple}(G)$ storing~$G$ as a ternary relation $\textit{triple}$ with the attributes $\textit{sub}$, $\textit{pred}$ and $\textit{obj}$ (note that $\textit{triple}(G)$ is a set).

The translation of a triple pattern $\langle s,p,o\rangle$ is an RA query of the form $\pi_{\dots}\sigma_F\textit{triple}$, where the subscript of the extended projection $\pi$ and filter $F$ are determined by the variables, IRIs and literals in $s$, $p$ and $o$; see Appendix~\ref{app:proof}. 
SPARQL operators $\union$, $\filter$ and $\project$ are translated into their RA counterparts: $\uplus$, $\sigma$ and $\pi$, respectively, with SPARQL filters translated into RA by replacing each $\textit{bound}(v)$ with $\neg \isNull(v)$.

The translation of $\join$, $\leftjoin$ and $\minus$ is more elaborate and requires additional notation.
Let $P_1$ and $P_2$ be graph patterns with $U_i = \var(P_i)$, for $i = 1,2$, and denote by~$U$ their shared variables, $U_1\cap U_2$.  To rename the shared attributes apart, we introduce fresh attributes $u^1$ and $u^2$ for each $u\in U$, set $U^i = \{ u^i  \mid u \in U\}$ and use abbreviations~$U^i/U$ and $U/U^i$ for $\{u^i / u \mid\, u \in U\}$ and $\{u / u^i \mid u \in U\}$, respectively, \mbox{for $i = 1,2$}. Now we can express the SPARQL solution mapping compatibility: 
\begin{equation*}
\compatible_U \ \ = \ \bigwedge\nolimits_{u\in U}  \bigl[(u^1 = u^2) \lor
  \isNull(u^1) \lor
  \isNull(u^2)\bigr]
\end{equation*}
(intuitively, the $\Null$ value of an attribute in the context of RA queries represents the fact that the corresponding SPARQL variable is not bound). Next, the renamed apart attributes need to be coalesced to provide the value in the  representation of the resulting solution mapping; see $\smerge$ in Sec.~\ref{sec:sparql}. To this end, 
given an RA filter $F$ over a set of attributes $V$, terms $v_1,\dots,v_k$ over $V$ and attributes $u_1,\dots,u_k\notin V$, we denote by $F[u_1/v_1,\dots, u_k/v_k]$ the result of replacing each $u_i$ by $v_i$ in $F$. We also 
denote by $\coalesce_U$ the substitution of each $u\in U$ with $\coalesce(u^1,u^2)$; thus, $F[\coalesce_U]$ is the result of replacing each $u\in U$ in $F$ with $\coalesce(u^1,u^2)$.  We now set
\begin{align*}
\tra(\join(P_1, P_2))  &= \rho^{U^1\cup U^2}_{\coalesce_U} \bigl[\rho_{U^1/U} \tra(P_1) \Join_{\compatible_U} \rho_{U^2/U}\tra(P_2)\bigr], \\[2pt]
  \tra(\leftjoin(P_1, P_2, F))  &= \rho^{U^1 \cup U^2}_{\coalesce_U} \bigl[\rho_{U^1/U} \tra(P_1) \LJoin_{\compatible_U\wedge \tra(F)[\coalesce_U]}  \rho_{U^2/U} \tra(P_2)\bigr], \\[2pt]
  \tra(\minus(P_1, P_2))  &= \pi_{U_1}\rho_{U/U^1}  \sigma_{\isNull(w)}\\[-2pt] & \hspace*{4em}\bigl[\rho_{U^1/U} \tra(P_1) \LJoin_{\compatible_U\wedge \bigvee\limits_{u \in U}(u^1 = u^2)} \nu_{w\mapsto 1}\rho_{U^2/U} \tra(P_2)\bigr],    
\end{align*}
where $w\notin U_1 \cup U_2$ is an attribute and  $1\in \Delta\setminus \{\Null\}$ is any domain element. The translation of $\join$ and $\leftjoin$ is straightforward. For $\minus$, observe that $\nu_{w\mapsto 1}$ extends the relation for $P_2$ by a fresh attribute $w$ with a non-$\Null$ value. The join condition encodes compatibility of solution mappings whose domains, in addition, share a variable (both $u^1$ and $u^2$ are non-$\Null$). Tuples satisfying the condition are then filtered out by $\sigma_{\isNull(w)}$, leaving only representations of solution mappings for~$P_1$ that have no compatible solution mapping in~$P_2$ with a shared variable. Finally, the attributes are renamed back by $\rho_{U/U^1}$ and unnecessary attributes are projected out by~$\pi_{U_1}$.

\begin{theorem}\label{thm:sparql-sql}
For any RDF graph $G$ and any graph pattern $P$,  
$\| P\|_G = \|\tra(P)\|_{\textit{triple}(G)}.$
\end{theorem}

The complete proof of Theorem~\ref{thm:sparql-sql} can be found in Appendix~\ref{app:proof}.


\section{Optimisations of Translated SPARQL Queries}
\label{sec:optimization}

We present optimisations on a series of examples. We begin by revisiting Example~\ref{ex:email-pref-simple}, which can now be given in algebraic form (for brevity, we ignore projecting away \texttt{?p}, which does not affect any of the optimisations discussed):
\begin{multline*}
\leftjoin(\leftjoin(\texttt{?p :name ?n}, \ \ \texttt{?p :workEmail ?e}, \ \ \top), \texttt{?p :personalEmail ?e}, \ \ \top),
\end{multline*}
where $\top$ denotes the tautological filter (true). Suppose we have the mapping
\begin{align*}
      \triple{\urione(\texttt{id})}{\texttt{:name}}{\texttt{fullName}} & \ \ \leftarrow\ \ 
                                   \sigma_{\notnull{\texttt{id}}  \land \notnull{\texttt{fullName}}}\texttt{people}, \\[-2pt]
      \triple{\urione(\texttt{id})}{\texttt{:workEmail}}{\texttt{workEmail}} & \ \ \leftarrow \ \
                                   \sigma_{\notnull{\texttt{id}}  \land\notnull{\texttt{workEmail}}}\,\texttt{people},\\[-2pt]
      \triple{\urione(\texttt{id})}{\texttt{:personalEmail}}{\texttt{homeEmail}} & \ \ \leftarrow\ \ 
                                   \sigma_{\notnull{\texttt{id}}  \land\notnull{\texttt{homeEmail}}}\,\texttt{people},
\end{align*}
where $\urione$ is a function that constructs the IRI for a person from their ID (an \emph{IRI template}, in R2RML parlance). We assume that the IRI functions are injective and map only $\Null$ to $\Null$; thus, joins on $\urione(\texttt{id})$ can be reduced to joins on $\texttt{id}$, and $\isNull(\texttt{id})$ holds just in case $\isNull(\urione(\texttt{id}))$ holds.  Interestingly,
the IRI functions can encode GLAV mappings, where the
target query is a full-fledged CQ (in contrast to GAV mappings, where atoms do
not contain existential variables); for more details,
see~\cite{DLLPR18}.

The translation given in Sec.~\ref{sec:translation}
and unfolding produce the following RA query, where we abbreviate, for example, $\rho^{\{p^1, p^2\}}_{\{p^4/\coalesce(p^1,p^2)\}}$ by $\bar{\rho}_{\{p^4/\coalesce(p^1,p^2)\}}$ (in other words, the $\bar{\rho}$ operation always projects away the arguments of its $\coalesce$ functions):\\[2pt]
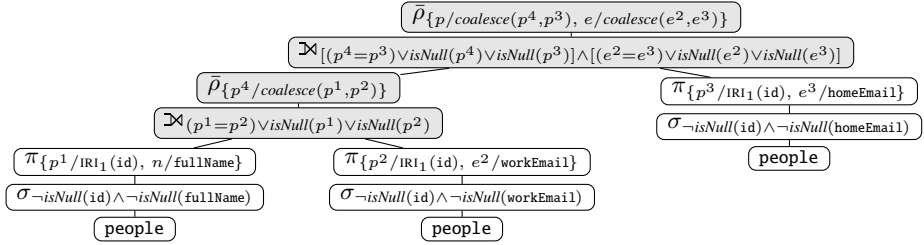
\centerline{\begin{tikzpicture}[xscale=2.9,yscale=0.47]\small
\node[nd,fill=white] (u3) at (4,2) {\footnotesize\texttt{people}};
\node[nd,fill=white] (u3s) at (4,2.9) {$\sigma_{\neg\isNull(\texttt{id})  \land\neg\isNull(\texttt{homeEmail})}$};
\node[nd,fill=white] (u3r) at (4,3.9) {$\pi_{\{p^3/\urione(\texttt{id}),\ e^3/\texttt{homeEmail}\}}$};
\node[nd,fill=white] (p2) at (2.5,0) {\footnotesize\texttt{people}};
\node[nd,fill=white] (p2s) at (2.5,0.9) {$\sigma_{\neg\isNull(\texttt{id})  \land \neg\isNull(\texttt{workEmail})}$};
\node[nd,fill=white] (p2r) at (2.5,1.9) {$\pi_{\{p^2/\urione(\texttt{id}),\ e^2/\texttt{workEmail}\}}$};
\node[nd,fill=white] (p1) at (1,0) {\footnotesize\texttt{people}};
\node[nd,fill=white] (p1s) at (1,0.9) {$\sigma_{\neg\isNull(\texttt{id})  \land\neg\isNull(\texttt{fullName})}$};
\node[nd,fill=white] (p1r) at (1,1.9) {$\pi_{\{p^1/\urione(\texttt{id}),\ n/\texttt{fullName}\}}$};
\node[nd] (lj1) at (1.75,3) {$\LJoin_{(p^1 = p^2)  \lor \isNull(p^1) \lor \isNull(p^2)}$}; 
\node[nd] (lj1p) at (1.75,4) {$\bar{\rho}_{\{p^4/\coalesce(p^1,p^2)\}}$};
\node[nd] (lj2) at (3,5) {$\LJoin_{ [(p^4 = p^3) \lor \isNull(p^4) \lor \isNull(p^3)]\land [(e^2 = e^3) \lor \isNull(e^2) \lor \isNull(e^3)]}$}; 
\node[nd] (lj2r) at (3,6) {$\bar{\rho}_{\{p/\coalesce(p^4,p^3), \ e/\coalesce(e^2,e^3)\}}$}; 
\draw (u3) -- (u3s);
\draw (u3s) -- (u3r);
\draw (p2) -- (p2s);
\draw (p2s) -- (p2r);
\draw (p1) -- (p1s);
\draw (p1s) -- (p1r);
\draw (p1r) -- (lj1);
\draw (p2r) -- (lj1);
\draw (lj1) -- (lj1p);
\draw (lj1p) -- (lj2);
\draw (u3r) -- (lj2);
\draw (lj2) -- (lj2r);
\end{tikzpicture}}\\[0pt]
In our diagrams, the white nodes are the contribution of the mapping and the translation of the basic graph patterns: for example, the basic graph pattern \texttt{?p :name ?n} produces $\smash{\pi_{\{p^1/\urione(\texttt{id}),\ n/\texttt{fullName}\}}\sigma_{\neg\isNull(\texttt{id})  \land\neg\isNull(\texttt{fullName})}\texttt{people}}$ (we use attributes without superscripts if there is only one occurrence; otherwise, the superscript identifies the relevant subquery).
The grey nodes correspond to the translation of the SPARQL operations: for instance, the innermost left join is on $\smash{\compatible_{\{p\}}}$ with $p$ renamed apart to $\smash{p^1}$ and~$\smash{p^2}$; the outermost left join is on $\smash{\compatible_{\{p, e\}}}$, where $p$ is renamed apart to $\smash{p^4}$ and $\smash{p^3}$ and $e$ to $\smash{e^2}$ and $\smash{e^3}$; the two $\bar{\rho}$ are the respective renaming operations with $\coalesce$.  

\subsection{Compatibility Filter Reduction (CFR)}\label{sec:comp-filter-red}

We begin by simplifying the filters in (left) joins and eliminating renaming operations with $\coalesce$ above them (if possible). 
First, we can pull up the filters of the mapping through the extended projection and union by means of standard database equivalences: for example, for relations~$R_1$ and $R_2$ and a filter $F$ over $U$, we  have $\sigma_F(R_1 \cup R_2) \equiv \sigma_F R_1 \cup \sigma_F R_2$, and
$\pi_{U'} \sigma_{F'} R_1 \equiv \sigma_{F'} \pi_{U'} R_1$, if $F'$ is a filter over~$U' \subseteq U$, and $\rho_{u/v} \sigma_F R_1 \equiv \sigma_{F[u/v]} \rho_{u/v} R_1$, if $v \in U$ and $u\notin U$.

Second, the filters can be moved (in a restricted way)  between the arguments of a left join to its join condition: for relations~$R_1$ and $R_2$ over $U_1$ and $U_2$, respectively, and filters $F_1$, $F_2$ and $F$ over $U_1$, $U_2$  and~$U_1 \cup U_2$, respectively, we have
\begin{align}
\label{eq:lj:filter-left:d}
\sigma_{F_1} R_1 \LJoin_{F} R_2 \ \ & \equiv \ \ \sigma_{F_1}(R_1 \LJoin_F R_2), \\
\label{eq:lj:filter-left}
\sigma_{F_1} R_1 \LJoin_{F} R_2 \ \ & \equiv \ \ \sigma_{F_1} R_1 \LJoin_{F\land F_1} R_2,\\
\label{eq:lj:filter-right}
R_1 \LJoin_{F} \sigma_{F_2} R_2 \ \ & \equiv \ \ R_1 \LJoin_{F\land F_2} R_2;
\end{align}
observe that unlike $\sigma_{F_2}$ in~\eqref{eq:lj:filter-right}, the
selection $\sigma_{F_1}$ cannot  be entirely  eliminated
in~\eqref{eq:lj:filter-left} but can rather be `duplicated' above the
left join
using~\eqref{eq:lj:filter-left:d}. (We note that~\eqref{eq:lj:filter-left:d}
and~\eqref{eq:lj:filter-right} are well-known and can be found, e.g.,
in~\cite{GaRo97}.) Simpler equivalences hold for inner join: $\sigma_{F_1} R_1 \Join_F R_2 \equiv \sigma_{F\land F_1} (R_1 \Join R_2)$. These equivalences can be, in particular, used to pull up the $\neg\isNull$ filters from mappings to eliminate the $\isNull$ disjuncts in the compatibility condition $\textit{comp}_U$ of the (left) joins in the  translation by means of the standard p-equivalences of the three-valued logic: 
\begin{align}
  \label{eq:rem-disjunct}
 (F_1 \lor F_2) \land \neg F_2 & \ \ \equiv^{\scriptscriptstyle+} \ \ F_1 \land \neg F_2, \\
  \label{eq:rem-isnull}
 (v = v') \land \neg \isNull(v) & \ \ \equiv^{\scriptscriptstyle+} \ \ (v = v');
\end{align}
we note in passing that this step refines Simplification~3 of Chebotko~\emph{et~al.}~\cite{ChLF09}, which relies on the absence of other left joins in the arguments of a (left) join. 

Third, the resulting simplified compatibility conditions can eliminate $\coalesce$ from the renaming operations: for a relation $R$ over $U$ and $u^1, u^2\in U$, we clearly have 
\begin{align}\label{eq:coalesce:elimination}
\rho^{\{u^1, u^2\}}_{\{u/\coalesce(u^1,u^2)\}} \sigma_{\neg\isNull(u^1)} R \ \ \ \equiv \ \ \ \sigma_{\neg\isNull(u)} \pi_{U\setminus \{u^2\}} R[u/u^1],
\end{align}
where $R[u/u^1]$ is the result of replacing each $u^1$ in $R$ by $u$. 
This step generalises Simplification~2 of Chebotko~\emph{et~al.}~\cite{ChLF09}, which does not eliminate $\coalesce$ above (left) joins that contain nested left joins.

By applying these three steps to our running example, we obtain (see Appendix~\ref{app:example1})\\[4pt]
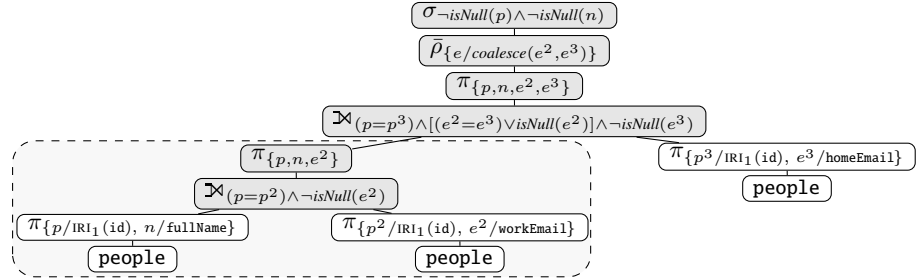
\centerline{\begin{tikzpicture}[xscale=2.9,yscale=0.47]\small
\draw[rounded corners=3mm,dashed,fill=black!3] (0.45,0.6) rectangle +(2.65,3.85);
\node[nd,fill=white] (u3) at (4,3.1) {\texttt{people}};
\node[nd,fill=white] (u3r) at (4,4) {$\pi_{\{p^3/\urione(\texttt{id}),\ e^3/\texttt{homeEmail}\}}$};
%
\node[nd,fill=white] (p2) at (2.5,1.1) {\texttt{people}};
\node[nd,fill=white] (p2r) at (2.5,2) {$\pi_{\{p^2/\urione(\texttt{id}),\ e^2/\texttt{workEmail}\}}$};
%
\node[nd,fill=white] (p1) at (1,1.1) {\texttt{people}};
\node[nd,fill=white] (p1r) at (1,2) {$\pi_{\{p/\urione(\texttt{id}),\ n/\texttt{fullName}\}}$};
%
\node[nd] (lj1) at (1.75,2.95) {$\LJoin_{(p = p^2) \land \neg\isNull(e^2)}$}; 
\node[nd] (lj1p) at (1.75,3.95) {$\pi_{\{p, n, e^2\}}$};
\node[nd] (lj2) at (2.75,5) {$\LJoin_{ (p = p^3)\land [(e^2 = e^3) \lor \isNull(e^2)] \land \neg\isNull(e^3)}$}; 
\node[nd] (lj2p) at (2.75,6) {$\pi_{\{p, n, e^2, e^3\}}$}; 
\node[nd] (lj2r) at (2.75,7) {$\bar{\rho}_{\{e/\coalesce(e^2,e^3)\}}$}; 
\node[nd] (lj2s) at (2.75,8) {$\sigma_{\neg\isNull(p)\land  \neg\isNull(n)}$}; 
\draw (u3) -- (u3r);
%
\draw (p2) -- (p2r);
%
\draw (p1) -- (p1r);
%
\draw (p1r) -- (lj1);
\draw (p2r) -- (lj1);
\draw (lj1) -- (lj1p);
\draw (lj1p) -- (lj2);
\draw (u3r) -- (lj2);
\draw (lj2) -- (lj2p);
\draw (lj2p) -- (lj2r);
\draw (lj2r) -- (lj2s);
\end{tikzpicture}}

\subsection{Left Join Naturalisation (LJN)}\label{sec:lj:nat}

Our next group of optimisations can remove join conditions in left joins (if their arguments satisfy certain properties), thus reducing them to \emph{natural left joins}.

Some equalities in the join conditions of left joins can be removed by means of attribute duplication: for relations $R_1$ and $R_2$ over~$U_1$ and $U_2$, respectively,  a filter $F$ over $U_1\cup U_2$ and attributes $u^1 \in U_1\setminus U_2$ and $u^2 \in U_2 \setminus U_1$, we have
\begin{align}\label{eq:lj:natural}
R_1 \LJoin_{F \land (u^1 = u^2)} R_2 & \ \  \equiv  \ \  R_1 \LJoin_F \nu_{u^1 \mapsto u^2} R_2.
\end{align}
Now, the duplicated $u^2$ can be eliminated in case it is actually projected away:
\begin{align}\label{eq:vacuous:nu}
\pi_{U_1\cup U_2\setminus \{u^2\}} (R_1 \LJoin_F \nu_{u^1 \mapsto u^2} R_2)  \equiv  R_1 \LJoin_F R_2[u^1/u^2] \text{ if }  F \text{ does not contain } u^2.  
\end{align}
So, if $F$ is a conjunction of suitable attribute equalities, then by repeated application of~\eqref{eq:lj:natural} and~\eqref{eq:vacuous:nu}, we can turn a left join  into a natural left join. In our running example, this procedure simplifies the innermost left join to\\[0pt]
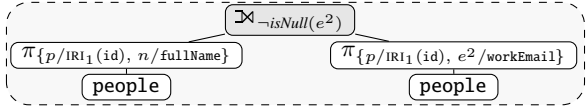
\centerline{\begin{tikzpicture}[xscale=2.9,yscale=0.47]\small
\draw[rounded corners=3mm,dashed,fill=black!3] (0.45,1.6) rectangle +(2.65,2.9);
%
\node[nd,fill=white] (p2) at (2.5,2.1) {\texttt{people}};
\node[nd,fill=white] (p2r) at (2.5,3) {$\pi_{\{p/\urione(\texttt{id}),\ e^2/\texttt{workEmail}\}}$};
%
\node[nd,fill=white] (p1) at (1,2.1) {\texttt{people}};
\node[nd,fill=white] (p1r) at (1,3) {$\pi_{\{p/\urione(\texttt{id}),\ n/\texttt{fullName}\}}$};
%
\node[nd] (lj1) at (1.75,4) {$\LJoin_{\neg\isNull(e^2)}$}; 
%
\draw (p2) -- (p2r);
%
\draw (p1) -- (p1r);
%
\draw (p1r) -- (lj1);
\draw (p2r) -- (lj1);
\end{tikzpicture}\hfill\mbox{}}

\smallskip

Another technique for converting a left join into a natural left join ($\LJoin$ is just an abbreviation for $\LJoin_\top$) is based on the conditional function $\textit{if}$: 
\begin{proposition}\label{prop:lj-nat}
For relations $R_1$ and $R_2$ over $U_1$ and $U_2$, respectively, and a filter $F$ over $U_1 \cup U_2$, we have
\begin{align}\label{eq:lj-nat}
R_1   \LJoin_F  R_2 \  \equiv \  \rho^{\{U_2\setminus U_1\}}_{\{u/\textit{if}(F, u, \Null) \,\mid\, u\in U_2\setminus U_1\}} (R_1 \LJoin R_2) && \text{ if } \ \ \pi_{U_1}(R_1 \Join R_2) \subseteq R_1.
\end{align}  
\end{proposition}
\begin{proof}
Denote $R_1\Join R_2$ by $S$. Then $\pi_{U_1} S
\subseteq R_1$ implies that every tuple $t_1$ in $R_1$ can have at most
one tuple $t_2$ in $R_2$ compatible with it, and $S$ consists of all
such extensions (with their cardinality determined by $R_1$). Therefore, $\pi_{U_1}(S \setminus \sigma_F S)$ is precisely the
tuples in $R_1$ that cannot be extended in such a way that the extension  satisfies~$F$, whence 
\begin{align}
\label{eq:LJ-nat:1}
\pi_{U_1}(S \setminus \sigma_F S) & \ \ \equiv \ \ \pi_{U_1}S \setminus \pi_{U_1}\sigma_F S.
\end{align}
By a similar argument, $R_1 \setminus \pi_{U_1} S$ consists of the tuples in $R_1$ (with the same cardinality) that cannot be extended by a tuple in $R_2$, and $\pi_{U_1} S \setminus \pi_{U_1} \sigma_F S$ of those tuples that can be extended but only when $F$ is not satisfied. By taking the union of the two, we obtain
\begin{align}
\label{eq:LJ-nat:2}
(R_1 \setminus \pi_{U_1} S) \ \ \cup \ \ (\pi_{U_1} S \setminus \pi_{U_1} \sigma_F S) & \ \ \equiv \ \ R_1 \setminus  \pi_{U_1} \sigma_F S. 
\end{align}
The claim  is then proved by distributivity of $\rho$ and $\mu$ over $\cup$; see Appendix~\ref{app:sec4}.
\end{proof}

Proposition~\ref{prop:lj-nat} is, in particular, applicable if the attributes shared by $R_1$ and~$R_2$ uniquely determine tuples of $R_2$. In our running example, $\texttt{id}$ is a primary key in $\texttt{people}$,  and so we can eliminate $\neg\isNull(e^2)$ from the innermost left join, which becomes a natural left join, and then simplify the term $\textit{if}(\neg\isNull(e^2), e^2, \Null)$ in the renaming to~$e^2$  by using equivalences on complex terms: for a term~$v$ and a filter $F$ over~$U$, we have
\begin{align}
\label{eq:if:1}
\textit{if}(F \land \neg\isNull(v), v, \Null) & \ \ \equiv \ \ \textit{if}(F, v, \Null),\\
\label{eq:if:2}
\textit{if}(\top, v, \Null) & \ \  \equiv \ \ v.
\end{align}
Thus, we effectively remove the renaming operator introduced by the application of Proposition~\ref{prop:lj-nat}; for full details, see Appendix~\ref{app:example1}.

\subsection{Translation for Well-Designed SPARQL}\label{sec:wd:sparql}

We remind the reader that a SPARQL pattern $P$ that uses only $\join$, $\filter$ and binary $\leftjoin$ (that is, $\leftjoin$ with the tautological filter $\top$) is \emph{well-designed}~\cite{PeAG09} if every its subpattern $P'$ of the form $\leftjoin(P_1, P_2, \top)$ satisfies the following condition: every variable~$u$ that occurs in $P_2$ and outside $P'$ also occurs in $P_1$. 
\begin{proposition}\label{prop:wd}
If $P$ is well-designed, then its unfolded translation can be equivalently simplified by \textup{(a)} removing  all compatibility filters $\compatible_U$ from joins and left joins and \textup{(b)} eliminating all renamings $u/\coalesce(u^1,u^2)$ by replacing both $u^1$ and $u^2$ with $u$.
\end{proposition}
\begin{proof}
Since $P$ is well-designed, any variable $u$  occurring in the right-hand side argument of any $\leftjoin$ either does not occur elsewhere (and so, can be projected away) or also occurs in the left-hand side argument. The claim then follows from an observation that, if the translation of  $P_1$ or $P_2$ can be equivalently transformed to contain a selection with $\neg\isNull(u)$ at the top, then the translation of $\join(P_1, P_2)$, $\leftjoin(P_1, P^*, \top)$ and $\filter(P_1,F)$ can also be equivalently simplified so that it contains a selection with the $\neg\isNull(u^1)$ or, respectively, $\neg\isNull(u^2)$ condition at the top. 
\end{proof}

Rodr\'\i{}guez-Muro \& Rezk~\cite{DBLP:journals/ws/Rodriguez-MuroR15} made a similar observation. Alas, Example~\ref{ex:email-pref-simple} shows that Proposition~\ref{prop:wd} is not directly applicable to \emph{weakly} well-designed SPARQL~\cite{DBLP:conf/icdt/KaminskiK16}.

\subsection{Natural Left Join Reduction (NJR)}\label{sec:lj:red}

A natural left join can then be replaced by a natural \emph{inner} join if every tuple of its left-hand side argument has a match on the right, which can be formalised as follows.  
\begin{proposition}\label{prop:lj-red}
For relations $R_1$ and~$R_2$ over $U_1$ and $U_2$, respectively, we have  
\begin{multline}\label{eq:lj-red}
\sigma_{\neg\isNull(K)} R_1 \LJoin R_2 \ \ \equiv \ \ R_1 \Join R_2,\quad \text{ if } \delta \pi_K R_1 \subseteq \pi_K R_2,   \text{ for }  K = U_1 \cap U_2.
\end{multline}
\end{proposition}
\begin{proof}
By careful inspection of definitions. Alternatively, one can assume
that the left join has an additional selection on top with filters of
the form $(u^1 = u^2) \lor \isNull(u^2)$, for $u\in K$,  where $u^1$
and $u^2$ are duplicates of attributes from $R_1$ and $R_2$,
respectively. Given $\delta \pi_K R_1 \subseteq \pi_K
R_2$, one can eliminate the $\isNull(u^2)$ because any
tuple of $R_1$ has a match in $R_2$. The resulting
$\Null$-rejecting filter then effectively turns the left join to an
inner join by the outer join simplification of Galindo-Legaria \& Rosenthal~\cite{GaRo97}.
\end{proof}

Observe that the inclusion $\delta \pi_K R_1 \subseteq \pi_K R_2$ is
satisfied, for example,  if $R_1$ has a foreign key $K$ referencing
$R_2$. It can also be satisfied if both $R_1$ and $R_2$ 
are based on the same relation, that is, $R_i \equiv \sigma_{F_i}\pi_{\dots} R$, for $i = 1,2$, and $F_1$ logically
implies $F_2$, where $F_1$ and/or $F_2$ can be $\top$ for the vacuous selection. Note that, due to $\delta$, attributes~$K$ do not have to uniquely determine tuples in $R_1$ or $R_2$. In our running example, trivially, 
$\delta \pi_{\{p\}} (\pi_{\{p/\urione(\texttt{id}),\ n/\texttt{fullName}\}} \texttt{people}) \ \ \subseteq \ \ \pi_{\{p\}} (\pi_{\{p/\urione(\texttt{id}),\ e^2/\texttt{workEmail}\}} \texttt{people})$.
Therefore, the inner left join can be replaced by a natural inner join, which can then be eliminated altogether because \texttt{id} is the primary key in \texttt{people} (this is a well-known optimisation; see, e.g.,~\cite{Elmasri:2010uk,RoKZ13}). As a result, we obtain\\[2pt]
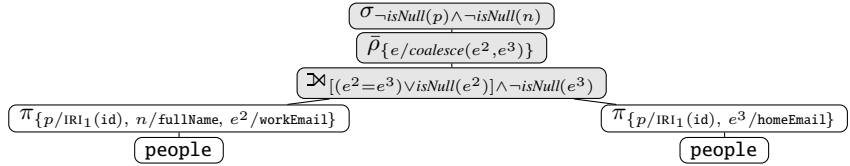
\centerline{\begin{tikzpicture}[xscale=2.9,yscale=0.47]\small
\node[nd,fill=white] (u3) at (4.25,3.1) {\texttt{people}};
\node[nd,fill=white] (u3r) at (4.25,4) {$\pi_{\{p/\urione(\texttt{id}),\ e^3/\texttt{homeEmail}\}}$};
%
%
\node[nd,fill=white] (p1) at (1.75,3.1) {\texttt{people}};
\node[nd,fill=white] (p1r) at (1.75,4) {$\pi_{\{p/\urione(\texttt{id}),\ n/\texttt{fullName},\ e^2/\texttt{workEmail}\}}$};
%
\node[nd] (lj2) at (3,5) {$\LJoin_{[(e^2 = e^3) \lor \isNull(e^2)] \land \neg\isNull(e^3)}$}; 
\node[nd] (lj2s) at (3,6.9) {$\sigma_{\neg\isNull(p)\land  \neg\isNull(n)}$}; 
\node[nd] (lj2r) at (3,6) {$\bar{\rho}_{\{e/\coalesce(e^2,e^3)\}}$}; 
\draw (u3) -- (u3r);
%
%
\draw (p1) -- (p1r);
%
\draw (p1r) -- (lj2);
%
\draw (u3r) -- (lj2);
\draw (lj2) -- (lj2r);
\draw (lj2r) -- (lj2s);
\end{tikzpicture}}\\[0pt]
The running example is wrapped up and discussed in detail in Appendices~\ref{app:example1} and~\ref{app:example1:discussion}.

\subsection{Join Transfer (JT)}\label{sec:join-transfer}

To introduce and explain another optimisation, we need an extension of relation \texttt{people} with a nullable attribute \texttt{spouseId}, which contains the \texttt{id} of the person's spouse if they are married and \texttt{NULL} otherwise. The attribute is mapped by an additional assertion:
\begin{equation*}
      \triple{\urione(\texttt{id})}{\texttt{:hasSpouse}}{\urione(\texttt{spouseId})} \quad\leftarrow\quad
                                   \sigma_{\notnull{\texttt{id}}  \land \notnull{\texttt{spouseId}}}\texttt{people}.
\end{equation*}
Consider now the following query in SPARQL algebra: 
\begin{equation*}
\project(\leftjoin(\texttt{?p :name ?n}, \ \join(\texttt{?p :hasSpouse ?s},\ \texttt{?s :name ?sn}), \ \top), \ \{\,\texttt{?n}, \texttt{?sn} \,\}),
\end{equation*}
whose translation can be unfolded and simplified with optimisations in Secs.~\ref{sec:comp-filter-red} and~\ref{sec:lj:nat} into the following RA query
(we have also pushed down the filter~$\neg\isNull(sn)$ to the right argument of the join and, for brevity, omitted selection and projection at the top):\\[2pt]
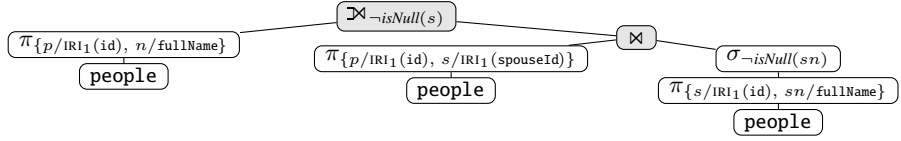
\centerline{\begin{tikzpicture}[xscale=2.9,yscale=0.47]\small
\node[nd,fill=white] (p3) at (4,1.6) {\texttt{people}};
\node[nd,fill=white] (p3r) at (4,2.5) {$\pi_{\{s/\urione(\texttt{id}),\ sn/\texttt{fullName}\}}$};
\node[nd,fill=white] (p3s) at (4,3.4) {$\sigma_{\neg\isNull(sn)}$};
%
\node[nd,fill=white] (p2) at (2.5,2.5) {\texttt{people}};
\node[nd,fill=white] (p2r) at (2.5,3.4) {$\pi_{\{p/\urione(\texttt{id}),\ s/\urione(\texttt{spouseId})\}}$};
%
\node[nd,fill=white] (p1) at (1,2.9) {\texttt{people}};
\node[nd,fill=white] (p1r) at (1,3.8) {$\pi_{\{p/\urione(\texttt{id}),\ n/\texttt{fullName}\}}$};
%
\node[nd] (j) at (3.35,4) {$\Join$}; 
%
\node[nd] (lj) at (2.25,4.6) {$\LJoin_{\neg\isNull(s)}$};
%
\draw (p3) -- (p3r);
\draw (p3r) -- (p3s);
\draw (p2) -- (p2r);
%
\draw (p1) -- (p1r);
%
\draw (p2r) -- (j);
\draw (p3s) -- (j);
%
\draw (j) -- (lj);
\draw (p1r) -- (lj);
\end{tikzpicture}}\\ 
see Appendix~\ref{app:join:transfer} for full details. Observe that the inner join cannot be eliminated using the standard self-join elimination techniques because it is not on a primary (or alternate) key. The next proposition (proved in Appendix~\ref{app:sec4}) provides a solution for the issue.
\begin{proposition}\label{prop:transfer}
Let $R_1$, $R_2$ and $R_3$ be relations over $U_1$, $U_2$ and $U_3$, 
respectively, $F$ a filter over $U_1\cup U_2 \cup U_3$ and $w$ an attribute in $U_3 \setminus (U_1 \cup U_2)$. Then
\begin{multline}\label{eq:transfer}
R_1   \LJoin_F  (R_2 \Join \sigma_{\neg\isNull(w)} R_3) \ \ \equiv\\[-2pt]
\rho^{\{U_2\setminus U_1\}}_{\{u/\textit{if}(\neg\isNull(w), u, \Null)\ \mid \ u\in
    U_2\setminus U_1 \}} ((R_1 \Join R_2) \LJoin_F \sigma_{\neg\isNull(w)} R_3),\\[-2pt] \text{ if } \pi_{U_1}(R_1 \Join R_2) \equiv R_1.
\end{multline}
\end{proposition}

By Proposition~\ref{prop:transfer},  we take $sn$ as the non-nullable attribute $w$ and get the following:\\[4pt]
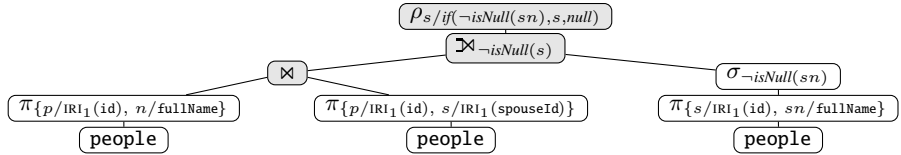
\centerline{\begin{tikzpicture}[xscale=2.9,yscale=0.47]\small
\node[nd,fill=white] (p3) at (4,2.1) {\texttt{people}};
\node[nd,fill=white] (p3r) at (4,3) {$\pi_{\{s/\urione(\texttt{id}),\ sn/\texttt{fullName}\}}$};
\node[nd,fill=white] (p3s) at (4,3.9) {$\sigma_{\neg\isNull(sn)}$};
\node[nd,fill=white] (p2) at (2.5,2.1) {\texttt{people}};
\node[nd,fill=white] (p2r) at (2.5,3) {$\pi_{\{p/\urione(\texttt{id}),\ s/\urione(\texttt{spouseId})\}}$};
%
\node[nd,fill=white] (p1) at (1,2.1) {\texttt{people}};
\node[nd,fill=white] (p1r) at (1,3) {$\pi_{\{p/\urione(\texttt{id}),\ n/\texttt{fullName}\}}$};
%
\node[nd] (j) at (1.75,4) {$\Join$}; 
%
\node[nd] (lj) at (2.75,4.7) {$\LJoin_{\neg\isNull(s)}$};
\node[nd] (ljr) at (2.75,5.6) {$\rho_{s / \textit{if}(\neg\isNull(sn), s, \Null)}$};  
%
\draw (p3) -- (p3r);
\draw (p3r) -- (p3s);
\draw (p2) -- (p2r);
%
\draw (p1) -- (p1r);
%
\draw (p1r) -- (j);
\draw (p2r) -- (j);
%
\draw (j) -- (lj);
\draw (p3s) -- (lj);
\draw (lj) -- (ljr);
\end{tikzpicture}}\\ 
Now, the inner self-join can be eliminated (as \texttt{id} is the primary key of \texttt{people}) and the~$\rho$ operation removed (as its result is projected away); see Appendix~\ref{app:join:transfer}. 

\subsection{Left Join Decomposition (LJD): Left Join Simplification~\cite{GaRo97} Revisited}\label{sec:lj:simpl}

In Sec.~\ref{sec:lj:red}, we have given an example of a reduction of a left join to an inner join. The following equivalence is also helpful (for an example, see Appendix~\ref{app:example1:discussion2}): for relations $R_1$ and $R_2$ over $U_1$ and $U_2$, respectively, and a filter $F$ over $U_1\cup U_2$,
\begin{align}\label{eq:lj:to:simple}
\pi_{U_1}(R_1 \LJoin_F R_2) \ \ \equiv \ \ R_1, \quad \text{ if }\quad \pi_{U_1} (R_1 \Join R_2) \subseteq R_1.
\end{align}

Galindo-Legaria \& Rosenthal~\cite{GaRo97} observe that $\sigma_G(R_1 \LJoin_F R_2)\equiv R_1 \Join_{F \land G} R_2$ whenever $G$ rejects $\Null$s on $U_2\setminus U_1$. In the context of SPARQL, however, the compatibility condition $\compatible_U$ does not satisfy the $\Null$-rejection requirement, and so, this optimisation is often not applicable. In the rest of this section we refine the basic idea. 

Let $R_1$ and $R_2$ be relations over $U_1$ and $U_2$, respectively, and $F$ and $G$ filters over~\mbox{$U_1 \cup U_2$}. It can easily be verified that, in general, we can \emph{decompose} the left join:
\begin{multline}
\label{eq:lj:general-filter}
 \sigma_{G}(R_1 \LJoin_F R_2) \ \  \equiv \ \  (R_1 \Join_{F \land
                                     G} R_2)  \ \ \ \ \uplus \\[-2pt]    
     \sigma_{\nullify_{U_2 \setminus U_1}(G)}R_1
\setminus \pi_{U_1}(R_1 \Join_{F \land \nullify_{U_2 \setminus U_1}(G)} R_2),
\end{multline}
where $\nullify_{U_2\setminus U_1}(G)$ is the result of replacing every occurrence of an attribute from $U_2\setminus U_1$ in $G$ with $\Null$. Observe that if $G$ is $\Null$-rejecting on $U_2\setminus U_1$, then\linebreak $\nullify_{U_2\setminus U_1}(G)\equiv^{\scriptscriptstyle+}\bot$, and the second component of the union in~\eqref{eq:lj:general-filter} is empty. We, however, are interested in a subtler interaction of the filters when the second component of the difference or, respectively, the first component of the union is empty: 
\begin{align}
\notag{} 
\sigma_{G}(R_1 \LJoin_F R_2) \ \ & \equiv \ \  R_1 \Join_{F \land G} R_2 \ \ \
                              \uplus \ \ \ \sigma_{\nullify_{U_2 \setminus U_1}(G)}R_1, \\[-2pt]
\label{eq:lj:no-minus}                              
& \hspace*{7em} \text{ if } F\land \nullify_{U_2 \setminus U_1}(G) \equiv^{\scriptscriptstyle+} \bot, \\[4pt]
\notag{} 
 \sigma_{G}(R_1 \LJoin_F \sigma_{\neg\isNull(w)} R_2) \ \ & \equiv \ \ 
                                     \sigma_{\isNull(w)\land \nullify_{U_2
                                     \setminus U_1}(G)}(R_1 \LJoin_F  \sigma_{\neg\isNull(w)} R_2),\\[-2pt]
\label{eq:lj:minus-encoding}                                    
& \hspace*{7em} \text{ if } F\land G \equiv^{\scriptscriptstyle+} \bot \text{ and } w\in U_2\setminus U_1.
\end{align}
These cases are of particular relevance for the SPARQL-to-SQL translation of \texttt{OPTIONAL} and \texttt{MINUS}.  We illustrate the technique in Appendix~\ref{app:vertical} on the following example:
\begin{align*}
&\filter(\leftjoin(\leftjoin(\texttt{?p a :Product}, \\[-3pt]
&\hspace*{1em} \filter(\texttt{\{ ?p :hasReview ?r . ?r :hasLang ?l \}}, \texttt{?l} = \texttt{"en"}), \ \top),\\[-3pt]
&\hspace*{1em} \filter(\texttt{\{ ?p :hasReview ?r . ?r :hasLang ?l \}}, \texttt{?l} = \texttt{"zh"}), \ \top),\  \textit{bound}(\texttt{?r})).
\end{align*}
The technique relies on two properties of $\Null$ propagation from the right-hand side of left joins. Let $R_1$ and $R_2$ be relations over $U_1$ and $U_2$, respectively. First, if $v = v'$ is a left join condition and $v$ is a term over $U_2\setminus U_1$, then $v$ is either $\Null$ or~$v'$ in the result:
\begin{align}\label{eq:lj:prop:1}
R_1 \LJoin_{F\land (v = v')} R_2 \ \equiv \ \ \sigma_{\isNull(v) \lor (v = v')} (R_1 \LJoin_{F\land (v = v')} R_2). 
\end{align}
Second, non-nullable terms $v, v'$ over $U_2\setminus U_1$ are simultaneously either $\Null$ or not $\Null$:
\begin{multline}\label{eq:lj:prop:2}
R_1 \LJoin_{F} \sigma_{\neg\isNull(v) \land\neg\isNull(v')} R_2 \ \ \equiv \\[-2pt] \sigma_{[\neg\isNull(v) \land \neg\isNull(v')] \lor [\isNull(v) \land \isNull(v')]} (R_1 \LJoin_F \sigma_{\neg\isNull(v) \land\neg\isNull(v')} R_2). 
\end{multline}
The two equivalences introduce \emph{no new} filters apart from $\isNull$ and their negations. The introduced filters, however, can help simplify the join conditions of the left joins containing the left join under consideration.


\section{Experiments}
\label{sec:experiments}

In order to verify effectiveness of our optimisation techniques,
we carried out a set of experiments based on the BSBM
benchmark~\cite{BiSc09};  the materials for reproducing the
experiments are available
online\footnote{\url{https://github.com/ontop/ontop-examples/tree/master/iswc-2018-optional}}\!.\
The BSBM benchmark is built around an e-commerce use case in which
vendors offer products that can be reviewed by customers. It
comes with a mapping, a data generator and a set of SPARQL 
and equivalent SQL queries.

\noindent\textbf{Hardware and Software.}
The experiments were performed on a \texttt{t2.xlarge} Amazon EC2 instance with four 64-bit vCPUs,
16G memory and 500G SSD hard disk under Ubuntu~16.04LTS. 
We used five database engines: free MySQL~5.7 and
PostgreSQL~9.6 are run normally, and 3 commercial systems (which we shall call X, Y and Z) in
Docker.

\noindent\textbf{Queries.} In total, we consider 11 SPARQL queries. 
Queries Q1--Q4 are based on the original BSBM queries 2, 3, 7 and 8, which contain
\texttt{OPTIONAL}; 
we modified them to reduce selectivity: e.g., Q1, Q3 and Q4 retrieve information about 1000 products rather than a single product in the original BSBM queries; we also removed \texttt{ORDER BY} and \texttt{LIMIT} clauses. Q1--Q4 are well-designed (WD).
In addition, we created 7 weakly well-designed (WWD) SPARQL queries:
Q5--Q7 are similar to Example~\ref{ex:email-pref-simple}, Q8--Q10 
to the query in Sec.~\ref{sec:lj:simpl}, and Q11 is along the lines of Sec.~\ref{sec:join-transfer}. More information is below:\\[2pt]
\centerline{\small\renewcommand{\tabcolsep}{2pt}\renewcommand{\arraystretch}{1}%
  \begin{tabular}{cp{76mm}cc}
    \toprule
    query & \hfil description & \!SPARQL\! & optimisations \\
    \midrule
 \rowcolor{lightgray!20}    Q1 & 2 simple \texttt{OPTIONAL}s
         for the padding effect\newline\scriptsize (derived from BSBM query 2) & WD & LJN, NLJR \\ 
    Q2 & 1 \texttt{OPTIONAL} with a \texttt{!BOUND} filter  (encodes \texttt{MINUS})\newline\scriptsize derived from BSBM query 3 & WD & JT \\ 
 \rowcolor{lightgray!20}    Q3 & 2 outer-level \texttt{OPTIONAL}s, the latter with 2 nested \texttt{OPTIONAL}s\newline\scriptsize derived from BSBM query 7 & WD & LJN, NLJR \\ 
    Q4 & 4 \texttt{OPTIONAL}s: ratings from attributes of the same relation\newline\scriptsize derived from BSBM query 8
                    &  WD   & LJN, NLJR \\ 
 \rowcolor{lightgray!20}    Q5/6/7 & 2/3/4 \texttt{OPTIONAL}s: preference over 2/3/4 ratings of reviews & WWD & LJN, NLJR \\ 
    Q8/9/10 & 2/3/4 \texttt{OPTIONAL}s: preference of reviews over 2/3/4 languages  & WWD & LJN, LJD\\
 \rowcolor{lightgray!20}   Q11 & 2 \texttt{OPTIONAL}s: country-based preference of home pages of\newline reviewed products  & WWD & LJN, NLJR, JT \\ 
    \bottomrule
  \end{tabular}
}

\smallskip

\noindent\textbf{Data.} We used the BSBM generator to produce CSV files for 1M products and 10M
reviews. The CSV files (20GB) were loaded into 
DBs, with the required indexes created.

\noindent\textbf{Evaluation.}
For each SPARQL query, we computed two SQL translations. The
\emph{non-optimised} (N/O) translation is obtained by applying to the
unfolded query only the standard (previously known and widely adopted)
structural and semantic optimisations~\cite{CCKK*17} 
as well as CFR (Sec.~\ref{sec:comp-filter-red}) to
simplify compatibility filters and eliminate unnecessary \texttt{COALESCE}.
To obtain the
\emph{optimised}  (O) translations, we  further applied the other optimisation techniques presented in Sec.~\ref{sec:optimization} (as described in the table above).
We note that the optimised Q1 and Q4 have the
same structure as the SQL queries in the original benchmark suite.
On the other hand, the optimised Q2 is different from the SQL query in BSBM  because the latter uses (\texttt{NOT}) \texttt{IN}, which is not considered in our optimisations.

Each query was executed three times with cold runs to avoid any variation due to caching.
The size of query answers and their running times (in secs) are as follows:
\\[6pt]
\centerline{\renewcommand{\tabcolsep}{3.5pt}\renewcommand{\arraystretch}{0.95}\small\begin{tabular}{c|r|rr|rr|rr|rr|rr} \toprule
      & \multicolumn{1}{c|}{\#} &  \multicolumn{2}{c|}{PostgreSQL} & \multicolumn{2}{c|}{MySQL} & \multicolumn{2}{c|}{X} & \multicolumn{2}{c|}{Y} & \multicolumn{2}{c}{Z}\\
                    query &answers & \multicolumn{1}{c}{N/O} & \multicolumn{1}{c|}{O} &  \multicolumn{1}{c}{N/O} & \multicolumn{1}{c|}{O} &  \multicolumn{1}{c}{N/O} & \multicolumn{1}{c|}{O} &  \multicolumn{1}{c}{N/O} & \multicolumn{1}{c|}{O}&  \multicolumn{1}{c}{N/O} & \multicolumn{1}{c}{O}\\
                    \midrule
 \rowcolor{lightgray!20}                    Q1	&	19,267	&	1.79	&	1.77	&	0.43	&	0.38	&	0.90	&	0.80	&	0.56	&	0.52	&	29.06	&	25.09	\\
                       Q2	&	6,746	&	18.75	&	2.07	&	19.95	&	0.36	&	40.00	&	16.07	&	0.44	&	0.37	&	27.99	&	5.97	\\
 Q2{\scriptsize \textsc{bsbm}} &  & &  3.88 & & 0.37 & & 20.55 & & 0.38 & & 5.91\\
\rowcolor{lightgray!20}                    Q3	&	1,355	&	4.20	&	0.09	&	4.70	&	0.11	&	5.50	&	1.60	&	2.04	&	0.14	&	5.45	&	0.65	\\
                    Q4	&	1,174	&	2.14	&	0.16	&	0.86	&	0.04	&	3.00	&	0.60	&	1.78	&	0.11	&	4.38	&	0.53	\\
  \rowcolor{lightgray!20}                   Q5	&	2,294	&	0.56	&	0.05	&	0.01	&	0.01	&	1.80	&	0.30	&	0.30	&	0.08	&	0.51	&	0.53	\\
                    Q6	&	2,294	&	102.35	&	0.18	&	>10{\scriptsize min}	&	0.04	&	1.90	&	0.40	&	4.50	&	0.14	&	0.82	&	0.54	\\
  \rowcolor{lightgray!20}                   Q7	&	2,294	&	102.00	&	0.17	&	>10{\scriptsize min}	&	0.04	&	2.60	&	0.40	&	14.57	&	0.14	&	1.21	&	0.53	\\
                    Q8	&	1,257	&	0.07	&	0.06	&	0.01	&	0.01	&	8.40	&	1.30	&	0.08	&	0.08	&	295.25	&	0.40	\\
   \rowcolor{lightgray!20}                  Q9	&	1,311	&	101.20	&	0.16	&	>10{\scriptsize min}	&	0.04	&	>10{\scriptsize min}	&	2.70	&	4.30	&	0.11	&	>10{\scriptsize min}	&	0.43	\\
                    Q10	&	1,331	&	103.30	&	0.15	&	>10{\scriptsize min}	&	0.05	&	>10{\scriptsize min}	&	4.20	&	5.20	&	0.14	&	>10{\scriptsize min}	&	0.43	\\
 \rowcolor{lightgray!20} Q11	&	3,388	&	5.26	&	0.87	&	3.80	&	0.21	&	107.06	&	2.68	&	177.95	&	0.22	&	7.82	&	0.13	\\															                    \bottomrule
  \end{tabular}}

\bigskip

\noindent The main outcomes of our experiments can be summarised as follows.
\begin{compactenum}[(a)]
\item The running times confirm that the optimisations are effective
  for all database engines. All
  optimised translations show better
  performance in all DB engines, and most of them can be evaluated in
  less than a second.
\item Interestingly, our optimised translation is even slightly 
  more efficient than the SQL
  with (\texttt{NOT}) \texttt{IN} from the original
  BSBM suite (see Q2{\scriptsize \textsc{bsbm}} in the table).
\item The effects of the optimisations are significant. In particular, for challenging
  queries (some of which time out after 10 mins), it can be up to three orders of magnitude.
\end{compactenum}


\section{Discussion and Conclusions}
\label{sec:conclusions}

The optimisation techniques we presented are intrinsic to SQL queries obtained by translating SPARQL in the context of OBDA with mappings,
and their novelty is due to the interaction of the components in the OBDA setting.
Indeed, the optimisation of \texttt{LEFT} \texttt{JOIN}s 
can be seen as a form of ``reasoning'' on the structure of the query, the
 data source and the mapping. For instance, 
when functional and inclusion dependencies along with attribute nullability are taken into account, 
one may infer that every tuple
from the left argument of a \texttt{LEFT} \texttt{JOIN} is guaranteed to match (\emph{i}) at least one
or (\emph{ii}) at most one tuple on the right. 
This information can allow one to replace \texttt{LEFT} \texttt{JOIN} by a simpler operator
such as an \texttt{INNER} \texttt{JOIN},  which can further be optimised by the
known techniques.

Observe that, in normal SQL queries, most of the
\texttt{NULL}s come from the database rather than from
operators like \texttt{LEFT} \texttt{JOIN}. In contrast, SPARQL triple patterns 
always bind their variables (no \texttt{NULL}s), and only operators like \texttt{OPTIONAL} can ``unbind''
them. 
In our experiments, we noticed that avoiding the padding effect is probably the most
effective outcome of the \texttt{LEFT} \texttt{JOIN} optimisation techniques in the
OBDA setting.

From the Semantic Web
perspective, our optimisations exploit information unavailable in RDF 
triplestores, namely, database integrity
constraints and mappings. From the DB perspective, we believe that such techniques
have not been developed because 
\texttt{LEFT} \texttt{JOIN}s and/or complex conditions like compatibility filters are not introduced accidentally in
expert-written SQL queries. The results of our
evaluation support this hypothesis and show a significant performance improvement, even for commercial DBMSs.

We are working on implementing these techniques in the  OBDA
system Ontop~\cite{CCKK*17}.


\noindent\textbf{Acknowledgements}
We thank the reviewers for their suggestions. This work
was supported by the OBATS project at the Free University of
Bozen-Bolzano and by the Euregio (EGTC) IPN12 project KAOS.

\bibliographystyle{abbrv}


\clearpage
\appendix

\newenvironment{theoremnum}[1]{\smallskip\noindent\textbf{Theorem~#1.}
  \hspace*{0.3em}\em}{\par\smallskip}
\newenvironment{lemmanum}[1]{\smallskip\noindent\textbf{Lemma~#1.}
  \hspace*{0.3em}\em}{\par\smallskip}

\section{Full Translation and Proof of Theorem~\ref{thm:sparql-sql}}\label{app:proof}

\paragraph{Syntax.} We consider \emph{graph patterns}, $P$, defined by the grammar
\begin{multline*}
  P \ ::= \
    B \ \mid \  \filter(P,F)  \ \mid \ \bind(P,v,c) \ \mid  \
    \union(P_1,P_2) \ \mid \join(P_1,P_2) \ \mid\\
    \leftjoin(P_1,P_2,F) \ \mid \ \minus(P_1, P_2) \ \mid \  \project(P,
    L) \ \mid \ \distinct(P),
\end{multline*}  
where $B$ is a BGP, $v\in \myVAR$ does not occur in $P$,
$c\in \sDOM$ is a constant,
$L\subseteq\myVAR$, and $F$, called \emph{filter}, is a formula constructed
using the logical connectives $\land$ and~$\neg$ from atoms of the form
$\textit{bound}(v)$, \mbox{$(v=c)$}, $(v=v')$, for $v,v' \in \myVAR$ and
$c \in \sDOM$, and possibly other built-in predicates.  The set of variables in
$P$ is denoted by $\var(P)$.  We assume (without mentioning it again) that all graph patterns 
of the form $\filter(P,F)$ are safe in the sense that
every variable in $F$ also occurs in $P$.

We do not consider solution modifiers other than $\distinct$ and $\project$; we 
also define a simplified variant of $\bind$, where $c$ is a constant rather than an (arithmetic)
expression (which are beyond the scope of the paper).  Our results, however, can easily be extended 
to the general form of $\bind$.

\paragraph{Semantics.} The semantics of SPARQL operations is defined as follows:
\begin{compactitem} 
\item $\filter(\Omega,F) = \Omega'$, where $\Omega'(s) = \Omega(s)$ if $s\in\Omega$ and $F^s = \top$, and $0$ otherwise;
\item $\bind(\Omega,v,c) = \Omega'$, where $\Omega'( s \smerge \{ v \mapsto c \}) = \Omega(s)$ if $s\in\Omega$, and $0$ otherwise;
\item $\union(\Omega_1, \Omega_2) = \Omega$, where  $\Omega(s) = \Omega_1(s) + \Omega_2(s)$;
\item $\join(\Omega_1,\Omega_2) = \Omega$, where
$\Omega(s) = \hspace*{-0.5em}\sum\limits_{\begin{subarray}{c}s_1\in \Omega_1, s_2 \in \Omega_2 \text{ with}\\s_1\sim s_2 \text{ and } s_1\smerge s_2 = s\end{subarray}} \hspace*{-1.5em}\Omega_1(s_1) \times \Omega_2(s_2)$;
\item $\leftjoin(\Omega_1, \Omega_2, F) = \union(\filter(\join(\Omega_1, \Omega_2), F), \Omega)$, where 
$\Omega(s) = \Omega_1(s)$ if $F^{s\smerge s_2} \ne \top$, for all $s_2\in \Omega_2$ compatible with $s$, and $0$ otherwise;
\item $\minus(\Omega_1,\Omega_2) = \Omega$, where $\Omega(s) = \Omega_1(s)$ 
if $\dom(s) \cap \dom(s_2) = \emptyset$, for all solution mappings~$s_2\in \Omega_2$ compatible with $s$, and $0$ otherwise;
\item $\project(\Omega, L) = \Omega'$, where $\Omega'(s') = \sum\limits_{s \in \Omega \text{ with }s|_L = s'} \hspace*{-1em}\Omega(s)$;
\item $\distinct(\Omega) = \Omega'$, where $\Omega'(s) = 1$ if $s \in \Omega$, and $0$ otherwise.
\end{compactitem}

\paragraph{Translation.} The translation of triple patterns depends on their shape:
\begin{equation*}
\tra(\langle s,p,o\rangle)  =  \begin{cases}%
\pi_\emptyset \sigma_{(\textit{subj} = s) \land (\textit{pred} = p) \land (\textit{obj} = o)}\,\textit{triple}, & \text{if } s,p,o\in \myIRI \cup \myLIT,\\ 
\proj{\{s \mapsto \textit{subj}\}}\, \sigma_{(\textit{pred} = p) \land (\textit{obj} = o)}\,\textit{triple}, & \text{if } s\in\myVAR \text{ and } p,o\in \myIRI \cup \myLIT,\\ 
\proj{\{s\mapsto\textit{subj},  \ o \mapsto \textit{obj}\}} \,\sigma_{\textit{pred} = p}\,\textit{triple}, & \text{if } s,o\in\myVAR, s \ne o, p\in \myIRI \cup \myLIT,\\ 
\proj{\{s\mapsto \textit{subj} \}} \,\sigma_{(\textit{pred} = p)\land (\textit{subj} = \textit{obj})}\,\textit{triple},\hspace*{-0.5em} & \text{if } s,o\in\myVAR, s = o,  p\in \myIRI \cup \myLIT,\\[-2pt]
\dots
\end{cases}
\end{equation*}
the remaining cases are similar. The translation of SPARQL operators is as follows, where the $\textit{tp}_i$ are triple patterns
\begin{align*}
  \tra (\{\textit{tp}_1, \textit{tp}_2, \dots, \textit{tp}_k\}) & = \ \ \tra (\textit{tp}_1) \Join  \tra (\textit{tp}_2) \Join \dots \Join  \tra (\textit{tp}_k),  \\
 \tra(\union(P_1,P_2)) &  \ = \ \tra(P_1) \ \uplus \ \tra(P_2),\\
 \tra(\filter(P_1,F_1)) & \ = \ \sigma_{\tra(F_1)} \tra(P_1), \\
  \tra(\bind(P_1,v,c)) & \ =  \ \nu_{v\mapsto c}\tra(P_1),\\
  \tra(\project(P,L)) & \ = \ \pi_L \tra(P), \\
  \tra(\distinct(P)) & \ = \ \delta \tra(P),\\
\tra(\join(P_1, P_2))  &\ =  \ \rho^{U^1\cup U^2}_{\coalesce_U} \bigl[\rho_{U^1/U} \tra(P_1) \Join_{\compatible_U} \rho_{U^2/U}\tra(P_2)\bigr], \\[2pt]
  \tra(\leftjoin(P_1, P_2, F))  &\ = \ \rho^{U^1 \cup U^2}_{\coalesce_U} \bigl[\rho_{U^1/U} \tra(P_1) \LJoin_{\compatible_U\wedge \tra(F)[\coalesce_U]}  \rho_{U^2/U} \tra(P_2)\bigr], \\[2pt]
  \tra(\minus(P_1, P_2))  &\ = \ \pi_{U_1}\rho_{U/U^1}  \sigma_{\isNull(w)}\\[-2pt] & \hspace*{4em}\bigl[\rho_{U^1/U} \tra(P_1) \LJoin_{\compatible_U\wedge \bigvee\limits_{u \in U}(u^1 = u^2)} \nu_{w\mapsto 1}\rho_{U^2/U} \tra(P_2)\bigr],    
\end{align*}

\noindent\textbf{Theorem~\ref{thm:sparql-sql}. \ }
\textit{For any RDF graph $G$ and any graph pattern $P$,  
$\| P\|_G = \|\tra(P)\|_{\textit{triple}(G)}.$}

\smallskip
\begin{proof} 
The proof is by induction on the structure of $P$. The basis of induction (for basic graph patterns) is straightforward.
The cases for $\union(P_1, P_2)$, $\filter(P_1, F)$, $\bind(P_1, v, c)$, $\project(P_1, L)$ and $\distinct(P_1)$ 
easily follow from the definitions and the induction hypothesis.
It remains to consider the  induction step for 
$P=\join(P_1,P_2)$,  $P=\leftjoin(P_1,P_2, F)$ and $P=\minus(P_1,P_2)$.
Let $U_i = \var(P_i)$, $i = 1,2$, and $U = U_1 \cap U_2$.

\smallskip

If $\|\join(P_1,P_2)\|_{G}(t) = m > 0$, then there is a unique solution mapping  $s$ such that $\extV[U_1\cup U_2](s) = t$ and $\sANS{\join(P_1,P_2)}{G}(s) = m$. By definition, $m$ is the sum of all $m_1 \cdot m_2$ such that $m_i = \sANS{P_i}{G}(s_i)$ for compatible $s_1$ and $s_2$ 
with $s_1 \smerge s_2 = s$. Consider any compatible $s_1$ and $s_2$ with $s_1 \smerge s_2 = s$.
By IH, we have $\|P_i\|_G(\extV[U_i](s_i)) = \|\tra(P_i) \|_{\textit{triple}(G)}(\extV[U_i](s_i))$.  
Since $s_1$ and $s_2$ are compatible, the structure of the filter and renaming operations in  $\tra(\join(P_1, P_2))$ guarantee that $\extV[U_1\cup U_2](s_1\oplus s_2)$ belongs to $\| \tra(\join(P_1, P_2))  \|_{\textit{triple}(G)}$ with multiplicity $\geq m_1 \cdot m_2$. It remains to observe that any  such $\extV[U_1\cup U_2](s_1\oplus s_2)$  coincides with $t$, and so $m \leq \| \tra(\join(P_1, P_2))  \|_{\textit{triple}(G)}(t)$.

Conversely, $\| \tra(\join(P_1,P_2)) \|_{\textit{triple}(G)}(t) = m > 0$, then, for $i = 1,2$, there are $t_i$ with $\|\tra(P_i)\|_{\textit{triple}(G)}(t_i) = m_i > 0$ and unique solution mappings $s_i$  such that $t_i = \extV[U_i](s_i)$, for $i = 1,2$, $s_1$ and $s_2$ are compatible (due to the filter in $\Join$) and $t = \extV[U_1\cup U_2](s_1\smerge s_2)$ (due to the renaming operations in $\tra(\join(P_1,P_2))$).
By IH, $\| P_i\|_{G}(\extV[U_i](s_i)) = m_i$ and so, $\sANS{P_i}{G}(s_i) = m_i$. Thus, 
$\sANS{\join(P_1,P_2)}{G}(s_1 \smerge s_2) \geq m_1 \cdot m_2$ and $\| \join(P_1,P_2) \|_G(t) \geq m$.

\bigskip

If $\|\leftjoin(P_1,P_2,F)\|_{G}(t) = m > 0$, then there is a unique $s$ with $\extV[U_1\cup U_2](s) = t$ and $\sANS{\leftjoin(P_1,P_2,F)}{G}(s) = m$. Then, by definition, $m$ is the sum of (\emph{a})~all $m_1\cdot m_2$ such that there are compatible $s_1$ and $s_2$ with $s_1 \smerge s_2 = s$, $F^{s} = \top$ and $\sANS{P_i}{G}(s_i) = m_i > 0$,  and (\emph{b})~$m' = \sANS{P_1}{G}(s) > 0$ in case there is no $s_2\in\sANS{P_2}{G}$ compatible with $s$ such that $F^{s\smerge s_2} = \top$. Item~(\emph{a}) is as the case of $\join$ (with an additional filter), so we consider only item~(\emph{b}). By IH, $\| \tra(P_1)\|_{\textit{triple}(G)}(\extV[U_1](s)) = m'$ and there is no $s_2$ such that $\extV[U_2](s_2)\in\|\tra(P_2)\|_{\textit{triple}(G)}$, $s$ and $s_2$ are compatible and $F^{s\smerge s_2} = \top$. Due to the shape of the condition, $\theta = \compatible_U\land\tra(F)[\coalesce_U]$, in the $\LJoin$ operation, the tuple $\extV[U_1](s)$ belongs to $\|\rho^{U^1\cup U^2}_{\coalesce_U} (\rho_{U^1/U}\tra(P_1) \setminus \pi_{U^1_1} (\rho_{U^1/U}\tra(P_1) \Join_{\theta} \rho_{U^2/U}\tra(P_2))) \|_{\textit{triple}(G)}$ with multiplicity $\geq m'$, where $U_1^1 = (U_1\setminus U)\cup U^1$; note that $\rho^{U^1\cup U^2}_{\coalesce_U}$ simply renames all $u^1$ into $u$, for $u\in U$, because all $u^2$ are removed by the projection. Therefore, $\extV[U_1\cup U_2](s)$ belongs to $\| \tra(\leftjoin(P_1, P_2,F))  \|_{\textit{triple}(G)}$ with multiplicity $\geq m'$. It remains to observe that the tuple $\extV[U_1\cup U_2](s)$ coincides with $t$, and so, $\| \tra(\leftjoin(P_1,P_2, F)) \|_{\textit{triple}(G)}(t)\geq m = \|\leftjoin(P_1,P_2,F)\|_{G}(t) $.

Conversely, if  $\| \tra(\leftjoin(P_1,P_2, F)) \|_{\textit{triple}(G)}(t) = m > 0$, then $m = m'' + m'$, where
\begin{itemize}\itemsep=0pt
\item[(\emph{a})]~$m'' = \|\rho^{U^1\cup U^2}_{\coalesce(U)}(\rho_{U^1/U}\tra(P_1) \Join_\theta  \rho_{U^2/U}(P_2))\|_{\textit{triple}(G)}(t)$, 
\item[(\emph{b})]~$m' = \|\rho^{U^1\cup U^2}_{\coalesce_U}(\rho_{U^1/U}\tra(P_1) \setminus \pi_{U_1^1} (\rho_{U^1/U}\tra(P_1) \Join_\theta  \rho_{U^2/U}\tra(P_2)))\|_{\textit{triple}(G)}(t)$,
\end{itemize}
$U_1^1$ denotes $(U_1\setminus U)\cup U^1$ and $\theta = \compatible_U\wedge \tra(F)[\coalesce_U]$.
Again, item~(\emph{a}) is identical to $\join$ with $\filter$ (we just point out that the $\tra(F)[\coalesce_U]$ component of the filter in $\Join$ can be pulled outside $\rho_{\coalesce_U}$ as $\sigma_{\tra(F)}$); so, we focus on item~(\emph{b}) only.\linebreak By construction, there is a unique solution mapping $s$\linebreak such that $\extV[U_1](s) = t$, $\|\rho^{U^1\cup U^2}_{\coalesce_U}\rho_{U^1/U}\tra(P_1)\|_{\textit{triple}(G)}(s) = m'$ but\linebreak \mbox{$\extV[U_1\cup U_2](s)\notin \| \rho^{U^1\cup U^2}_{\coalesce_U}\pi_{U_1^1} (\rho_{U^1/U}\tra(P_1) \Join_\theta  \rho_{U^2/U}\tra(P_2))\|_{\textit{triple}(G)}$}. The latter implies that there is no $s_2$ compatible with $s$ such that $F^{s\smerge s_2} = \top$ and $\extV[U_2](s_2)\in \| \rho^{U^1\cup U^2}_{\coalesce_U} \rho_{U^2/U}\tra(P_2)\|_{\textit{triple}(G)}$. It follows that $\|\tra(P_1)\|_{\textit{triple}(G)}(\extV[U_1](s)) \geq m'$ but there is no $s_2$ compatible with $s$ such that $F^{s\smerge s_2} = \top$ and  $\extV[U_2](s_2)\in\|\tra(P_2)\|_{\textit{triple}(G)}$. By IH and the definition of $\leftjoin$, $\sANS{\leftjoin(P_1,P_2,F)}{G}(s) \geq m'$. It follows that\linebreak $\|\leftjoin(P_1,P_2,F)\|_{G}(t)\geq m = \| \tra(\leftjoin(P_1,P_2, F)) \|_{\textit{triple}(G)}(t)$.

\bigskip

If $\|\minus(P_1,P_2)\|_{G}(t) = m > 0$, then there is a unique $s$ with $\extV[U_1\cup U_2](s) = t$ and $\sANS{\minus(P_1,P_2)}{G}(s) = m$. Then, by definition, $\sANS{P_1}{G}(s) = m > 0$ and there is no $s_2\in\sANS{P_2}{G}$ compatible with $s$ and such that $\dom(s)\cap \dom(s_2) \ne\emptyset$. By IH, $\| \tra(P_1)\|_{\textit{triple}(G)}(\extV[U_1](s)) = m$ and there is no $s_2$ compatible with $s$ such that $\extV[U_2](s_2)\in\|\tra(P_2)\|_{\textit{triple}(G)}$ and  $\dom(s)\cap \dom(s_2) \ne\emptyset$. Due to the shape of the filter, $\theta = \compatible_U\land\bigvee_{u\in U}(u^1 = u^2)$, in the $\LJoin$ operation and the extension $\nu_{w\mapsto 1}$ with the filter $\isNull(w)$, the tuple $\extV[U_1](s)$ belongs to\linebreak $\| \pi_{U_1}\rho_{U/U^1}  \sigma_{\isNull(w)}\bigl[\rho_{U^1/U} \tra(P_1) \LJoin_\theta \nu_{w\mapsto 1}\rho_{U^2/U} \tra(P_2)\bigr] \|_{\textit{triple}(G)}$ with multiplicity $m$. Therefore, $\extV[U_1](s)$ belongs to $\| \tra(\minus(P_1, P_2))  \|_{\textit{triple}(G)}$ with multiplicity $m$. It remains to observe that the tuple $\extV[U_1](s)$ coincides with $t$, and so,\linebreak $m = \| \tra(\minus(P_1,P_2)) \|_{\textit{triple}(G)}(t)$.

Conversely, if  $\| \tra(\minus(P_1,P_2)) \|_{\textit{triple}(G)}(t) = m > 0$,
by construction, there is a unique solution mapping $s$ such that $\extV[U_1](s) = t$, $\|\tra(P_1)\|_{\textit{triple}(G)}(s) = m$ but 
{$\extV[U_1\cup U_2 \cup U^1 \cup U^2 \cup \{w\}\setminus U](s)\notin \| \rho_{U^1/U} \tra(P_1) \Join_\theta \nu_{w\mapsto 1}\rho_{U^2/U} \tra(P_2))\|_{\textit{triple}(G)}$}, where  $\theta = \compatible_U\land\bigvee_{u\in U}(u^1 = u^2)$. It follows that there is no $s_2$ such that $\extV[U_2](s_2)\in \| \tra(P_2)\|_{\textit{triple}(G)}$ such that $s$ and $s_2$ are compatible and $\dom(s)\cap \dom(s_2) \ne\emptyset$. 
By IH and the definition of $\minus$, $\sANS{\minus(P_1,P_2)}{G}(s) = m$. 

\bigskip

This completes the proof of Theorem~\ref{thm:sparql-sql}.
\end{proof}


\newpage

\section{Section~\ref{sec:optimization} Proofs}\label{app:sec4}

\noindent\textbf{Proposition~\ref{prop:lj-nat}.\ }
\textit{For relations $R_1$ and $R_2$ over $U_1$ and $U_2$, respectively, and a filter $F$ over $U_1 \cup U_2$, we have
\begin{align}\label{eq:lj-nat}
R_1   \LJoin_F  R_2 \  \equiv \  \rho^{\{U_2\setminus U_1\}}_{\{u/\textit{if}(F, u, \Null) \,\mid\, u\in U_2\setminus U_1\}} (R_1 \LJoin R_2) && \text{ if } \ \ \pi_{U_1}(R_1 \Join R_2) \subseteq R_1.
\end{align}}  

\smallskip

\begin{proof}
Denote $R_1\Join R_2$ by $S$. Then $\pi_{U_1} S
\subseteq R_1$ implies that every tuple $t_1$ in $R_1$ can have at most
one tuple $t_2$ in $R_2$ compatible with it, and $S$ consists of all
such extensions (with their cardinality determined by $R_1$). Therefore, $\pi_{U_1}(S \setminus \sigma_F S)$ is precisely the
tuples in $R_1$ that cannot be extended in such a way that the extension  satisfies~$F$, whence 
\begin{align}
\tag{\ref{eq:LJ-nat:1}}
\pi_{U_1}(S \setminus \sigma_F S) & \ \ \equiv \ \ \pi_{U_1}S \setminus \pi_{U_1}\sigma_F S.
\end{align}
By a similar argument, $R_1 \setminus \pi_{U_1} S$ consists of the tuples in $R_1$ (with the same cardinality) that cannot be extended by a tuple in $R_2$, and $\pi_{U_1} S \setminus \pi_{U_1} \sigma_F S$ of those tuples that can be extended but only when $F$ is not satisfied. By taking the union of the two, we obtain
\begin{align}
\tag{\ref{eq:LJ-nat:2}}
(R_1 \setminus \pi_{U_1} S) \cup (\pi_{U_1} S \setminus \pi_{U_1} \sigma_F S) & \ \ \equiv \ \ R_1 \setminus  \pi_{U_1} \sigma_F S. 
\end{align}
Having these two equivalences at hand, we can now prove the claim:
\begin{align*}
\rho^{\{U_2\setminus U_1\}}_{\{\textit{if}(F, u, \Null)/u \,\mid\, u\in U_2\setminus U_1\}} (R_1 \LJoin R_2) \hspace*{0em}& \\
{\scriptstyle(\text{definition of LJ})} & \equiv \ \rho^{\{U_2\setminus U_1\}}_{\{\textit{if}(F, u, \Null)/u \,\mid\, u\in U_2\setminus U_1\}} \bigl[S \ \uplus \ (R_1 \setminus \pi_{U_1} S)\bigr] \\
{\scriptstyle(\rho \text{ distributes over } \cup \text{ and the definiton of } \textit{if})} & \equiv \  \rho^{\{U_2\setminus U_1\}}_{\{\textit{if}(F, u, \Null)/u \,\mid\, u\in U_2\setminus U_1\}} S \ \ \uplus \ \ (R_1 \setminus \pi_{U_1} S) \\
{\scriptstyle(\text{definitions of } \rho \text{ and }\textit{if})}& \equiv \ \sigma_F S \ \ \uplus \ \ \pi_{U_1} (S \setminus \sigma_F S) \ \ \uplus \ \ (R_1 \setminus \pi_{U_1} S) \\
{\scriptstyle\text{\eqref{eq:LJ-nat:1}}} & \equiv  \ \sigma_F S \ \ \uplus \ \  (\pi_{U_1} S \setminus \pi_{U_1} \sigma_F S) \ \uplus \  (R_1 \setminus \pi_{U_1} S) \\
{\scriptstyle(\mu \text{ distributes  over } \cup)} & \equiv  \ \sigma_F S \ \ \uplus \ \ \bigl[(\pi_{U_1} S \setminus \pi_{U_1} \sigma_F S) \ \cup \  (R_1 \setminus \pi_{U_1} S)\bigr] \\
{\scriptstyle\text{\eqref{eq:LJ-nat:2}}}  & \equiv  \  \sigma_F S \ \ \uplus \ \ (R_1 \setminus  \pi_{U_1} \sigma_F S)\\
{\scriptstyle(\text{definition of LJ})} & \equiv  \ R_1 \LJoin_F R_2.
\end{align*}
This completes the proof of Proposition~\ref{prop:lj-nat}.
\end{proof}

By combining Propositions~\ref{prop:lj-nat} and~\ref{prop:lj-red}, we obtain
\begin{corollary}
For relations $R_1$ and $R_2$ over $U_1$ and $U_2$, respectively, and a filter $F$ over~$U_1 \cup U_2$, we have
\begin{multline}\label{eq:lj-nat}
\sigma_{\neg\isNull(U_1 \cap U_2)} R_1   \LJoin_F  R_2 \  \equiv \  \rho^{\{U_2\setminus U_1\}}_{\{u/\textit{if}(F, u, \Null) \,\mid\, u\in U_2\setminus U_1\}} (R_1 \Join R_2),\\ \text{ if } \ \ \pi_{U_1}(R_1 \Join R_2) = R_1.
\end{multline}  
\end{corollary}
Note that the condition $\pi_{U_1}(R_1 \Join R_2) = R_1$ is
satisfied, in particular, when $R_1$ and~$R_2$ are both
(extended) projections of the same relation (with a primary key) or when $R_1$ has a foreign key
referencing (a primary or alternate key of) $R_2$.

\medskip

\noindent\textbf{Proposition~\ref{prop:transfer}.\ }
\textit{Let $R_1$, $R_2$ and $R_3$ be relations over $U_1$, $U_2$ and $U_3$, 
respectively, $F$ a filter over $U_1\cup U_2 \cup U_3$ and $w$ an attribute in $U_3 \setminus (U_1 \cup U_2)$. Then
\begin{multline}\label{eq:transfer}
R_1   \LJoin_F  (R_2 \Join \sigma_{\neg\isNull(w)} R_3) \ \ \equiv\\
\rho^{\{U_2\setminus U_1\}}_{\{u/\textit{if}(\neg\isNull(w), u, \Null)\ \mid \ u\in
    U_2\setminus U_1 \}} ((R_1 \Join R_2) \LJoin_F \sigma_{\neg\isNull(w)} R_3),\\ \text{ if } \pi_{U_1}(R_1 \Join R_2) = R_1.
\end{multline}}
\begin{proof}
Denote $R_1\Join R_2$ by $S$ and $\sigma_{\neg\isNull(w)} R_3$ by $R_3'$. First, we can easily establish the following analogue of~\eqref{eq:LJ-nat:1}:
\begin{align}\label{eq:LJ-nat:1:p}
\pi_{U_1}(S \setminus \pi_{U_1\cup U_2} (S\Join_F R_3')) & \ \ \equiv \ \ \pi_{U_1}S \setminus \pi_{U_1}(S \Join_F R_3'),
\end{align}
where $\pi_{U_1\cup U_2}(S\Join_F R_3')$ plays the same role as $\sigma_F S$ in~\eqref{eq:LJ-nat:1}.
Then the proof of the proposition is immediate from the following sequence of equivalences:
\begin{align*}
R_1 \LJoin_F (R_2 \Join R_3')\hspace*{-3em} \\
{\scriptstyle(\text{def. of LJ})} & \equiv \ \ \sigma_F (S \Join R_3') \ \ \uplus \ \ (R_1 \setminus \pi_{U_1}\sigma_F (S \Join R_3'))\\
{\scriptstyle(\text{assumption})} & \equiv \ \ \sigma_F (S  \Join R_3') \ \ \uplus \ \ (\pi_{U_1} S \setminus \pi_{U_1}\sigma_F (S \Join R_3'))\\
{\scriptstyle(\text{def. of join})} & \equiv \ \ (S \Join_F R_3')  \ \ \uplus \ \ (\pi_{U_1} S \setminus \pi_{U_1}(S \Join_F R_3'))\\
{\footnotesize\eqref{eq:LJ-nat:1:p}} & \equiv \ \ (S \Join_F R_3') \ \ \uplus \ \ \pi_{U_1} (S \setminus \pi_{U_1\cup U_2} (S\Join_F R_3'))\\
{\scriptstyle(\text{def. of } \mu)} & \equiv \ \ (S \Join_F R_3') \ \ \uplus \ \ \mu_{U_2\setminus U_1}\pi_{U_1} (S \setminus \pi_{U_1\cup U_2} (S\Join_F R_3'))\\
{\scriptstyle(\text{def. of \textit{if}})} & \equiv \ \ \rho^{\{U_2\setminus U_1\}}_{\{u/\textit{if}(\neg\isNull(w), u, \Null)\ \mid \ u\in
    U_2\setminus U_1 \}} [(S \Join_F R_3') \uplus (S \setminus \pi_{U_1\cup U_2} (S\Join_F R_3'))]\\
{\scriptstyle(\text{def. of LJ})} &  \equiv \ \ \rho^{\{U_2\setminus U_1\}}_{\{u/\textit{if}(\neg\isNull(w), u, \Null)\ \mid \ u\in
    U_2\setminus U_1 \}} (S \LJoin_F R_3').
\end{align*}
\end{proof}


\newpage 

\section{Complete Examples in Section~\ref{sec:optimization}}\label{app:examples}

\subsection{Example~\ref{ex:email-pref-simple}}\label{app:example1}

We begin by revisiting Example~\ref{ex:email-pref-simple}, which can now be given in algebraic form (for brevity, we ignore projecting away \texttt{?p}, which does not affect the optimisations):
\begin{multline*}
\leftjoin(\leftjoin(\texttt{?p :name ?n}, \ \ \texttt{?p :workEmail ?e}, \ \ \top), \texttt{?p :personalEmail ?e}, \ \ \top),
\end{multline*}
where $\top$ denotes the tautological filter (true). Suppose we have  the following mapping:
\begin{align*}
      \triple{\urione(\texttt{id})}{\texttt{:name}}{\texttt{fullName}} & \ \ \leftarrow\ \ 
                                   \sigma_{\notnull{\texttt{id}}  \land \notnull{\texttt{fullName}}}\texttt{people}, \\[-2pt]
      \triple{\urione(\texttt{id})}{\texttt{:workEmail}}{\texttt{workEmail}} & \ \ \leftarrow \ \
                                   \sigma_{\notnull{\texttt{id}}  \land\notnull{\texttt{workEmail}}}\,\texttt{people},\\[-2pt]
      \triple{\urione(\texttt{id})}{\texttt{:personalEmail}}{\texttt{homeEmail}} & \ \ \leftarrow\ \ 
                                   \sigma_{\notnull{\texttt{id}}  \land\notnull{\texttt{homeEmail}}}\,\texttt{people}, 
\end{align*}
where $\urione$ is a function that constructs the IRI for a person from their ID (an IRI template, in R2RML parlance). 

The translation given in Section~\ref{sec:translation}
and unfolding produce the following RA query, where we abbreviate, for example, $\rho^{\{p^1, p^2\}}_{\{p^4/\coalesce(p^1,p^2)\}}$ by $\bar{\rho}_{\{p^4/\coalesce(p^1,p^2)\}}$ (in other words, the $\bar{\rho}$ operation always projects away the arguments of its $\coalesce$ functions):\\[6pt]
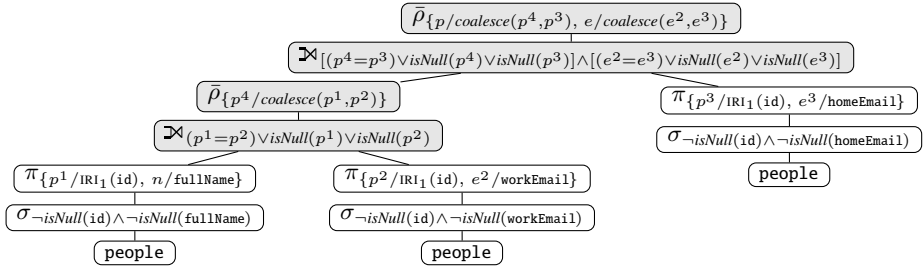
\centerline{\begin{tikzpicture}[xscale=2.9,yscale=0.52]\small
\node[nd,fill=white] (u3) at (4,2) {\footnotesize\texttt{people}};
\node[nd,fill=white] (u3s) at (4,2.9) {$\sigma_{\neg\isNull(\texttt{id})  \land\neg\isNull(\texttt{homeEmail})}$};
\node[nd,fill=white] (u3r) at (4,3.9) {$\pi_{\{p^3/\urione(\texttt{id}),\ e^3/\texttt{homeEmail}\}}$};
\node[nd,fill=white] (p2) at (2.5,0) {\footnotesize\texttt{people}};
\node[nd,fill=white] (p2s) at (2.5,0.9) {$\sigma_{\neg\isNull(\texttt{id})  \land \neg\isNull(\texttt{workEmail})}$};
\node[nd,fill=white] (p2r) at (2.5,1.9) {$\pi_{\{p^2/\urione(\texttt{id}),\ e^2/\texttt{workEmail}\}}$};
\node[nd,fill=white] (p1) at (1,0) {\footnotesize\texttt{people}};
\node[nd,fill=white] (p1s) at (1,0.9) {$\sigma_{\neg\isNull(\texttt{id})  \land\neg\isNull(\texttt{fullName})}$};
\node[nd,fill=white] (p1r) at (1,1.9) {$\pi_{\{p^1/\urione(\texttt{id}),\ n/\texttt{fullName}\}}$};
\node[nd] (lj1) at (1.75,3) {$\LJoin_{(p^1 = p^2)  \lor \isNull(p^1) \lor \isNull(p^2)}$}; 
\node[nd] (lj1p) at (1.75,4) {$\bar{\rho}_{\{p^4/\coalesce(p^1,p^2)\}}$};
\node[nd] (lj2) at (3,5) {$\LJoin_{ [(p^4 = p^3) \lor \isNull(p^4) \lor \isNull(p^3)]\land [(e^2 = e^3) \lor \isNull(e^2) \lor \isNull(e^3)]}$}; 
\node[nd] (lj2r) at (3,6) {$\bar{\rho}_{\{p/\coalesce(p^4,p^3), \ e/\coalesce(e^2,e^3)\}}$}; 
\draw (u3) -- (u3s);
\draw (u3s) -- (u3r);
\draw (p2) -- (p2s);
\draw (p2s) -- (p2r);
\draw (p1) -- (p1s);
\draw (p1s) -- (p1r);
\draw (p1r) -- (lj1);
\draw (p2r) -- (lj1);
\draw (lj1) -- (lj1p);
\draw (lj1p) -- (lj2);
\draw (u3r) -- (lj2);
\draw (lj2) -- (lj2r);
\end{tikzpicture}}\\[4pt]
In our diagrams, the white nodes are the contribution of the mapping and the translation of the basic graph patterns: for example, the basic graph pattern \texttt{?p :name ?n} produces $\smash{\pi_{\{p^1/\urione(\texttt{id}),\ n/\texttt{fullName}\}}\sigma_{\neg\isNull(\texttt{id})  \land\neg\isNull(\texttt{fullName})}\texttt{people}}$ (we use attributes without superscripts if there is only one occurrence; otherwise, the superscript identifies the relevant subquery).
The grey nodes correspond to the translation of the SPARQL operations: for instance, the innermost left join is on $\smash{\compatible_{\{p\}}}$ with $p$ renamed apart to $\smash{p^1}$ and~$\smash{p^2}$; the outermost left join is on $\smash{\compatible_{\{p, e\}}}$, where $p$ is renamed apart to $\smash{p^4}$ and $\smash{p^3}$ and $e$ to $\smash{e^2}$ and $\smash{e^3}$; the two $\bar{\rho}$ are the respective renaming operations with $\coalesce$.

By pulling the filters up for the first time using the standard database equivalences and~\eqref{eq:lj:filter-left:d}--\eqref{eq:lj:filter-right}, we obtain (the additions are shown in blue)\\[6pt]
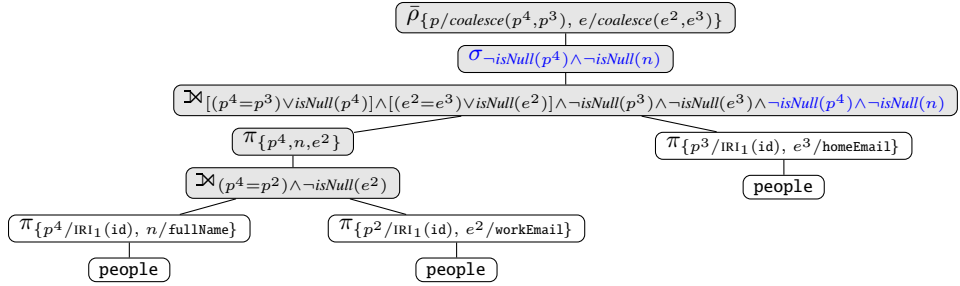
\centerline{\begin{tikzpicture}[xscale=2.9,yscale=0.52]\small
\node[nd,fill=white] (u3) at (3.75,2.9) {\footnotesize\texttt{people}};
\node[nd,fill=white] (u3r) at (3.75,3.9) {$\pi_{\{p^3/\urione(\texttt{id}),\ e^3/\texttt{homeEmail}\}}$};
\node[nd,fill=white] (p2) at (2.5,0.9) {\footnotesize\texttt{people}};
\node[nd,fill=white] (p2r) at (2.5,1.9) {$\pi_{\{p^2/\urione(\texttt{id}),\ e^2/\texttt{workEmail}\}}$};
\node[nd,fill=white] (p1) at (1,0.9) {\footnotesize\texttt{people}};
\node[nd,fill=white] (p1r) at (1,1.9) {$\pi_{\{p^1/\urione(\texttt{id}),\ n/\texttt{fullName}\}}$};
\node[nd] (lj1) at (1.75,3) {$\LJoin_{[(p^1 = p^2)  \lor \isNull(p^1) \lor \isNull(p^2)]\land\textcolor{blue}{\neg\isNull(p^1)  \land\neg\isNull(n) \land \neg\isNull(p^2)  \land \neg\isNull(e^2)}}$}; 
\node[nd] (lj1p) at (1.75,5) {$\bar{\rho}_{\{p^4/\coalesce(p^1,p^2)\}}$};
\node[nd] (lj1s) at (1.75,4) {\textcolor{blue}{$\sigma_{\neg\isNull(p^1)  \land\neg\isNull(n)}$}};
\node[nd] (lj2) at (2.75,6) {$\LJoin_{ [(p^4 = p^3) \lor \isNull(p^4) \lor \isNull(p^3)]\land [(e^2 = e^3) \lor \isNull(e^2) \lor \isNull(e^3)]\land\textcolor{blue}{\neg\isNull(p^3)  \land\neg\isNull(e^3)}}$}; 
\node[nd] (lj2r) at (2.75,7) {$\bar{\rho}_{\{p/\coalesce(p^4,p^3), \ e/\coalesce(e^2,e^3)\}}$}; 
%
\draw (u3) -- (u3r);
%
\draw (p2) -- (p2r);
%
\draw (p1) -- (p1r);
%
\draw (p1r) -- (lj1);
\draw (p2r) -- (lj1);
\draw (lj1) -- (lj1s);
\draw (lj1s) -- (lj1p);
\draw (lj1p) -- (lj2);
\draw (u3r) -- (lj2);
\draw (lj2) -- (lj2r);
\end{tikzpicture}}\\[2pt]
We then remove disjuncts from the join condition of the left joins by using~\eqref{eq:rem-disjunct}, with the following results (the eliminated disjuncts are in gray):
\begin{align*}
& [(p^1 = p^2)  \textcolor{gray}{{}\lor \isNull(p^1) \lor \isNull(p^2)}]\land {} \\[-2pt] & \hspace*{10em}\neg\isNull(p^1)  \land\neg\isNull(n) \land \neg\isNull(p^2)  \land \neg\isNull(e^2),\\
& [(p^4 = p^3) \lor \isNull(p^4) \textcolor{gray}{{}\lor \isNull(p^3)}]\land [(e^2 = e^3) \lor \isNull(e^2) \textcolor{gray}{{}\lor \isNull(e^3)}]\land{} \\[-2pt]&\hspace*{10em}\neg\isNull(p^3)  \land\neg\isNull(e^3),
\end{align*}
which are then simplified by~\eqref{eq:rem-isnull}, thus obtaining\\[6pt]
\centerline{\begin{tikzpicture}[xscale=2.9,yscale=0.55]\small
\node[nd,fill=white] (u3) at (4,2.9) {\footnotesize\texttt{people}};
\node[nd,fill=white] (u3r) at (4,3.9) {$\pi_{\{p^3/\urione(\texttt{id}),\ e^3/\texttt{homeEmail}\}}$};
\node[nd,fill=white] (p2) at (2.5,0.9) {\footnotesize\texttt{people}};
\node[nd,fill=white] (p2r) at (2.5,1.9) {$\pi_{\{p^2/\urione(\texttt{id}),\ e^2/\texttt{workEmail}\}}$};
\node[nd,fill=white] (p1) at (1,0.9) {\footnotesize\texttt{people}};
\node[nd,fill=white] (p1r) at (1,1.9) {$\pi_{\{p^1/\urione(\texttt{id}),\ n/\texttt{fullName}\}}$};
\node[nd] (lj1) at (1.75,3) {$\LJoin_{(p^1 = p^2)\land\neg\isNull(n) \land \neg\isNull(e^2)}$}; 
\node[nd] (lj1s) at (1.75,4) {$\sigma_{\neg\isNull(p^1)  \land\neg\isNull(n)}$};
\node[nd] (lj1p) at (1.75,5) {$\bar{\rho}_{\{p^4/\coalesce(p^1,p^2)\}}$};
\node[nd] (lj2) at (3,6) {$\LJoin_{ [(p^4 = p^3) \lor \isNull(p^4)]\land [(e^2 = e^3) \lor \isNull(e^2)]\land \neg\isNull(p^3)  \land\neg\isNull(e^3)}$}; 
\node[nd] (lj2r) at (3,7) {$\bar{\rho}_{\{p/\coalesce(p^4,p^3), \ e/\coalesce(e^2,e^3)\}}$}; 
%
\draw (u3) -- (u3r);
%
\draw (p2) -- (p2r);
%
\draw (p1) -- (p1r);
%
\draw (p1r) -- (lj1);
\draw (p2r) -- (lj1);
\draw (lj1) -- (lj1s);
\draw (lj1s) -- (lj1p);
\draw (lj1p) -- (lj2);
\draw (u3r) -- (lj2);
\draw (lj2) -- (lj2r);
\end{tikzpicture}}\\[6pt]
Then we apply~\eqref{eq:coalesce:elimination} to eliminate $\coalesce(p^1, p^2)$ with the following result (note that, by~\eqref{eq:lj:filter-left:d} and~\eqref{eq:lj:filter-left}, we also remove $\neg\isNull(n)$ from the join condition):\\[6pt]
\centerline{\begin{tikzpicture}[xscale=2.9,yscale=0.55]\small
\node[nd,fill=white] (u3) at (4,2.9) {\footnotesize\texttt{people}};
\node[nd,fill=white] (u3r) at (4,3.9) {$\pi_{\{p^3/\urione(\texttt{id}),\ e^3/\texttt{homeEmail}\}}$};
\node[nd,fill=white] (p2) at (2.5,0.9) {\footnotesize\texttt{people}};
\node[nd,fill=white] (p2r) at (2.5,1.9) {$\pi_{\{p^2/\urione(\texttt{id}),\ e^2/\texttt{workEmail}\}}$};
\node[nd,fill=white] (p1) at (1,0.9) {\footnotesize\texttt{people}};
\node[nd,fill=white] (p1r) at (1,1.9) {$\pi_{\{p^4/\urione(\texttt{id}),\ n/\texttt{fullName}\}}$};
\node[nd] (lj1) at (1.75,3) {$\LJoin_{(p^4 = p^2)\land \neg\isNull(e^2)}$}; 
\node[nd] (lj1p) at (1.75,4) {$\pi_{\{p^4, n, e^2 \}}$};
\node[nd] (lj1s) at (1.75,5) {$\sigma_{\neg\isNull(p^4) \land \neg\isNull(n)}$};
\node[nd] (lj2) at (3,6) {$\LJoin_{ [(p^4 = p^3) \lor \isNull(p^4)]\land [(e^2 = e^3) \lor \isNull(e^2)]\land \neg\isNull(p^3)  \land\neg\isNull(e^3)}$}; 
\node[nd] (lj2r) at (3,7) {$\bar{\rho}_{\{p/\coalesce(p^4,p^3), \ e/\coalesce(e^2,e^3)\}}$}; 
%
\draw (u3) -- (u3r);
%
\draw (p2) -- (p2r);
%
\draw (p1) -- (p1r);
%
\draw (p1r) -- (lj1);
\draw (p2r) -- (lj1);
\draw (lj1) -- (lj1p);
\draw (lj1p) -- (lj1s);
\draw (lj1s) -- (lj2);
\draw (u3r) -- (lj2);
\draw (lj2) -- (lj2r);
\end{tikzpicture}}\\[6pt]
We can now repeat the procedure of pulling up the filter with~\eqref{eq:lj:filter-left:d}--\eqref{eq:lj:filter-right} to get:\\[6pt]
\centerline{\begin{tikzpicture}[xscale=2.9,yscale=0.55]\small
\node[nd,fill=white] (u3) at (4,2.9) {\footnotesize\texttt{people}};
\node[nd,fill=white] (u3r) at (4,3.9) {$\pi_{\{p^3/\urione(\texttt{id}),\ e^3/\texttt{homeEmail}\}}$};
\node[nd,fill=white] (p2) at (2.5,0.9) {\footnotesize\texttt{people}};
\node[nd,fill=white] (p2r) at (2.5,1.9) {$\pi_{\{p^2/\urione(\texttt{id}),\ e^2/\texttt{workEmail}\}}$};
\node[nd,fill=white] (p1) at (1,0.9) {\footnotesize\texttt{people}};
\node[nd,fill=white] (p1r) at (1,1.9) {$\pi_{\{p^4/\urione(\texttt{id}),\ n/\texttt{fullName}\}}$};
\node[nd] (lj1) at (1.75,3) {$\LJoin_{(p^4 = p^2)\land \neg\isNull(e^2)}$}; 
\node[nd] (lj1p) at (1.75,4) {$\pi_{\{p^4, n, e^2 \}}$};
\node[nd] (lj2) at (3,5) {$\LJoin_{ [(p^4 = p^3) \lor \isNull(p^4)]\land [(e^2 = e^3) \lor \isNull(e^2)]\land \neg\isNull(p^3)  \land\neg\isNull(e^3) \land\textcolor{blue}{\neg\isNull(p^4) \land \neg\isNull(n)}}$}; 
\node[nd] (lj2r) at (3,7) {$\bar{\rho}_{\{p/\coalesce(p^4,p^3), \ e/\coalesce(e^2,e^3)\}}$}; 
\node[nd] (lj2s) at (3,6) {$\textcolor{blue}{\sigma_{\neg\isNull(p^4) \land \neg\isNull(n)}}$}; 
\draw (u3) -- (u3r);
%
\draw (p2) -- (p2r);
%
\draw (p1) -- (p1r);
%
\draw (p1r) -- (lj1);
\draw (p2r) -- (lj1);
\draw (lj1) -- (lj1p);
\draw (lj1p) -- (lj2);
\draw (u3r) -- (lj2);
\draw (lj2) -- (lj2s);
\draw (lj2s) -- (lj2r);
\end{tikzpicture}}\\[6pt]
Then we simplify the joining condition by~\eqref{eq:rem-disjunct} and~\eqref{eq:rem-isnull} with the following result:\\[6pt]
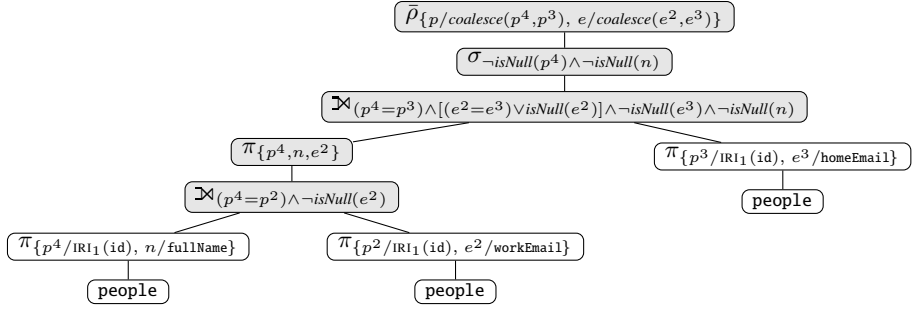
\centerline{\begin{tikzpicture}[xscale=2.9,yscale=0.6]\small
\node[nd,fill=white] (u3) at (4,2.9) {\footnotesize\texttt{people}};
\node[nd,fill=white] (u3r) at (4,3.9) {$\pi_{\{p^3/\urione(\texttt{id}),\ e^3/\texttt{homeEmail}\}}$};
\node[nd,fill=white] (p2) at (2.5,0.9) {\footnotesize\texttt{people}};
\node[nd,fill=white] (p2r) at (2.5,1.9) {$\pi_{\{p^2/\urione(\texttt{id}),\ e^2/\texttt{workEmail}\}}$};
\node[nd,fill=white] (p1) at (1,0.9) {\footnotesize\texttt{people}};
\node[nd,fill=white] (p1r) at (1,1.9) {$\pi_{\{p^4/\urione(\texttt{id}),\ n/\texttt{fullName}\}}$};
\node[nd] (lj1) at (1.75,3) {$\LJoin_{(p^4 = p^2)\land \neg\isNull(e^2)}$}; 
\node[nd] (lj1p) at (1.75,4) {$\pi_{\{p^4, n, e^2 \}}$};
\node[nd] (lj2) at (3,5) {$\LJoin_{(p^4 = p^3) \land [(e^2 = e^3) \lor \isNull(e^2)] \land\neg\isNull(e^3) \land \neg\isNull(n)}$}; 
\node[nd] (lj2r) at (3,7) {$\bar{\rho}_{\{p/\coalesce(p^4,p^3), \ e/\coalesce(e^2,e^3)\}}$}; 
\node[nd] (lj2s) at (3,6) {$\sigma_{\neg\isNull(p^4) \land \neg\isNull(n)}$}; 
\draw (u3) -- (u3r);
%
\draw (p2) -- (p2r);
%
\draw (p1) -- (p1r);
%
\draw (p1r) -- (lj1);
\draw (p2r) -- (lj1);
\draw (lj1) -- (lj1p);
\draw (lj1p) -- (lj2);
\draw (u3r) -- (lj2);
\draw (lj2) -- (lj2s);
\draw (lj2s) -- (lj2r);
\end{tikzpicture}}\\[6pt]
Now we apply~\eqref{eq:coalesce:elimination} to eliminate $\coalesce(p^4, p^3)$ (and, by~\eqref{eq:lj:filter-left:d} and~\eqref{eq:lj:filter-left}, remove $\neg\isNull(n)$ from the join condition) and obtain:\\[6pt]
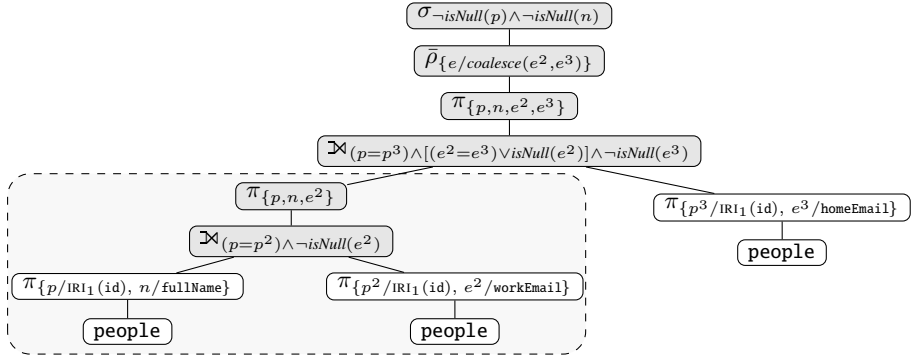
\centerline{\begin{tikzpicture}[xscale=2.9,yscale=0.6]\small
\draw[rounded corners=3mm,dashed,fill=black!3] (0.45,0.5) rectangle +(2.65,4);
\node[nd,fill=white] (u3) at (4,2.75) {\texttt{people}};
\node[nd,fill=white] (u3r) at (4,3.75) {$\pi_{\{p^3/\urione(\texttt{id}),\ e^3/\texttt{homeEmail}\}}$};
%
\node[nd,fill=white] (p2) at (2.5,1) {\texttt{people}};
\node[nd,fill=white] (p2r) at (2.5,2) {$\pi_{\{p^2/\urione(\texttt{id}),\ e^2/\texttt{workEmail}\}}$};
%
\node[nd,fill=white] (p1) at (1,1) {\texttt{people}};
\node[nd,fill=white] (p1r) at (1,2) {$\pi_{\{p/\urione(\texttt{id}),\ n/\texttt{fullName}\}}$};
%
\node[nd] (lj1) at (1.75,3) {$\LJoin_{(p = p^2) \land \neg\isNull(e^2)}$}; 
\node[nd] (lj1p) at (1.75,4) {$\pi_{\{p, n, e^2\}}$};
\node[nd] (lj2) at (2.75,5) {$\LJoin_{ (p = p^3)\land [(e^2 = e^3) \lor \isNull(e^2)] \land \neg\isNull(e^3)}$}; 
\node[nd] (lj2p) at (2.75,6) {$\pi_{\{p, n, e^2, e^3\}}$}; 
\node[nd] (lj2r) at (2.75,7) {$\bar{\rho}_{\{e/\coalesce(e^2,e^3)\}}$}; 
\node[nd] (lj2s) at (2.75,8) {$\sigma_{\neg\isNull(p)\land  \neg\isNull(n)}$}; 
\draw (u3) -- (u3r);
%
\draw (p2) -- (p2r);
%
\draw (p1) -- (p1r);
%
\draw (p1r) -- (lj1);
\draw (p2r) -- (lj1);
\draw (lj1) -- (lj1p);
\draw (lj1p) -- (lj2);
\draw (u3r) -- (lj2);
\draw (lj2) -- (lj2p);
\draw (lj2p) -- (lj2r);
\draw (lj2r) -- (lj2s);
\end{tikzpicture}}\\[6pt]
This completes application of the Compatibility Filters Reduction optimisation described in Section~\ref{sec:comp-filter-red}.

\medskip

By removing equalities with the help of~\eqref{eq:lj:natural} and~\eqref{eq:vacuous:nu}, we simplify the innermost left join (Left Join Naturalisation, Section~\ref{sec:lj:nat}) to obtain\\[6pt]
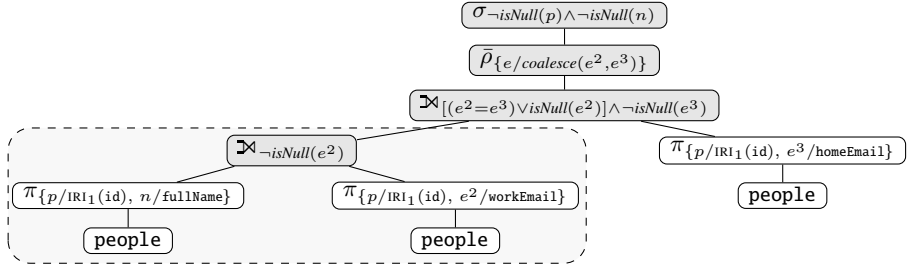
\centerline{\begin{tikzpicture}[xscale=2.9,yscale=0.6]\small
\draw[rounded corners=3mm,dashed,fill=black!3] (0.45,1.5) rectangle +(2.65,3);
\node[nd,fill=white] (u3) at (4,3) {\texttt{people}};
\node[nd,fill=white] (u3r) at (4,4) {$\pi_{\{p/\urione(\texttt{id}),\ e^3/\texttt{homeEmail}\}}$};
\node[nd,fill=white] (p2) at (2.5,2) {\texttt{people}};
\node[nd,fill=white] (p2r) at (2.5,3) {$\pi_{\{p/\urione(\texttt{id}),\ e^2/\texttt{workEmail}\}}$};
%
\node[nd,fill=white] (p1) at (1,2) {\texttt{people}};
\node[nd,fill=white] (p1r) at (1,3) {$\pi_{\{p/\urione(\texttt{id}),\ n/\texttt{fullName}\}}$};
%
\node[nd] (lj1) at (1.75,4) {$\LJoin_{\neg\isNull(e^2)}$}; 
\node[nd] (lj2) at (3,5) {$\LJoin_{[(e^2 = e^3) \lor \isNull(e^2)] \land \neg\isNull(e^3)}$}; 
\node[nd] (lj2s) at (3,7) {$\sigma_{\neg\isNull(p)\land  \neg\isNull(n)}$}; 
\node[nd] (lj2r) at (3,6) {$\bar{\rho}_{\{e/\coalesce(e^2,e^3)\}}$}; 
\draw (u3) -- (u3r);
%
\draw (p2) -- (p2r);
%
\draw (p1) -- (p1r);
%
\draw (p1r) -- (lj1);
\draw (p2r) -- (lj1);
\draw (lj1) -- (lj2);
\draw (u3r) -- (lj2);
\draw (lj2) -- (lj2r);
\draw (lj2r) -- (lj2s);
\end{tikzpicture}}\\[6pt]
Proposition~\ref{prop:lj-nat} is, in particular, applicable if the attributes shared by $R_1$ and~$R_2$ uniquely determine tuples of $R_2$. In our running example, $\texttt{id}$ is a primary key in $\texttt{people}$,  and so we can eliminate $\neg\isNull(e^2)$ from the innermost left join, which becomes a natural left join, and then simplify the term $\textit{if}(\neg\isNull(e^2), e^2, \Null)$ in the renaming to~$e^2$  by using equivalences~\eqref{eq:if:1} and~\eqref{eq:if:2} on complex terms. Thus, we effectively remove the renaming operator we have introduced by the application of Proposition~\ref{prop:lj-nat}:\\[6pt]
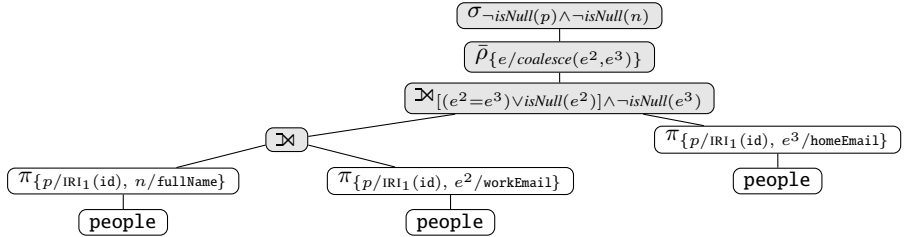
\centerline{\begin{tikzpicture}[xscale=2.9,yscale=0.55]\small
%
%
\node[nd,fill=white] (u3) at (4,3) {\texttt{people}};
\node[nd,fill=white] (u3r) at (4,4) {$\pi_{\{p/\urione(\texttt{id}),\ e^3/\texttt{homeEmail}\}}$};
\node[nd,fill=white] (p2) at (2.5,2) {\texttt{people}};
\node[nd,fill=white] (p2r) at (2.5,3) {$\pi_{\{p/\urione(\texttt{id}),\ e^2/\texttt{workEmail}\}}$};
%
\node[nd,fill=white] (p1) at (1,2) {\texttt{people}};
\node[nd,fill=white] (p1r) at (1,3) {$\pi_{\{p/\urione(\texttt{id}),\ n/\texttt{fullName}\}}$};
%
\node[nd] (lj1) at (1.75,4) {$\LJoin$}; 
\node[nd] (lj2) at (3,5) {$\LJoin_{[(e^2 = e^3) \lor \isNull(e^2)] \land \neg\isNull(e^3)}$}; 
\node[nd] (lj2s) at (3,7) {$\sigma_{\neg\isNull(p)\land  \neg\isNull(n)}$}; 
\node[nd] (lj2r) at (3,6) {$\bar{\rho}_{\{e/\coalesce(e^2,e^3)\}}$}; 
\draw (u3) -- (u3r);
%
\draw (p2) -- (p2r);
%
\draw (p1) -- (p1r);
%
\draw (p1r) -- (lj1);
\draw (p2r) -- (lj1);
\draw (lj1) -- (lj2);
\draw (u3r) -- (lj2);
\draw (lj2) -- (lj2r);
\draw (lj2r) -- (lj2s);
\end{tikzpicture}}

Next, by applying Proposition~\ref{prop:lj-red} (Natural Left Join Reduction, Section~\ref{sec:lj:red}), we can replace the natural left join by an inner join and then eliminate it (because it is on the primary key \texttt{id}). So, we arrive at\\[6pt]
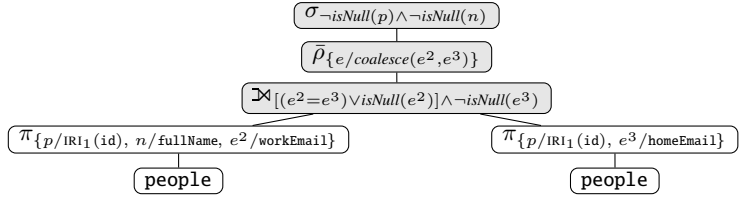
\centerline{\begin{tikzpicture}[xscale=2.9,yscale=0.55]\small
%
%
\node[nd,fill=white] (u3) at (4,3) {\texttt{people}};
\node[nd,fill=white] (u3r) at (4,4) {$\pi_{\{p/\urione(\texttt{id}),\ e^3/\texttt{homeEmail}\}}$};
%
\node[nd,fill=white] (p1) at (2,3) {\texttt{people}};
\node[nd,fill=white] (p1r) at (2,4) {$\pi_{\{p/\urione(\texttt{id}),\ n/\texttt{fullName}, \ e^2/\texttt{workEmail}\}}$};
%
\node[nd] (lj2) at (3,5) {$\LJoin_{[(e^2 = e^3) \lor \isNull(e^2)] \land \neg\isNull(e^3)}$}; 
\node[nd] (lj2s) at (3,7) {$\sigma_{\neg\isNull(p)\land  \neg\isNull(n)}$}; 
\node[nd] (lj2r) at (3,6) {$\bar{\rho}_{\{e/\coalesce(e^2,e^3)\}}$}; 
\draw (u3) -- (u3r);
%
%
\draw (p1) -- (p1r);
%
\draw (p1r) -- (lj2);
%
\draw (u3r) -- (lj2);
\draw (lj2) -- (lj2r);
\draw (lj2r) -- (lj2s);
\end{tikzpicture}}

To complete the running example, observe that now, by Proposition~\ref{prop:lj-nat},  the remaining left join can be replaced with a natural left join at the expense of introducing a renaming for $e^3$ with \mbox{$\textit{if}([(e^2 = e^3) \lor \isNull(e^2)] \land \neg\isNull(e^3), e^3, \Null)$}, which, by~\eqref{eq:if:1}, is equivalent to \mbox{$\textit{if}((e^2 = e^3) \lor \isNull(e^2), e^3, \Null)$}.
This renaming operation is then combined with $\bar{\rho}_{\{e / \coalesce(e^2, e^3)\}}$ to produce 
\mbox{$\textit{if}(\neg \isNull(e^2), e^2, \textit{if}((e^2 = e^3) \lor \isNull(e^2), e^3, \Null))$}. Next, we use equivalences
\begin{align}
\textit{if}(F_1 \lor F_2,  v_1, v_2)) & \ \equiv \  \textit{if}(F_1, v_1, \textit{if}(F_2, v_1, v_2)),\\
\textit{if}(F_1, \textit{if}(F_2, v_1, v_2), v_3) & \ \equiv \  \textit{if}(F_1 \land F_2, v_1, \textit{if}(F_1, v_2, v_3)),\\
\textit{if}(F_0, v_1, \textit{if}(F, v_2, v_3)) & \ \equiv \ \textit{if}(F_0, v_1, \textit{if}(F \land \neg F_0, v_2, v_3)), && \text{ if } F_0 \text{ is 2-valued},
\end{align}
to obtain $\coalesce(e^2, e^3) = \textit{if}(\neg\isNull(e^2), e^2, e^3)$; by a \emph{2-valued filter} we understand any filter that does not produce~$\varepsilon$, for example, $\neg\isNull(e^3)$. So, we obtain\\[6pt]
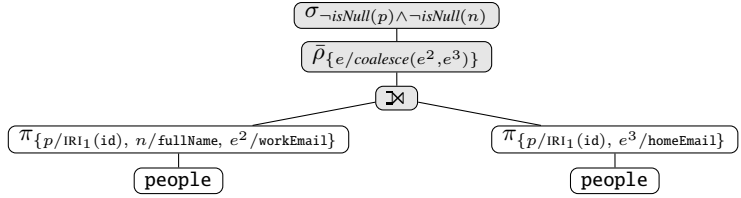
\centerline{\begin{tikzpicture}[xscale=2.9,yscale=0.55]\small
%
%
\node[nd,fill=white] (u3) at (4,3) {\texttt{people}};
\node[nd,fill=white] (u3r) at (4,4) {$\pi_{\{p/\urione(\texttt{id}),\ e^3/\texttt{homeEmail}\}}$};
%
\node[nd,fill=white] (p1) at (2,3) {\texttt{people}};
\node[nd,fill=white] (p1r) at (2,4) {$\pi_{\{p/\urione(\texttt{id}),\ n/\texttt{fullName}, \ e^2/\texttt{workEmail}\}}$};
%
\node[nd] (lj2) at (3,5) {$\LJoin$}; 
\node[nd] (lj2s) at (3,7) {$\sigma_{\neg\isNull(p)\land  \neg\isNull(n)}$}; 
\node[nd] (lj2r) at (3,6) {$\bar{\rho}_{\{e/\coalesce(e^2,e^3)\}}$}; 
\draw (u3) -- (u3r);
%
%
\draw (p1) -- (p1r);
%
\draw (p1r) -- (lj2);
%
\draw (u3r) -- (lj2);
\draw (lj2) -- (lj2r);
\draw (lj2r) -- (lj2s);
\end{tikzpicture}}

\bigskip

Finally, by applying Proposition~\ref{prop:lj-red} again, we can replace the natural left join by an inner join and then eliminate it:\\[0pt]
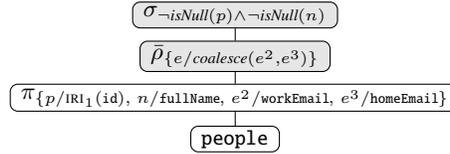
\centerline{\begin{tikzpicture}[xscale=2.9,yscale=0.55]\small
%
%
\node[nd,fill=white] (p1) at (3,3) {\texttt{people}};
\node[nd,fill=white] (p1r) at (3,4) {$\pi_{\{p/\urione(\texttt{id}),\ n/\texttt{fullName},\ e^2/\texttt{workEmail}, \ e^3/\texttt{homeEmail}\}}$};
%
\node[nd] (lj2s) at (3,6) {$\sigma_{\neg\isNull(p)\land  \neg\isNull(n)}$}; 
\node[nd] (lj2r) at (3,5) {$\bar{\rho}_{\{e/\coalesce(e^2,e^3)\}}$}; 
%
%
%
\draw (p1) -- (p1r);
%
%
\draw (p1r) -- (lj2r);
\draw (lj2r) -- (lj2s);
\end{tikzpicture}}\\ 
which can be simplified to $\pi_{\{p/\urione(\texttt{id}),\ n/\texttt{fullName},\ e/\coalesce(\texttt{workEmail},\ \texttt{homeEmail})\}} \texttt{people}$ because 
\texttt{id} and \texttt{fullName} are non-nullable in \texttt{people}, obtaining the following SQL query:
\begin{lstlisting}[language=SQLT]
SELECT fullName AS n, COALESCE(workEmail, homeEmail) AS e 
FROM people
\end{lstlisting}
(In this SQL, as well as in the subsequent variations of the example, we ignore variable~\texttt{?p} because it requires IRI construction, and we do not specify the IRI template.)

\subsection{Variation 1 on Example~\ref{ex:email-pref-simple}}\label{app:example1:discussion}

Consider now an extension of Example~\ref{ex:email-pref-simple} with another mapping for property \texttt{personalEmail}:
\begin{align*}
      \triple{\urione(\texttt{id})}{\texttt{:personalEmail}}{\texttt{homeEmail2}} & \quad\leftarrow\quad
                                   \sigma_{\notnull{\texttt{id}}  \land\notnull{\texttt{homeEmail2}}}\,\texttt{people2}, 
\end{align*}                                   
which uses another relation~\texttt{people2} with attributes \texttt{id} and \texttt{homeEmail2} (in a possibly different datasource). The additional mapping creates a union in the right-hand side of the outermost left join operation, which blocks an application of Proposition~\ref{prop:lj-nat} to this left join (a person can now have two personal e-mail addresses, and so, we cannot make the left join natural). The resulting SQL query is
\begin{lstlisting}[language=SQLT]
SELECT t1.fullName AS n, COALESCE(t1.workEmail, t2.e3) AS e 
FROM people t1 LEFT JOIN 
   (SELECT id, homeEmail AS e3 FROM people 
    UNION 
    SELECT id, homeEmail2 AS e3 FROM people2) t2
ON (t1.id = t2.id) AND 
   ((t1.workEmail = t2.e2) OR IS NULL(t1.workEmail)) AND 
   IS NOT NULL(t2.e3)
\end{lstlisting}

\newpage

\subsection{Variation 2 on Example~\ref{ex:email-pref-simple}}\label{app:example1:discussion2}

As another variant, consider again the setting of Example~\ref{ex:email-pref-simple}, where \texttt{workEmail} is a non-nullable attribute in \texttt{people} and there is a single mapping assertion given in Section~\ref{app:example1:discussion} for \texttt{:personalEmail}. Then the translated and unfolded query is as follows:\\[6pt]
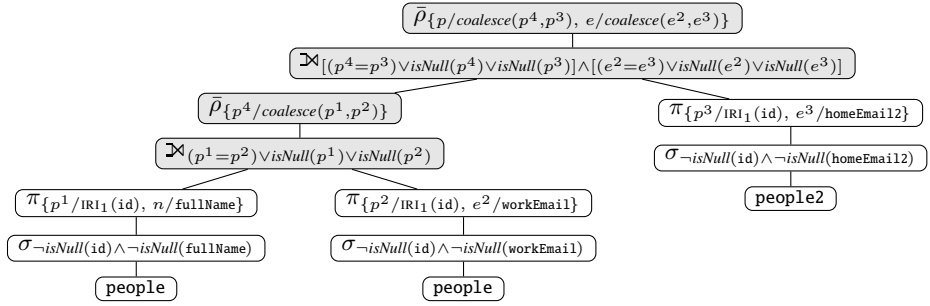
\centerline{\begin{tikzpicture}[xscale=2.9,yscale=0.6]\small
\node[nd,fill=white] (u3) at (4,2) {\footnotesize\texttt{people2}};
\node[nd,fill=white] (u3s) at (4,2.9) {$\sigma_{\neg\isNull(\texttt{id})  \land\neg\isNull(\texttt{homeEmail2})}$};
\node[nd,fill=white] (u3r) at (4,3.9) {$\pi_{\{p^3/\urione(\texttt{id}),\ e^3/\texttt{homeEmail2}\}}$};
\node[nd,fill=white] (p2) at (2.5,0) {\footnotesize\texttt{people}};
\node[nd,fill=white] (p2s) at (2.5,0.9) {$\sigma_{\neg\isNull(\texttt{id})  \land \neg\isNull(\texttt{workEmail})}$};
\node[nd,fill=white] (p2r) at (2.5,1.9) {$\pi_{\{p^2/\urione(\texttt{id}),\ e^2/\texttt{workEmail}\}}$};
\node[nd,fill=white] (p1) at (1,0) {\footnotesize\texttt{people}};
\node[nd,fill=white] (p1s) at (1,0.9) {$\sigma_{\neg\isNull(\texttt{id})  \land\neg\isNull(\texttt{fullName})}$};
\node[nd,fill=white] (p1r) at (1,1.9) {$\pi_{\{p^1/\urione(\texttt{id}),\ n/\texttt{fullName}\}}$};
\node[nd] (lj1) at (1.75,3) {$\LJoin_{(p^1 = p^2)  \lor \isNull(p^1) \lor \isNull(p^2)}$}; 
\node[nd] (lj1p) at (1.75,4) {$\bar{\rho}_{\{p^4/\coalesce(p^1,p^2)\}}$};
\node[nd] (lj2) at (3,5) {$\LJoin_{ [(p^4 = p^3) \lor \isNull(p^4) \lor \isNull(p^3)]\land [(e^2 = e^3) \lor \isNull(e^2) \lor \isNull(e^3)]}$}; 
\node[nd] (lj2r) at (3,6) {$\bar{\rho}_{\{p/\coalesce(p^4,p^3), \ e/\coalesce(e^2,e^3)\}}$}; 
\draw (u3) -- (u3s);
\draw (u3s) -- (u3r);
\draw (p2) -- (p2s);
\draw (p2s) -- (p2r);
\draw (p1) -- (p1s);
\draw (p1s) -- (p1r);
\draw (p1r) -- (lj1);
\draw (p2r) -- (lj1);
\draw (lj1) -- (lj1p);
\draw (lj1p) -- (lj2);
\draw (u3r) -- (lj2);
\draw (lj2) -- (lj2r);
\end{tikzpicture}}\\[6pt]
By following the same steps as in Section~\ref{app:example1}, we arrive at:\\[6pt]
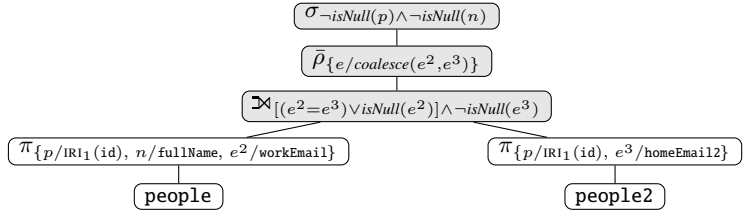
\centerline{\begin{tikzpicture}[xscale=2.9,yscale=0.6]\small
%
%
\node[nd,fill=white] (u3) at (4,3) {\texttt{people2}};
\node[nd,fill=white] (u3r) at (4,4) {$\pi_{\{p/\urione(\texttt{id}),\ e^3/\texttt{homeEmail2}\}}$};
%
\node[nd,fill=white] (p1) at (2,3) {\texttt{people}};
\node[nd,fill=white] (p1r) at (2,4) {$\pi_{\{p/\urione(\texttt{id}),\ n/\texttt{fullName}, \ e^2/\texttt{workEmail}\}}$};
%
\node[nd] (lj2) at (3,5) {$\LJoin_{[(e^2 = e^3) \lor \isNull(e^2)] \land \neg\isNull(e^3)}$}; 
\node[nd] (lj2s) at (3,7) {$\sigma_{\neg\isNull(p)\land  \neg\isNull(n)}$}; 
\node[nd] (lj2r) at (3,6) {$\bar{\rho}_{\{e/\coalesce(e^2,e^3)\}}$}; 
\draw (u3) -- (u3r);
%
%
\draw (p1) -- (p1r);
%
\draw (p1r) -- (lj2);
%
\draw (u3r) -- (lj2);
\draw (lj2) -- (lj2r);
\draw (lj2r) -- (lj2s);
\end{tikzpicture}}\\[6pt]
Now, since \texttt{workEmail} is  a non-nullable attribute in \texttt{people}, we can insert selection $\sigma_{\neg\isNull( \texttt{workEmail})}$ above \texttt{people} and then pull the filter up to the renaming operation $\bar{\rho}$, which, by~\eqref{eq:coalesce:elimination} simplifies to a projection:\\[6pt]
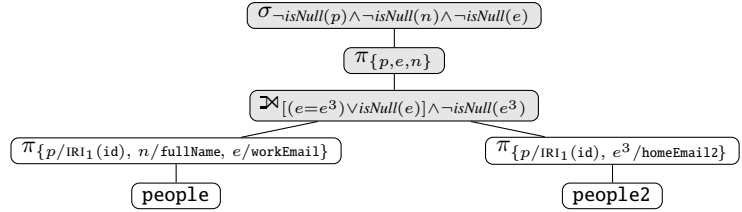
\centerline{\begin{tikzpicture}[xscale=2.9,yscale=0.6]\small
%
%
\node[nd,fill=white] (u3) at (4,3) {\texttt{people2}};
\node[nd,fill=white] (u3r) at (4,4) {$\pi_{\{p/\urione(\texttt{id}),\ e^3/\texttt{homeEmail2}\}}$};
%
\node[nd,fill=white] (p1) at (2,3) {\texttt{people}};
\node[nd,fill=white] (p1r) at (2,4) {$\pi_{\{p/\urione(\texttt{id}),\ n/\texttt{fullName}, \ e/\texttt{workEmail}\}}$};
%
\node[nd] (lj2) at (3,5) {$\LJoin_{[(e = e^3) \lor \isNull(e)] \land \neg\isNull(e^3)}$}; 
\node[nd] (lj2s) at (3,7) {$\sigma_{\neg\isNull(p)\land  \neg\isNull(n)\land\neg\isNull(e)}$}; 
\node[nd] (lj2r) at (3,6) {$\pi_{\{p, e, n\}}$}; 
\draw (u3) -- (u3r);
%
%
\draw (p1) -- (p1r);
%
\draw (p1r) -- (lj2);
%
\draw (u3r) -- (lj2);
\draw (lj2) -- (lj2r);
\draw (lj2r) -- (lj2s);
\end{tikzpicture}}\\[0pt]
We now apply~\eqref{eq:lj:to:simple} to eliminate the left join together with its right-hand side argument:\\[6pt]
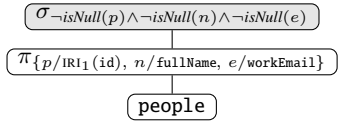
\centerline{\begin{tikzpicture}[xscale=2.9,yscale=0.6]\small
%
%
%
\node[nd,fill=white] (p1) at (2,3) {\texttt{people}};
\node[nd,fill=white] (p1r) at (2,4) {$\pi_{\{p/\urione(\texttt{id}),\ n/\texttt{fullName}, \ e/\texttt{workEmail}\}}$};
%
\node[nd] (lj2s) at (2,5) {$\sigma_{\neg\isNull(p)\land  \neg\isNull(n)\land\neg\isNull(e)}$}; 
%
%
\draw (p1) -- (p1r);
%
\draw (p1r) -- (lj2s);
%
\end{tikzpicture}}\\[6pt]
Thus, we obtain the following SQL query:
\begin{lstlisting}[language=SQLT]
SELECT fullName AS n, workEmail AS e FROM people
\end{lstlisting}

\newpage

\subsection{Example of Join Transfer}\label{app:join:transfer}

To illustrate an application of Proposition~\ref{prop:transfer}, we need an extension of table \texttt{people} with a nullable attribute \texttt{spouseId}, which contains the \texttt{id} of the person's spouse if they are married and \texttt{NULL} otherwise. The attribute is mapped by the following additional assertion:
\begin{equation*}
      \triple{\urione(\texttt{id})}{\texttt{:hasSpouse}}{\urione(\texttt{spouseId})} \quad\leftarrow\quad
                                   \sigma_{\notnull{\texttt{id}}  \land \notnull{\texttt{spouseId}}}\texttt{people}.
\end{equation*}
Consider now the following query in SPARQL algebra: 
\begin{equation*}
\project(\leftjoin(\texttt{?p :name ?n}, \ \join(\texttt{?p :hasSpouse ?s},\ \texttt{?s :name ?sn}), \ \top), \ \{\,\texttt{?n}, \texttt{?sn} \,\}),
\end{equation*}
which is translated and unfolded into the following RA query:\\[6pt]
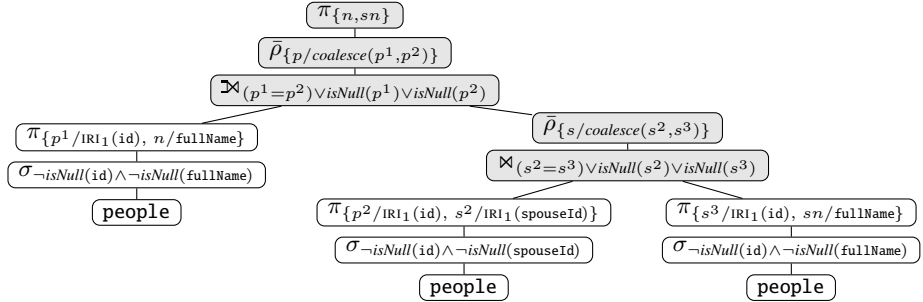
\centerline{\begin{tikzpicture}[xscale=2.9,yscale=0.5]\small
\node[nd,fill=white] (p3) at (4,0) {\texttt{people}};
\node[nd,fill=white] (p3s) at (4,1) {$\sigma_{\neg\isNull(\texttt{id}) \land \neg\isNull(\texttt{fullName})}$};
\node[nd,fill=white] (p3r) at (4,2) {$\pi_{\{s^3/\urione(\texttt{id}),\ sn/\texttt{fullName}\}}$};
\node[nd,fill=white] (p2) at (2.5,0) {\texttt{people}};
\node[nd,fill=white] (p2s) at (2.5,1) {$\sigma_{\neg\isNull(\texttt{id}) \land \neg \isNull(\texttt{spouseId})}$};
\node[nd,fill=white] (p2r) at (2.5,2) {$\pi_{\{p^2/\urione(\texttt{id}),\ s^2/\urione(\texttt{spouseId})\}}$};
\node[nd,fill=white] (p1) at (1,2) {\texttt{people}};
\node[nd,fill=white] (p1s) at (1,3) {$\sigma_{\neg\isNull(\texttt{id}) \land \neg\isNull(\texttt{fullName})}$};
\node[nd,fill=white] (p1r) at (1,4) {$\pi_{\{p^1/\urione(\texttt{id}),\ n/\texttt{fullName}\}}$};
\node[nd] (j) at (3.25,3.25) {$\Join_{(s^2 = s^3) \lor \isNull(s^2) \lor \isNull(s^3)}$}; 
\node[nd] (jr) at (3.25,4.25) {$\bar{\rho}_{\{s/\coalesce(s^2, s^3)\}}$};
\node[nd] (lj) at (2,5.25) {$\LJoin_{(p^1 = p^2) \lor \isNull(p^1) \lor \isNull(p^2)}$};
\node[nd] (ljr) at (2,6.25) {$\bar{\rho}_{\{p/\coalesce(p^1, p^2)\}}$};  
\node[nd] (mp) at (2,7.25) {$\pi_{\{n, sn\}}$};  
\draw (p3) -- (p3s);
\draw (p3s) -- (p3r);
\draw (p2) -- (p2s);
\draw (p2s) -- (p2r);
\draw (p1) -- (p1s);
\draw (p1s) -- (p1r);
\draw (p2r) -- (j);
\draw (p3r) -- (j);
\draw (j) -- (jr);
\draw (jr) -- (lj);
\draw (p1r) -- (lj);
\draw (lj) -- (ljr);
\draw (ljr) -- (mp);
\end{tikzpicture}}\\ 
By pulling up the filters, removing $\coalesce$ and eliminating $\neg\isNull(s^3)$ and~$\neg\isNull(p^2)$ from the filters of left joins (as described in Section~\ref{sec:comp-filter-red}), we obtain\\[6pt]
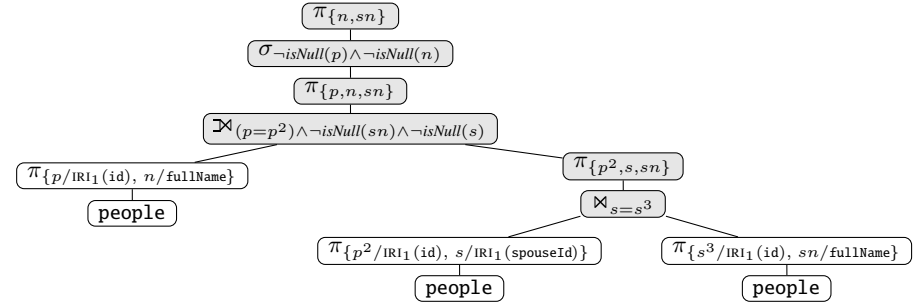
\centerline{\begin{tikzpicture}[xscale=2.9,yscale=0.5]\small
\node[nd,fill=white] (p3) at (4,1) {\texttt{people}};
\node[nd,fill=white] (p3r) at (4,2) {$\pi_{\{s^3/\urione(\texttt{id}),\ sn/\texttt{fullName}\}}$};
\node[nd,fill=white] (p2) at (2.5,1) {\texttt{people}};
\node[nd,fill=white] (p2r) at (2.5,2) {$\pi_{\{p^2/\urione(\texttt{id}),\ s/\urione(\texttt{spouseId})\}}$};
\node[nd,fill=white] (p1) at (1,3) {\texttt{people}};
\node[nd,fill=white] (p1r) at (1,4) {$\pi_{\{p/\urione(\texttt{id}),\ n/\texttt{fullName}\}}$};
\node[nd] (j) at (3.25,3.25) {$\Join_{s = s^3}$};  
\node[nd] (jr) at (3.25,4.25) {$\pi_{\{p^2, s, sn\}}$};
\node[nd] (lj) at (2,5.25) {$\LJoin_{(p = p^2) \land \neg\isNull(sn) \land  \neg\isNull(s)}$}; 
\node[nd] (ljr) at (2,6.25) {$\pi_{\{p, n, sn\}}$};  
\node[nd] (ljs) at (2,7.25) {$\sigma_{\neg\isNull(p) \land \neg \isNull(n)}$};  
\node[nd] (mp) at (2,8.25) {$\pi_{\{n, sn\}}$};  
\draw (p3) -- (p3r);
%
\draw (p2) -- (p2r);
%
\draw (p1) -- (p1r);
%
\draw (p2r) -- (j);
\draw (p3r) -- (j);
\draw (j) -- (jr);
\draw (jr) -- (lj);
\draw (p1r) -- (lj);
\draw (lj) -- (ljr);
\draw (ljr) -- (ljs);
\draw (ljs) -- (mp);
\end{tikzpicture}}\\[2pt]
We then turn the inner join into a natural inner join by renaming $s^3$ into $s$ using the inner-join counterparts of~\eqref{eq:lj:natural} and~~\eqref{eq:vacuous:nu}:\\[4pt]
%
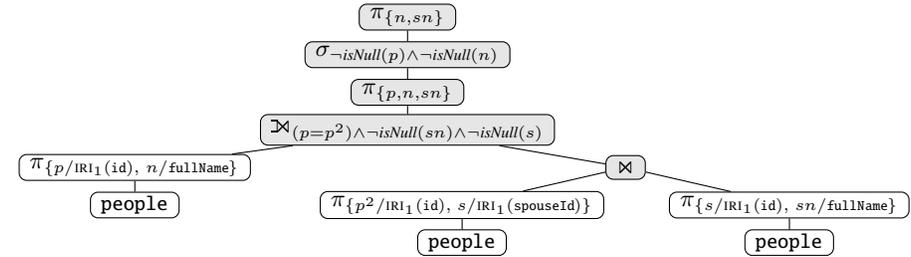
\centerline{\begin{tikzpicture}[xscale=2.9,yscale=0.5]\small
\node[nd,fill=white] (p3) at (4,2) {\texttt{people}};
\node[nd,fill=white] (p3r) at (4,3) {$\pi_{\{s/\urione(\texttt{id}),\ sn/\texttt{fullName}\}}$};
\node[nd,fill=white] (p2) at (2.5,2) {\texttt{people}};
\node[nd,fill=white] (p2r) at (2.5,3) {$\pi_{\{p^2/\urione(\texttt{id}),\ s/\urione(\texttt{spouseId})\}}$};
\node[nd,fill=white] (p1) at (1,3) {\texttt{people}};
\node[nd,fill=white] (p1r) at (1,4) {$\pi_{\{p/\urione(\texttt{id}),\ n/\texttt{fullName}\}}$};
\node[nd] (j) at (3.25,4) {$\Join$}; 
%
\node[nd] (lj) at (2.25,5) {$\LJoin_{(p = p^2) \land \neg\isNull(sn) \land  \neg\isNull(s)}$}; 
\node[nd] (ljr) at (2.25,6) {$\pi_{\{p, n, sn\}}$};  
\node[nd] (ljs) at (2.25,7) {$\sigma_{\neg\isNull(p) \land \neg \isNull(n)}$};  
\node[nd] (mp) at (2.25,8) {$\pi_{\{n, sn\}}$};  
\draw (p3) -- (p3r);
%
\draw (p2) -- (p2r);
%
\draw (p1) -- (p1r);
%
\draw (p2r) -- (j);
\draw (p3r) -- (j);
%
\draw (j) -- (lj);
\draw (p1r) -- (lj);
\draw (lj) -- (ljr);
\draw (ljr) -- (ljs);
\draw (ljs) -- (mp);
\end{tikzpicture}}\\ 
Then we similarly turn the left join into a natural left join and push down the filter on~$\neg\isNull(sn)$ using~\eqref{eq:lj:filter-right} as well as its trivial inner-join counterpart:\\[6pt]
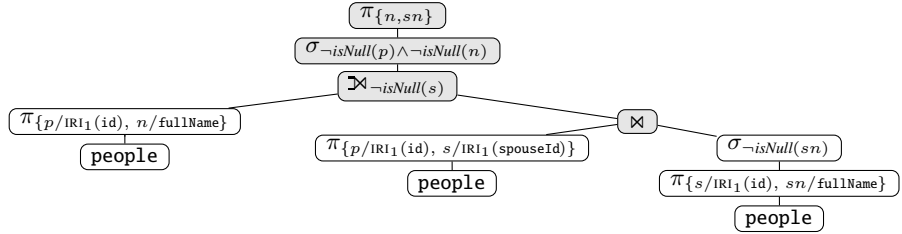
\centerline{\begin{tikzpicture}[xscale=2.9,yscale=0.47]\small
\node[nd,fill=white] (p3) at (4,1.25) {\texttt{people}};
\node[nd,fill=white] (p3r) at (4,2.25) {$\pi_{\{s/\urione(\texttt{id}),\ sn/\texttt{fullName}\}}$};
\node[nd,fill=white] (p3s) at (4,3.25) {$\sigma_{\neg\isNull(sn)}$};
%
\node[nd,fill=white] (p2) at (2.5,2.25) {\texttt{people}};
\node[nd,fill=white] (p2r) at (2.5,3.25) {$\pi_{\{p/\urione(\texttt{id}),\ s/\urione(\texttt{spouseId})\}}$};
%
\node[nd,fill=white] (p1) at (1,3) {\texttt{people}};
\node[nd,fill=white] (p1r) at (1,4) {$\pi_{\{p/\urione(\texttt{id}),\ n/\texttt{fullName}\}}$};
%
\node[nd] (j) at (3.35,4) {$\Join$}; 
%
\node[nd] (lj) at (2.25,5) {$\LJoin_{\neg\isNull(s)}$};
\node[nd] (ljs) at (2.25,6) {$\sigma_{\neg\isNull(p) \land \neg \isNull(n)}$};  
\node[nd] (mp) at (2.25,7) {$\pi_{\{n, sn\}}$};  
\draw (p3) -- (p3r);
\draw (p3r) -- (p3s);
\draw (p2) -- (p2r);
%
\draw (p1) -- (p1r);
\draw (p1r) -- (p1s);
\draw (p2r) -- (j);
\draw (p3s) -- (j);
%
\draw (j) -- (lj);
\draw (p1r) -- (lj);
\draw (lj) -- (ljs);
\draw (ljs) -- (mp);
\end{tikzpicture}}\\[6pt]
Observe that the inner join cannot be eliminated using the standard self-join elimination techniques because it is not on a primary (or alternate) key. 
However, by Proposition~\ref{prop:transfer},  we take $sn$ as the non-nullable attribute $w$ and obtain the following:\\[6pt]
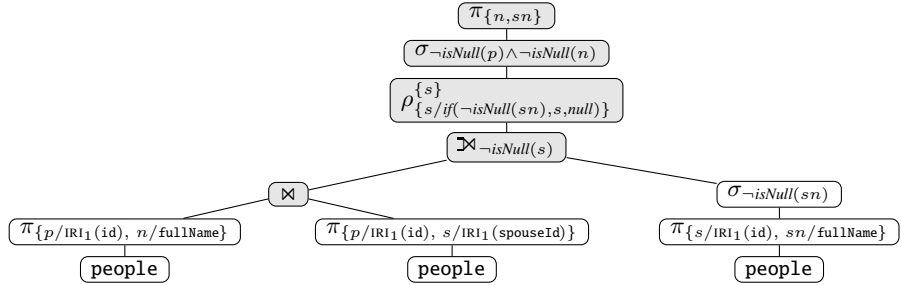
\centerline{\begin{tikzpicture}[xscale=2.9,yscale=0.5]\small
\node[nd,fill=white] (p3) at (4,2) {\texttt{people}};
\node[nd,fill=white] (p3r) at (4,3) {$\pi_{\{s/\urione(\texttt{id}),\ sn/\texttt{fullName}\}}$};
\node[nd,fill=white] (p3s) at (4,4) {$\sigma_{\neg\isNull(sn)}$};
\node[nd,fill=white] (p2) at (2.5,2) {\texttt{people}};
\node[nd,fill=white] (p2r) at (2.5,3) {$\pi_{\{p/\urione(\texttt{id}),\ s/\urione(\texttt{spouseId})\}}$};
%
\node[nd,fill=white] (p1) at (1,2) {\texttt{people}};
\node[nd,fill=white] (p1r) at (1,3) {$\pi_{\{p/\urione(\texttt{id}),\ n/\texttt{fullName}\}}$};
%
\node[nd] (j) at (1.75,4) {$\Join$}; 
%
\node[nd] (lj) at (2.75,5.25) {$\LJoin_{\neg\isNull(s)}$};
\node[nd] (ljr) at (2.75,6.5) {$\rho^{\{s\}}_{\{ s / \textit{if}(\neg\isNull(sn), s, \Null) \}}$};  
\node[nd] (ljs) at (2.75,7.75) {$\sigma_{\neg\isNull(p) \land \neg \isNull(n)}$};  
\node[nd] (mp) at (2.75,8.75) {$\pi_{\{n, sn\}}$};  
\draw (p3) -- (p3r);
\draw (p3r) -- (p3s);
\draw (p2) -- (p2r);
%
\draw (p1) -- (p1r);
\draw (p1r) -- (p1s);
\draw (p1r) -- (j);
\draw (p2r) -- (j);
%
\draw (j) -- (lj);
\draw (p3s) -- (lj);
\draw (lj) -- (ljr);
\draw (ljr) -- (ljs);
\draw (ljs) -- (mp);
\end{tikzpicture}}\\[6pt]
Now, the inner self-join can be eliminated (because \texttt{id} is the primary key of \texttt{people}); we can also remove the $\rho$ operation (because its result is projected away) and, by~\eqref{eq:lj:filter-left:d} and \eqref{eq:lj:filter-right}, push down the filters $\neg\isNull(s)$ and $\neg \isNull(p) \land \neg\isNull(n)$ to the right- and left-hand side arguments of the left join, respectively, to obtain\\[6pt]
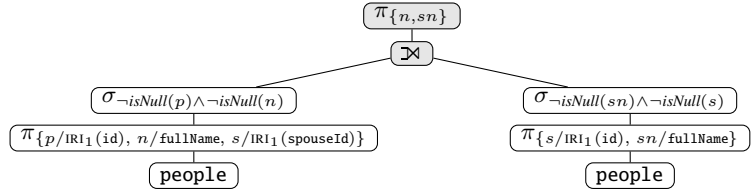
\centerline{\begin{tikzpicture}[xscale=2.9,yscale=0.5]\small
\node[nd,fill=white] (p3) at (3,1) {\texttt{people}};
\node[nd,fill=white] (p3r) at (3,2) {$\pi_{\{s/\urione(\texttt{id}),\ sn/\texttt{fullName}\}}$};
\node[nd,fill=white] (p3s) at (3,3) {$\sigma_{\neg \isNull(sn)\land  \neg\isNull(s)}$};
%
%
\node[nd,fill=white] (p1) at (1,1) {\texttt{people}};
\node[nd,fill=white] (p1r) at (1,2) {$\pi_{\{p/\urione(\texttt{id}),\ n/\texttt{fullName}, \ s/\urione(\texttt{spouseId})\}}$};
\node[nd,fill=white] (p1s) at (1,3) {$\sigma_{\neg\isNull(p) \land \neg \isNull(n)}$};
%
%
\node[nd] (lj) at (2,4.25) {$\LJoin$};
\node[nd] (mp) at (2,5.25) {$\pi_{\{n, sn\}}$};  
\draw (p3) -- (p3r);
\draw (p3r) -- (p3s);
%
%
\draw (p1) -- (p1r);
\draw (p1r) -- (p1s);
%
%
\draw (p1s) -- (lj);
\draw (p3s) -- (lj);
\draw (lj) -- (mp);
\end{tikzpicture}}\\[6pt]
Now, the filters can all be removed because all the respective attributes are non-nullable, which gives us the following SQL:
\begin{lstlisting}[language=SQLT]
SELECT p1.fullName AS n, p2.fullName AS sn
FROM people p1 LEFT JOIN people p2 ON p1.spouseId = p2.id
\end{lstlisting}

\newpage

\subsection{Example in Section~\ref{sec:lj:simpl}}\label{app:vertical}

We consider the following SPARQL query:
\begin{align*}
&\filter(\leftjoin(\leftjoin(\texttt{?p a :Product}, \\
&\hspace*{2em} \filter(\texttt{\{ ?p :hasReview ?r . ?r :hasLang ?l \}}, \texttt{?l} = \texttt{"en"}), \ \top),\\ 
&\hspace*{2em} \filter(\texttt{\{ ?p :hasReview ?r . ?r :hasLang ?l \}}, \texttt{?l} = \texttt{"zh"}), \ \top), \ \textit{bound}(\texttt{?r})).
\end{align*}
Note that both filters, \texttt{?l} = \texttt{"en"} and \texttt{?l} = \texttt{"zh"}, could in fact be moved to the third argument of $\leftjoin$, thus replacing $\top$.
In the database, we assume that  attribute \texttt{pid} of relation \texttt{review} is a foreign key referencing the primary key \texttt{pid} in relation \texttt{product}.
The following mapping connects the database to the ontology: 
\begin{align*}
      \triple{\urione(\texttt{pid})}{\texttt{a}}{\urione(\texttt{:Product})} &\quad\leftarrow\quad
                                   \sigma_{\notnull{\texttt{pid}}}\texttt{product},\\
      \triple{\urione(\texttt{pid})}{\texttt{:hasReview}}{\uritwo(\texttt{rid})} &\quad\leftarrow\quad
                                   \sigma_{\notnull{\texttt{pid}}\land\notnull{\texttt{rid}}}\texttt{review},\\
      \triple{\uritwo(\texttt{rid})}{\texttt{:hasLang}}{\texttt{lang}} &\quad\leftarrow\quad
                                   \sigma_{\notnull{\texttt{rid}}\land\notnull{\texttt{lang}}}\texttt{review}.
\end{align*}
We first translate the SPARQL query into SQL and apply the transformations in Section~\ref{sec:comp-filter-red} to pull up and simplify the filters, then eliminate equalities from the filters as in Section~\ref{sec:lj:nat} and eliminate inner self-joins (on the primary key \texttt{rid} of relation \texttt{review}), which results in the following (by~\eqref{eq:rem-isnull}, we can remove $\neg\isNull(l^3)$ from the condition of the outermost left join):\\[6pt]
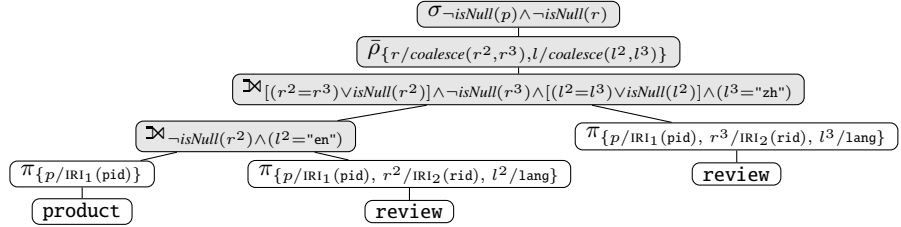
\centerline{\begin{tikzpicture}[xscale=2.9,yscale=0.5]\small
\node[nd,fill=white] (u3) at (4,2) {\texttt{review}};
\node[nd,fill=white] (u3r) at (4,3) {$\pi_{\{p/\urione(\texttt{pid}),\ r^3/\uritwo(\texttt{rid}),\ l^3/\texttt{lang}\}}$};
\node[nd,fill=white] (p2) at (2.5,1) {\texttt{review}};
\node[nd,fill=white] (p2r) at (2.5,2) {$\pi_{\{p/\urione(\texttt{pid}),\ r^2/\uritwo(\texttt{rid}),\ l^2/\texttt{lang}\}}$};
\node[nd,fill=white] (p1) at (1,1) {\texttt{product}};
\node[nd,fill=white] (p1r) at (1,2) {$\pi_{\{p/\urione(\texttt{pid})\}}$};
\node[nd] (lj1) at (1.75,3) {$\LJoin_{\neg\isNull(r^2) \land(l^2 = \texttt{"en"})}$}; 
\node[nd] (lj2) at (3,4.25) {$\LJoin_{[(r^2 = r^3) \lor \isNull(r^2)] \land \neg\isNull(r^3) \land [(l^2 = l^3) \lor\isNull(l^2)] \land(l^3 = \texttt{"zh"})}$}; 
\node[nd] (lj2r) at (3,5.25) {$\bar{\rho}_{\{r/\coalesce(r^2,r^3), l/\coalesce(l^2,l^3)\}}$}; 
\node[nd] (lj2f) at (3,6.25) {$\sigma_{\neg\isNull(p)\land \neg\isNull(r)}$}; 
\draw (u3) -- (u3r);
%
\draw (p2) -- (p2r);
%
\draw (p1) -- (p1r);
%
\draw (p1r) -- (lj1);
\draw (p2r) -- (lj1);
\draw (lj1) -- (lj2);
%
\draw (u3r) -- (lj2);
\draw (lj2) -- (lj2r);
\draw (lj2r) -- (lj2f);
\end{tikzpicture}}\\[2pt]
Observe that Proposition~\ref{prop:lj-nat} is not applicable to the innermost left join because each product can have many (in particular, many English) reviews. 

By~\eqref{eq:lj:prop:1}, $\sigma_{\isNull(l^2) \lor (l^2 = \texttt{"en"})}$ could be added above the innermost left join and subsequently lifted to the filter of the outermost left join by using~\eqref{eq:lj:filter-left}. The resulting join condition will then contain
\begin{equation*}
 [(l^2 = l^3) \lor\isNull(l^2)] \ \ \land\ \ (l^3 = \texttt{"zh"})\textcolor{blue}{{} \ \ \land \ \ [\isNull(l^2) \lor (l^2 = \texttt{"en"})]},
\end{equation*}
where the freshly lifted-up fragment is indicated in blue.
By using distributivity of $\land$ over $\lor$ and absorption ($F \lor (F \land F') \equiv^{\scriptscriptstyle+} F$), we obtain 
\begin{align*}
& [(l^2 = l^3) \lor\isNull(l^2)]  \land(l^3 = \texttt{"zh"})\land [\isNull(l^2) \lor (l^2 = \texttt{"en"})]\\
 & \ \ \equiv^{\scriptscriptstyle+} \ \bigl([(l^2 = l^3) \land \isNull(l^2)] \lor [\isNull(l^2) \land \isNull(l^2)] \lor {} \\ 
 & \hspace*{6em}[(l^2 = l^3) \land(l^2 = \texttt{"en"})] \lor [\isNull(l^2) \land  (l^2 = \texttt{"en"})]\bigr) \ \ \land \ \ (l^3 = \texttt{"zh"})\\
&  \ \ \equiv^{\scriptscriptstyle+} \ [\isNull(l^2)  \land (l^3 = \texttt{"zh"})]  \ \  \lor\ \  [(l^2 = l^3) \land(l^2 = \texttt{"en"})  \land (l^3 = \texttt{"zh"})], 
\end{align*}
which is clearly equivalent to $\isNull(l^2)  \land (l^3 =
\texttt{"zh"})$. 
Next, we can use~\eqref{eq:lj:prop:2} to pull the weakening of the filter $\neg\isNull(l^2) \land \neg\isNull(r^2)$ through the innermost left join and attach the resulting disjunction, $[\neg\isNull(l^2) \land \neg\isNull(r^2)] \lor [\isNull(l^2) \land \isNull(r^2)]$, to the outermost left join. The resulting join condition will then contain 
\begin{multline*}
[(r^2 = r^3) \lor \isNull(r^2)] \ \ \land \ \ \neg\isNull(r^3) \ \ \land \ \ \isNull(l^2) \ \  \land (l^3 =
\texttt{"zh"}) \ \ \land{}\\ \ \ \textcolor{blue}{\bigl([\neg\isNull(l^2) \land \neg\isNull(r^2)] \lor [\isNull(l^2) \land \isNull(r^2)]\bigr)},
\end{multline*} 
where the freshly lifted-up fragment is indicated in blue.
This allows us to simplify the condition of the outermost left join to:
\begin{equation*}
\isNull(r^2) \land \isNull(l^2) \land\neg\isNull(r^3) \land (l^3 = \texttt{"zh"}).
\end{equation*}
On the other hand, we can push $\neg\isNull(r)$ through the renaming with $\coalesce$ using the following equivalence 
\begin{equation}
\sigma_{\neg\isNull(v)} \rho^{\{v^1,v^2\}}_{v/\coalesce(v^1,v^2)} R \ \ \equiv \ \ \rho^{\{v^1,v^2\}}_{v/\coalesce(v^1,v^2)} \sigma_{\neg\isNull(v^1) \lor \neg\isNull(v^2)} R.
\end{equation}
We thus obtain\\[6pt] 
\centerline{\begin{tikzpicture}[xscale=2.9,yscale=0.5]\small
\node[nd,fill=white] (u3) at (4,2) {\texttt{review}};
\node[nd,fill=white] (u3r) at (4,3) {$\pi_{\{p/\urione(\texttt{pid}),\ r^3/\uritwo(\texttt{rid}),\ l^3/\texttt{lang}\}}$};
\node[nd,fill=white] (p2) at (2.5,1) {\texttt{review}};
\node[nd,fill=white] (p2r) at (2.5,2) {$\pi_{\{p/\urione(\texttt{pid}),\ r^2/\uritwo(\texttt{rid}),\ l^2/\texttt{lang}\}}$};
\node[nd,fill=white] (p1) at (1,1) {\texttt{product}};
\node[nd,fill=white] (p1r) at (1,2) {$\pi_{\{p/\urione(\texttt{pid})\}}$};
\node[nd] (lj1) at (1.75,3) {$\LJoin_{\neg\isNull(r^2) \land(l^2 = \texttt{"en"})}$}; 
\node[nd] (lj2) at (3,4.25) {$\LJoin_{\isNull(r^2) \land \isNull(l^2) \land\neg\isNull(r^3) \land (l^3 = \texttt{"zh"})}$}; 
\node[nd] (lj2ss) at (3, 5.25) {$\sigma_{\neg\isNull(r^2) \lor\neg\isNull(r^3)}$};
\node[nd] (lj2r) at (3,6.25) {$\bar{\rho}_{\{r/\coalesce(r^2,r^3), l/\coalesce(l^2,l^3)\}}$}; 
\node[nd] (lj2f) at (3,7.25) {$\sigma_{\neg\isNull(p)}$}; 
\draw (u3) -- (u3r);
%
\draw (p2) -- (p2r);
%
\draw (p1) -- (p1r);
%
\draw (p1r) -- (lj1);
\draw (p2r) -- (lj1);
\draw (lj1) -- (lj2);
%
\draw (u3r) -- (lj2);
\draw (lj2) -- (lj2ss);
\draw (lj2ss) -- (lj2r);
\draw (lj2r) -- (lj2f);
\end{tikzpicture}}\\[2pt]
Now, since $\nullify_{\{r^3,l^3\}}(\neg\isNull(r^2)
\lor\neg\isNull(r^3))$ is inconsistent with $\isNull(r^2)$ in the
join condition of the outermost left join, by using~\eqref{eq:lj:no-minus}, we
replace the outermost left join by an outer union, thus resulting in\\[6pt] 
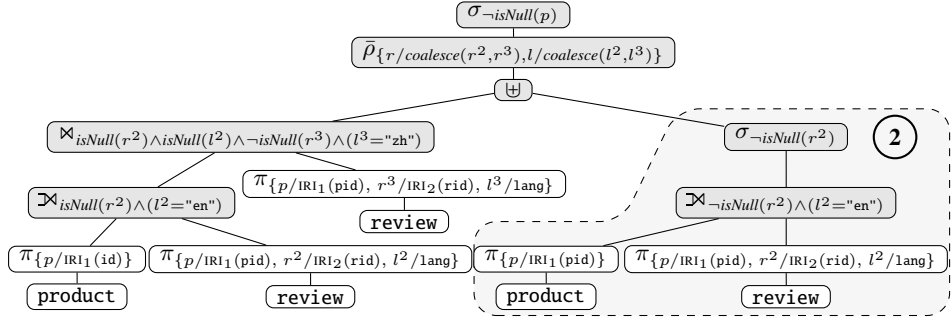
\centerline{\begin{tikzpicture}[xscale=2.9,yscale=0.5]\small
\draw[dashed, fill=gray!7,rounded corners=3mm] (3.32,1) -- (5.5,1) -- (5.5,6.5) -- (4.25,6.5) -- (4, 3.5) -- (3.32,3.5) -- cycle;
\node[circle,draw,thick] at (5.25,5.75) {\bf 2};
\node[nd,fill=white] (sp2) at (4.75,1.5) {\texttt{review}};
\node[nd,fill=white] (sp2r) at (4.75,2.5) {$\pi_{\{p/\urione(\texttt{pid}),\ r^2/\uritwo(\texttt{rid}),\ l^2/\texttt{lang}\}}$};
%
\node[nd,fill=white] (sp1) at (3.65,1.5) {\texttt{product}};
\node[nd,fill=white] (sp1r) at (3.65,2.5) {$\pi_{\{p/\urione(\texttt{pid})\}}$};
\node[nd] (slj1) at (4.75,4) {$\LJoin_{\neg\isNull(r^2) \land (l^2 = \texttt{"en"})}$}; 
\node[nd] (slj1s) at (4.75,5.75) {$\sigma_{\neg\isNull(r^2)}$}; 
\draw (sp2) -- (sp2r);
%
\draw (sp1) -- (sp1r);
%
\draw (sp1r) -- (slj1);
\draw (sp2r) -- (slj1);
\draw (slj1) -- (slj1s);
\node[nd,fill=white] (u3) at (3,3.5) {\texttt{review}};
\node[nd,fill=white] (u3r) at (3,4.5) {$\pi_{\{p/\urione(\texttt{pid}),\ r^3/\uritwo(\texttt{rid}),\ l^3/\texttt{lang}\}}$};
%
\node[nd,fill=white] (p2) at (2.57,1.5) {\texttt{review}};
\node[nd,fill=white] (p2r) at (2.57,2.5) {$\pi_{\{p/\urione(\texttt{pid}),\ r^2/\uritwo(\texttt{rid}),\ l^2/\texttt{lang}\}}$};
%
\node[nd,fill=white] (p1) at (1.5,1.5) {\texttt{product}};
\node[nd,fill=white] (p1r) at (1.5,2.5) {$\pi_{\{p/\urione(\texttt{id})\}}$};
\node[nd] (lj1) at (1.75,4) {$\LJoin_{\isNull(r^2)\land(l^2 = \texttt{"en"})}$}; 
\node[nd] (lj2) at (2.25,5.75) {$\Join_{\isNull(r^2) \land \isNull(l^2) \land\neg\isNull(r^3) \land (l^3 = \texttt{"zh"})}$}; 
\draw (u3) -- (u3r);
%
\draw (p2) -- (p2r);
%
\draw (p1) -- (p1r);
%
\draw (p1r) -- (lj1);
\draw (p2r) -- (lj1);
\draw (lj1) -- (lj2);
%
\draw (u3r) -- (lj2);
\node[nd] (u) at (3.5,7) {$\uplus$}; 
\node[nd] (ur) at (3.5,8) {$\bar{\rho}_{\{r/\coalesce(r^2,r^3), l/\coalesce(l^2,l^3)\}}$}; 
\node[nd] (uss) at (3.5,9) {$\sigma_{\neg\isNull(p)}$}; 
\draw (lj2) -- (u);
\draw (slj1s) -- (u);
\draw (u) -- (ur);
\draw (ur) -- (uss);
\end{tikzpicture}}\\[6pt]
Let us first focus on the second component of the outer union, which is depicted in the shaded area above. Since $\nullify_{\{r^2,l^2\}}(\neg \isNull(r^2))$ is false, by applying~\eqref{eq:lj:no-minus} to the left join in the area shaded in the diagram above, we replace the left join by another outer union, the second component of which, however, is trivially empty because $\nullify_{\{r^2,l^2\}}(\neg \isNull(r^2))\equiv^{\scriptscriptstyle+} \bot$. Therefore, the left join is effectively replaced by an inner join, whose filter can be pulled up, resulting in\\[6pt]
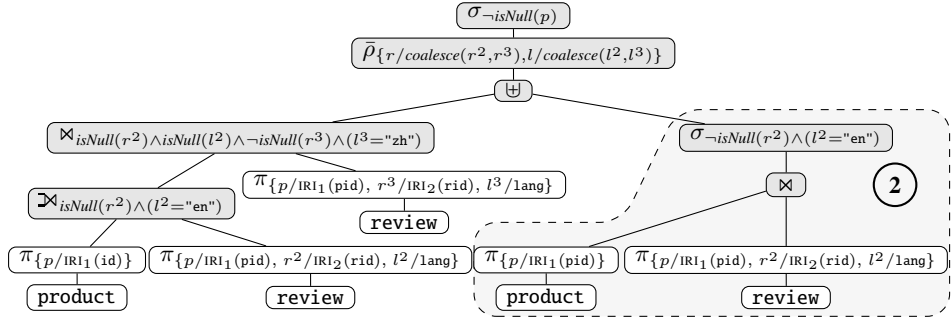
\centerline{\begin{tikzpicture}[xscale=2.9,yscale=0.5]\small
%
%
%
\draw[dashed, fill=gray!7,rounded corners=3mm] (3.32,1) -- (5.5,1) -- (5.5,6.5) -- (4.25,6.5) -- (4, 3.5) -- (3.32,3.5) -- cycle;
\node[circle,draw,thick] at (5.25,4.5) {\bf 2};
\node[nd,fill=white] (sp2) at (4.75,1.5) {\texttt{review}};
\node[nd,fill=white] (sp2r) at (4.75,2.5) {$\pi_{\{p/\urione(\texttt{pid}),\ r^2/\uritwo(\texttt{rid}),\ l^2/\texttt{lang}\}}$};
%
\node[nd,fill=white] (sp1) at (3.65,1.5) {\texttt{product}};
\node[nd,fill=white] (sp1r) at (3.65,2.5) {$\pi_{\{p/\urione(\texttt{pid})\}}$};
\node[nd] (slj1) at (4.75,4.5) {$\Join$}; 
\node[nd] (slj1s) at (4.75,5.75) {$\sigma_{\neg\isNull(r^2) \land (l^2 = \texttt{"en"})}$}; 
\draw (sp2) -- (sp2r);
%
\draw (sp1) -- (sp1r);
%
\draw (sp1r) -- (slj1);
\draw (sp2r) -- (slj1);
\draw (slj1) -- (slj1s);
\node[nd,fill=white] (u3) at (3,3.5) {\texttt{review}};
\node[nd,fill=white] (u3r) at (3,4.5) {$\pi_{\{p/\urione(\texttt{pid}),\ r^3/\uritwo(\texttt{rid}),\ l^3/\texttt{lang}\}}$};
%
\node[nd,fill=white] (p2) at (2.57,1.5) {\texttt{review}};
\node[nd,fill=white] (p2r) at (2.57,2.5) {$\pi_{\{p/\urione(\texttt{pid}),\ r^2/\uritwo(\texttt{rid}),\ l^2/\texttt{lang}\}}$};
%
\node[nd,fill=white] (p1) at (1.5,1.5) {\texttt{product}};
\node[nd,fill=white] (p1r) at (1.5,2.5) {$\pi_{\{p/\urione(\texttt{id})\}}$};
\node[nd] (lj1) at (1.75,4) {$\LJoin_{\isNull(r^2)\land(l^2 = \texttt{"en"})}$}; 
\node[nd] (lj2) at (2.25,5.75) {$\Join_{\isNull(r^2) \land \isNull(l^2) \land\neg\isNull(r^3) \land (l^3 = \texttt{"zh"})}$}; 
\draw (u3) -- (u3r);
%
\draw (p2) -- (p2r);
%
\draw (p1) -- (p1r);
%
\draw (p1r) -- (lj1);
\draw (p2r) -- (lj1);
\draw (lj1) -- (lj2);
%
\draw (u3r) -- (lj2);
\node[nd] (u) at (3.5,7) {$\uplus$}; 
\node[nd] (ur) at (3.5,8) {$\bar{\rho}_{\{r/\coalesce(r^2,r^3), l/\coalesce(l^2,l^3)\}}$}; 
\node[nd] (uss) at (3.5,9) {$\sigma_{\neg\isNull(p)}$}; 
\draw (lj2) -- (u);
\draw (slj1s) -- (u);
\draw (u) -- (ur);
\draw (ur) -- (uss);
\end{tikzpicture}}\\[6pt]
Next, the inner join can be eliminated because it is over a foreign key (\texttt{pid} in \texttt{review} references \texttt{pid} in \texttt{product}) and no attribute occurs only in the left-hand side argument of the join:\\[-6pt]
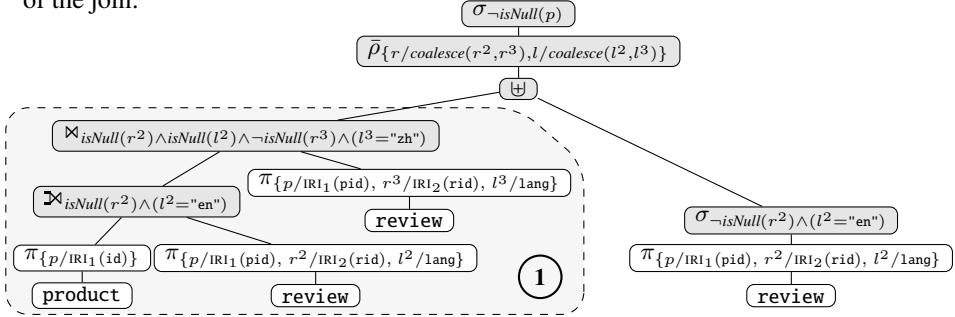
\centerline{\begin{tikzpicture}[xscale=2.9,yscale=0.5]\small
\draw[dashed, fill=gray!7,rounded corners=3mm] (3.8,1) -- (3.8,5) -- (3.5,6.5) -- (1.15,6.5) -- (1.15,1) -- cycle; 
\node[circle,draw,thick] at (3.6,2) {\bf 1};
%
%
\node[nd,fill=white] (sp2) at (4.75,1.5) {\texttt{review}};
\node[nd,fill=white] (sp2r) at (4.75,2.5) {$\pi_{\{p/\urione(\texttt{pid}),\ r^2/\uritwo(\texttt{rid}),\ l^2/\texttt{lang}\}}$};
%
%
\node[nd] (slj1s) at (4.75,3.5) {$\sigma_{\neg\isNull(r^2) \land (l^2 = \texttt{"en"})}$}; 
\draw (sp2) -- (sp2r);
%
%
\draw (sp2r) -- (slj1s);
\node[nd,fill=white] (u3) at (3,3.5) {\texttt{review}};
\node[nd,fill=white] (u3r) at (3,4.5) {$\pi_{\{p/\urione(\texttt{pid}),\ r^3/\uritwo(\texttt{rid}),\ l^3/\texttt{lang}\}}$};
%
\node[nd,fill=white] (p2) at (2.57,1.5) {\texttt{review}};
\node[nd,fill=white] (p2r) at (2.57,2.5) {$\pi_{\{p/\urione(\texttt{pid}),\ r^2/\uritwo(\texttt{rid}),\ l^2/\texttt{lang}\}}$};
%
\node[nd,fill=white] (p1) at (1.5,1.5) {\texttt{product}};
\node[nd,fill=white] (p1r) at (1.5,2.5) {$\pi_{\{p/\urione(\texttt{id})\}}$};
\node[nd] (lj1) at (1.75,4) {$\LJoin_{\isNull(r^2)\land(l^2 = \texttt{"en"})}$}; 
\node[nd] (lj2) at (2.25,5.75) {$\Join_{\isNull(r^2) \land \isNull(l^2) \land\neg\isNull(r^3) \land (l^3 = \texttt{"zh"})}$}; 
\draw (u3) -- (u3r);
%
\draw (p2) -- (p2r);
%
\draw (p1) -- (p1r);
%
\draw (p1r) -- (lj1);
\draw (p2r) -- (lj1);
\draw (lj1) -- (lj2);
%
\draw (u3r) -- (lj2);
\node[nd] (u) at (3.5,7) {$\uplus$}; 
\node[nd] (ur) at (3.5,8) {$\bar{\rho}_{\{r/\coalesce(r^2,r^3), l/\coalesce(l^2,l^3)\}}$}; 
\node[nd] (uss) at (3.5,9) {$\sigma_{\neg\isNull(p)}$}; 
\draw (lj2) -- (u);
\draw (slj1s) -- (u);
\draw (u) -- (ur);
\draw (ur) -- (uss);
\end{tikzpicture}}\\[6pt]
We now return to the first component of the outer union, which is in the shaded area in the diagram above. First, observe that the join condition of the inner join can be pulled up and the order of the left join and join changed using the well-known equivalence $(R_1 \LJoin_F R_2) \Join R_3 \equiv (R_1 \Join R_3) \LJoin_F R_2$  provided that $U_2 \cap U_3\subseteq U_1$, where $U_i$ are the attributes of $R_i$ (see, e.g., (6) in~\cite{GaRo97}). Thus, we obtain\\[6pt]
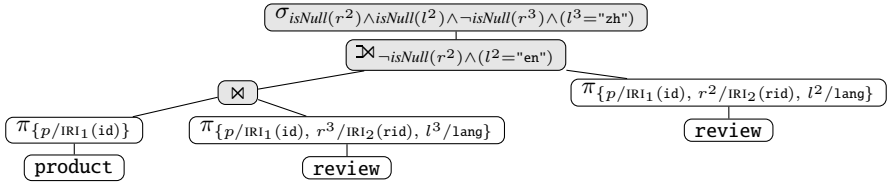
\centerline{\begin{tikzpicture}[xscale=2.9,yscale=0.5]\small
\node[nd,fill=white] (u3) at (2.5,1) {\texttt{review}};
\node[nd,fill=white] (u3r) at (2.5,2) {$\pi_{\{p/\urione(\texttt{id}),\ r^3/\uritwo(\texttt{rid}),\ l^3/\texttt{lang}\}}$};
%
\node[nd,fill=white] (p2) at (4.25,2) {\texttt{review}};
\node[nd,fill=white] (p2r) at (4.25,3) {$\pi_{\{p/\urione(\texttt{id}),\ r^2/\uritwo(\texttt{rid}),\ l^2/\texttt{lang}\}}$};
\node[nd,fill=white] (p1) at (1.25,1) {\texttt{product}};
\node[nd,fill=white] (p1r) at (1.25,2) {$\pi_{\{p/\urione(\texttt{id})\}}$};
%
\node[nd] (lj2) at (2,3) {$\Join$}; 
\node[nd] (lj1) at (3,4) {$\LJoin_{\neg\isNull(r^2) \land(l^2 = \texttt{"en"})}$}; 
\node[nd] (s) at (3,5) {$\sigma_{\isNull(r^2) \land \isNull(l^2) \land\neg\isNull(r^3) \land (l^3 = \texttt{"zh"})}$}; 
\draw (u3) -- (u3r);
%
\draw (p2) -- (p2r);
%
\draw (p1) -- (p1r);
%
\draw (p1r) -- (lj2);
\draw (u3r) -- (lj2);
\draw (lj2) -- (lj1);
\draw (p2r) -- (lj1);
\draw (lj1) -- (s);
\end{tikzpicture}}\\[6pt]
Now, the inner join can again be eliminated because it is over a foreign key (\texttt{pid}) and no attribute occurs only on the left:\\[6pt]
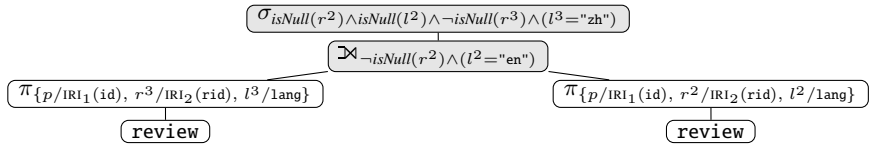
\centerline{\begin{tikzpicture}[xscale=2.9,yscale=0.5]\small
\node[nd,fill=white] (u3) at (1.75,2) {\texttt{review}};
\node[nd,fill=white] (u3r) at (1.75,3) {$\pi_{\{p/\urione(\texttt{id}),\ r^3/\uritwo(\texttt{rid}),\ l^3/\texttt{lang}\}}$};
%
\node[nd,fill=white] (p2) at (4.25,2) {\texttt{review}};
\node[nd,fill=white] (p2r) at (4.25,3) {$\pi_{\{p/\urione(\texttt{id}),\ r^2/\uritwo(\texttt{rid}),\ l^2/\texttt{lang}\}}$};
%
%
\node[nd] (lj1) at (3,4) {$\LJoin_{\neg\isNull(r^2) \land(l^2 = \texttt{"en"})}$}; 
\node[nd] (s) at (3,5) {$\sigma_{\isNull(r^2) \land \isNull(l^2) \land\neg\isNull(r^3) \land (l^3 = \texttt{"zh"})}$}; 
\draw (u3) -- (u3r);
%
\draw (p2) -- (p2r);
%
%
\draw (u3r) -- (lj1);
\draw (p2r) -- (lj1);
\draw (lj1) -- (s);
\end{tikzpicture}}\\[6pt]
Next, by~\eqref{eq:lj:filter-left:d}, we can push the $\neg\isNull(r^3) \land (l^3 = \texttt{"zh"})$ of the filter to the first component of the left join. By~\eqref{eq:lj:minus-encoding}, since the remaining part of the filter, that is, $\isNull(r^2) \land \isNull(l^2)$, is inconsistent with the left join condition $\neg\isNull(r^2) \land(l^2 = \texttt{"en"})$,  we can simplify the filter above the left join to $\isNull(r^2)$ because $r^2$ is not nullable in the right-hand side argument of the left join (and can be chosen as the $w$). Finally, by~\eqref{eq:lj:filter-right}, we push the left join condition and obtain a natural left join:\\[6pt]
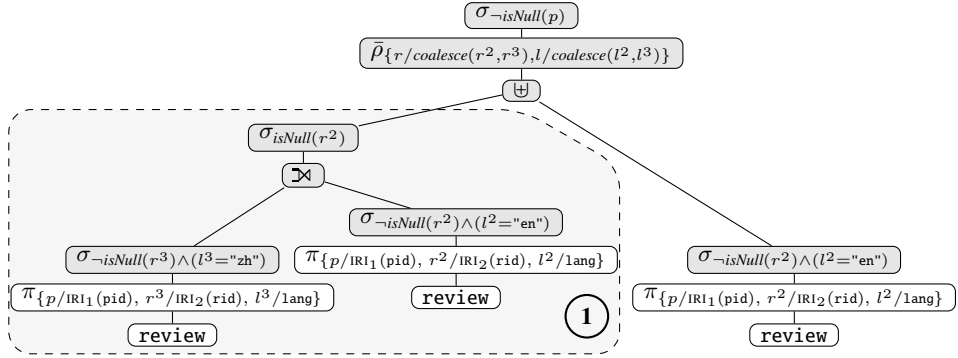
\centerline{\begin{tikzpicture}[xscale=2.9,yscale=0.5]\small
\draw[dashed, fill=gray!7,rounded corners=3mm] (3.95,1) -- (3.95,6) -- (3.5,7.5) -- (1.15,7.5) -- (1.15,1) -- cycle; 
\node[circle,draw,thick] at (3.8,2) {\bf 1};
%
%
\node[nd,fill=white] (sp2) at (4.75,1.5) {\texttt{review}};
\node[nd,fill=white] (sp2r) at (4.75,2.5) {$\pi_{\{p/\urione(\texttt{pid}),\ r^2/\uritwo(\texttt{rid}),\ l^2/\texttt{lang}\}}$};
%
%
\node[nd] (slj1s) at (4.75,3.5) {$\sigma_{\neg\isNull(r^2) \land (l^2 = \texttt{"en"})}$}; 
\draw (sp2) -- (sp2r);
%
%
\draw (sp2r) -- (slj1s);
\node[nd,fill=white] (u3) at (1.9,1.5) {\texttt{review}};
\node[nd,fill=white] (u3r) at (1.9,2.5) {$\pi_{\{p/\urione(\texttt{pid}),\ r^3/\uritwo(\texttt{rid}),\ l^3/\texttt{lang}\}}$};
\node[nd] (u3s) at (1.9,3.5)  {$\sigma_{\neg\isNull(r^3) \land (l^3 = \texttt{"zh"})}$};
\node[nd,fill=white] (p2) at (3.2,2.5) {\texttt{review}};
\node[nd,fill=white] (p2r) at (3.2,3.5) {$\pi_{\{p/\urione(\texttt{pid}),\ r^2/\uritwo(\texttt{rid}),\ l^2/\texttt{lang}\}}$};
\node[nd] (p2s) at (3.2,4.5) {$\sigma_{\neg\isNull(r^2) \land (l^2 = \texttt{"en"})}$};
%
%
\node[nd] (lj1) at (2.5,5.75) {$\LJoin$}; 
\node[nd] (lj2) at (2.5,6.75) {$\sigma_{\isNull(r^2)}$}; 
\draw (u3) -- (u3r);
\draw (u3r) -- (u3s);
\draw (p2) -- (p2r);
\draw (p2r) -- (p2s);
%
%
\draw (p2s) -- (lj1);
\draw (lj1) -- (lj2);
%
\draw (u3s) -- (lj1);
\node[nd] (u) at (3.5,8) {$\uplus$}; 
\node[nd] (ur) at (3.5,9) {$\bar{\rho}_{\{r/\coalesce(r^2,r^3), l/\coalesce(l^2,l^3)\}}$}; 
\node[nd] (uss) at (3.5,10) {$\sigma_{\neg\isNull(p)}$}; 
\draw (lj2) -- (u);
\draw (slj1s) -- (u);
\draw (u) -- (ur);
\draw (ur) -- (uss);
\end{tikzpicture}}

\bigskip

We can now push the two $\coalesce$ through $\uplus$, simplify them, push down the projections and remove unnecessary padding $\mu$ in the second argument of the union. Finally, we remove $\neg\isNull()$ for all non-nullable attributes and obtain\\[6pt] 
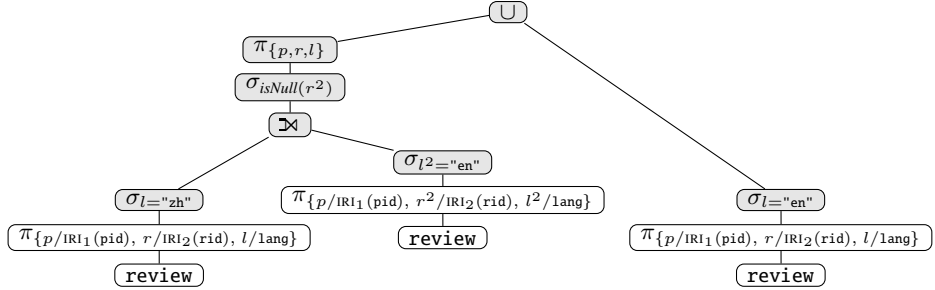
\centerline{\begin{tikzpicture}[xscale=2.9,yscale=0.5]\small
%
%
%
\node[nd,fill=white] (sp2) at (4.75,1.5) {\texttt{review}};
\node[nd,fill=white] (sp2r) at (4.75,2.5) {$\pi_{\{p/\urione(\texttt{pid}),\ r/\uritwo(\texttt{rid}),\ l/\texttt{lang}\}}$};
%
%
\node[nd] (slj1s) at (4.75,3.5) {$\sigma_{l = \texttt{"en"}}$}; 
\draw (sp2) -- (sp2r);
%
%
\draw (sp2r) -- (slj1s);
\node[nd,fill=white] (u3) at (1.9,1.5) {\texttt{review}};
\node[nd,fill=white] (u3r) at (1.9,2.5) {$\pi_{\{p/\urione(\texttt{pid}),\ r/\uritwo(\texttt{rid}),\ l/\texttt{lang}\}}$};
\node[nd] (u3s) at (1.9,3.5)  {$\sigma_{l = \texttt{"zh"}}$};
\node[nd,fill=white] (p2) at (3.2,2.5) {\texttt{review}};
\node[nd,fill=white] (p2r) at (3.2,3.5) {$\pi_{\{p/\urione(\texttt{pid}),\ r^2/\uritwo(\texttt{rid}),\ l^2/\texttt{lang}\}}$};
\node[nd] (p2s) at (3.2,4.5) {$\sigma_{l^2 = \texttt{"en"}}$};
%
%
\node[nd] (lj1) at (2.5,5.5) {$\LJoin$}; 
\node[nd] (lj2) at (2.5,6.5) {$\sigma_{\isNull(r^2)}$}; 
\node[nd] (lj2p) at (2.5,7.5) {$\pi_{\{p, r, l\}}$}; 
\draw (u3) -- (u3r);
\draw (u3r) -- (u3s);
\draw (p2) -- (p2r);
\draw (p2r) -- (p2s);
%
%
\draw (p2s) -- (lj1);
\draw (lj1) -- (lj2);
\draw (lj2) -- (lj2p);
%
\draw (u3s) -- (lj1);
\node[nd] (u) at (3.5,8.5) {$\cup$}; 
%
\draw (lj2p) -- (u);
\draw (slj1s) -- (u);
\end{tikzpicture}}\\[3pt]
which corresponds to the following SQL:
\begin{lstlisting}[language=SQLT]
SELECT CONCAT("IRI1", pid) AS p, CONCAT("IRI2", rid) AS r, lang AS l
FROM review WHERE l = "en"
UNION ALL
SELECT CONCAT("IRI1", r1.pid) AS p, 
       CONCAT("IRI2", r1.rid) AS r, r1.lang AS l
FROM review r1 LEFT JOIN review r2 
                         ON (r1.pid = r2.pid) AND r2.lang = "en"
WHERE r1.lang = "zh" AND r2.lang IS NULL
\end{lstlisting}
where the \texttt{CONCAT} functions construct IRIs. The second component of the union can also be expressed in SQL using \texttt{NOT IN}:
\begin{lstlisting}[language=SQLT]
SELECT CONCAT("IRI1", pid) AS p, CONCAT("IRI2", rid) AS r, lang AS l
FROM review  
WHERE lang = "zh" AND
      pid NOT IN (SELECT pid FROM review WHERE lang = "en")
\end{lstlisting}


\begin{thebibliography}{10}

\bibitem{AFMD11}
M.~Arias, J.~D. Fern{\'a}ndez, M.~A. Mart{\'\i}nez-Prieto, and P.~de~la Fuente.
\newblock An empirical study of real-world {SPARQL} queries.
\newblock In {\em Proc.\ USEWOD}, 2011.

\bibitem{atre15}
M.~Atre.
\newblock Left Bit Right: For {SPARQL} join queries with {OPTIONAL}
  patterns (left-outer-joins).
\newblock In {\em Proc.\ ACM SIGMOD}, pages 1793--1808, 2015.

\bibitem{BiSc09}
C.~Bizer and A.~Schultz.
\newblock The {B}erlin {SPARQL} benchmark.
\newblock {\em Int.\ J.\ on Semantic Web \& Information Systems}, 5(2):1--24,
  2009.

\bibitem{CCKK*17}
D.~Calvanese, B.~Cogrel, S.~Komla-Ebri, R.~Kontchakov, D.~Lanti, M.~Rezk,
  M.~Rodriguez-Muro, G.~Xiao.
\newblock {Ontop}: answering {SPARQL} queries over relational databases.
{\em SWJ}, 2017.

\bibitem{CDLLR07}
D.~Calvanese, G.~De~Giacomo, D.~Lembo, M.~Lenzerini, and R.~Rosati.
\newblock Tractable reasoning and efficient query answering in description
  logics: The \textit{DL-Lite} family.
\newblock {\em JAR}, 2007.

\bibitem{ChGM90}
U.~S. Chakravarthy, J.~Grant, and J.~Minker.
\newblock Logic-based approach to semantic query optimization.
\newblock {\em ACM TODS}, 15(2):162--207, 1990.

\bibitem{DBLP:conf/kesw/ChaloupkaN16}
M.~Chaloupka and M.~Necask{\'{y}}.
\newblock Efficient {SPARQL} to {SQL} translation with user defined mapping.
\newblock In {\em Proc.\ KESW}, volume 649 of {\em CCIS}, pages 215--229. Springer, 2016.

\bibitem{ChLF09}
A.~Chebotko, S.~Lu, and F.~Fotouhi.
\newblock Semantics preserving {SPARQL-to-SQL} translation.
\newblock {\em DKE}, 68(10):973--1000, 2009.

\bibitem{Cyga05}
R.~Cyganiak.
\newblock A relational algerba for {SPARQL}.
\newblock TR HPL-2005-170, HP Labs Bristol, 2005.

\bibitem{DLLPR18}
G.~De~Giacomo, D.~Lembo, M.~Lenzerini, A.~Poggi, and R.~Rosati.
\newblock Using ontologies for semantic data integration.
\newblock In {\em A Comprehensive Guide through the {Italian} Database Research
  over the Last 25 Years}, vol.~31 of {\em Studies in Big Data}. Springer,
  2018.

\bibitem{Elmasri:2010uk}
R.~Elmasri and S.~Navathe.
\newblock {\em Fundamentals of Database Systems}.
\newblock Addison-Wesley, 2010.

\bibitem{GaRo97}
C.~Galindo-Legaria and A.~Rosenthal.
\newblock Outerjoin simplification and reordering for query optimization.
\newblock {\em ACM TODS}, 22(1):43--74, 1997.

\bibitem{GL17}
P.~Guagliardo and L.~Libkin.
\newblock A formal semantics of {SQL} queries, its validation, and
  applications.
\newblock {\em PVLDB}, 11(1):27--39, 2017.

\bibitem{DBLP:conf/icdt/KaminskiK16}
M.~Kaminski and E.~V. Kostylev.
\newblock Beyond well-designed {SPARQL}.
\newblock In {\em Proc.\ ICDT}, 
  2016.

\bibitem{KRRXZ14}
R.~Kontchakov, M.~Rezk, M.~Rodriguez-Muro, G.~Xiao, and M.~Zakharyaschev.
\newblock Answering {SPARQL} queries over databases under {OWL~2~QL} entailment
  regime.
\newblock In {\em Proc.\ ISWC}, 
  2014.

\bibitem{PeAG09}
J.~P{\'e}rez, M.~Arenas, and C.~Gutierrez.
\newblock Semantics and complexity of {SPARQL}.
\newblock {\em ACM TODS}, 34(3), 2009.

\bibitem{PiVa11}
F.~Picalausa and S.~Vansummeren.
\newblock What are real {SPARQL} queries like?
\newblock In {\em SWIM}, 
2011.

\bibitem{PLCD*08}
A.~Poggi, D.~Lembo, D.~Calvanese, G.~De~Giacomo, M.~Lenzerini, and R.~Rosati.
\newblock Linking data to ontologies.
\newblock {\em J.\ Data Semantics}, 10:133--173, 2008.

\bibitem{PrCS14}
F.~Priyatna, {\'O}.~Corcho, and J.~F. Sequeda.
\newblock Formalisation and experiences of {R2RML}-based {SPARQL} to {SQL}
  query translation using morph.
\newblock In {\em Proc.\ WWW}, pages 479--490, 2014.

\bibitem{RaoPZ04}
J.~Rao, H.~Pirahesh, and C.~Zuzarte.
\newblock Canonical abstraction for outerjoin optimization.
\newblock In {\em Proc.\ ACM SIGMOD}, pages 671--682, 2004.

\bibitem{RoKZ13}
M.~Rodriguez-Muro, R.~Kontchakov, and M.~Zakharyaschev.
\newblock Ontology-based data access: Ontop of databases.
\newblock In {\em Proc.\ ISWC}, pages 558--573. Springer, 2013.

\bibitem{DBLP:journals/ws/Rodriguez-MuroR15}
M.~Rodriguez{-}Muro and M.~Rezk.
\newblock Efficient SPARQL-to-SQL with {R2RML} mappings.
\newblock {\em J. Web Sem.}, 33:141--169, 2015.

\bibitem{ScML10}
M.~Schmidt, M.~Meier, and G.~Lausen.
\newblock Foundations of {SPARQL} query optimization.
\newblock In {\em Proc.\ ICDT}, pages 4--33, 2010.

\bibitem{SeAM14}
J.~F. Sequeda, M.~Arenas, and D.~P. Miranker.
\newblock {OBDA}: {Query} rewriting or materialization? {In} practice, both!
\newblock In {\em Proc.\ ISWC}, volume 8796 of {\em LNCS}, pages 535--551.
  Springer, 2014.

\end{thebibliography}
\end{document}